\documentclass[letterpaper,11pt]{article}

\usepackage{fontawesome}

\usepackage{xcolor}
\usepackage[numbers]{natbib}
\usepackage{url}
\usepackage{xspace}

\newcommand{\infseqn}[1]{(#1)_{n \in \NN}}
\newcommand{\infseqnz}[1]{(#1)_{n \in \NN_0}}
\def\ddefloop#1{
  \ifx
    \ddefloop#1
  \else\ddef{#1}
    \expandafter
    \ddefloop
  \fi}
\def\ddef#1{\expandafter\def\csname #1cal\endcsname{\ensuremath{\mathcal{#1}}}}

\ddefloop
ABCDEFGHIJKLMNOPQRSTUVWXYZ
\ddefloop

\def\ddefloop#1{
  \ifx
    \ddefloop#1
  \else\ddef{#1}
    \expandafter
    \ddefloop
  \fi}
\def\ddef#1{\expandafter\def\csname #1#1\endcsname{\ensuremath{\mathbb{#1}}}}

\ddefloop
ABCDEFGHIJKLMNOPQRSTUVWXYZ
\ddefloop

\newcommand{\tth}{\text{th}}
\newcommand{\dd}{\mathrm{d}}

\newcommand{\1}{\mathds{1}}

\renewcommand{\P}{\mathsf P}
\newcommand{\Q}{\mathsf Q}
\newcommand{\nQ}{{n, \Q}}
\newcommand{\tQ}{{t, \Q}}
\newcommand{\jQ}{{j, \Q}}
\newcommand{\tauQ}{{\tau, \Q}}
\newcommand{\Pin}{{\P \in \Pcal}}
\newcommand{\Qin}{{\Q \in \Qcal}}

\newcommand{\Pe}{\Pcal\text{-}e}

\newcommand{\brackalpha}{{(\alpha)}}

\newcommand{\brackdelta}{{(\delta)}}

\newcommand{\CO}{{\mathrm{CO96}}}
\newcommand{\nto}{{n \to \infty}}
\newcommand{\alphato}{{\alpha \to 0^+}}

\usepackage{diagbox}
\usepackage{makecell}
\usepackage{booktabs}
\usepackage{tikz,tikz-3dplot}
\usepackage{amsmath}
\usepackage{mathtools}
\usepackage{bbm}
\usepackage{amsfonts}
\usepackage{amssymb}
\usepackage{amsthm}
\usepackage{thmtools}
\usepackage{url,ifthen}
\usepackage[shortlabels]{enumitem}
\usepackage{srcltx}
\usepackage{dsfont}
\usepackage{multirow}
\usepackage{boxedminipage}
\usepackage[margin=1.1in]{geometry}
\usepackage{nicefrac}
\usepackage{xspace}
\usepackage{graphicx}
\usepackage{color}
\usepackage{colortbl}
\usepackage{setspace}
\usepackage{pgfplots}
\pgfplotsset{compat=1.12}
\usepackage{algorithm,algpseudocode}
\usepackage[algo2e]{algorithm2e}
\RestyleAlgo{ruled}
\usepackage{cancel}
\usepackage{framed}
\usepackage{tcolorbox}

\usepackage{xcolor}
\definecolor{DarkGreen}{rgb}{0.1,0.5,0.1}
\definecolor{DarkRed}{rgb}{0.5,0.1,0.1}
\definecolor{DarkBlue}{rgb}{0.1,0.1,0.5}
\definecolor{LightBlue}{RGB}{69, 123, 157}
\definecolor{Gray}{rgb}{0.2,0.2,0.2}
\definecolor{AquaGreen}{RGB}{20, 116, 111}
\definecolor{Greenish}{RGB}{132, 169, 140}

\definecolor{c1}{RGB}{38, 70, 83}
\definecolor{c2}{RGB}{42, 157, 143}
\definecolor{c3}{RGB}{233, 196, 106}
\definecolor{c5}{RGB}{231, 111, 81}
\definecolor{c4}{RGB}{244, 162, 97}

\usepackage{listings}
\lstdefinestyle{mystyle}{
    commentstyle=\color{DarkBlue},
    keywordstyle=\color{DarkRed},
    numberstyle=\tiny\color{Gray},
    stringstyle=\color{Greenish},
    basicstyle=\footnotesize,
    breakatwhitespace=false,         
    breaklines=true,                 
    captionpos=b,                    
    keepspaces=true,                 
    numbers=left,                    
    numbersep=5pt,                  
    showspaces=false,                
    showstringspaces=false,
    showtabs=false,                  
    tabsize=2
}
\usepackage[small]{caption}
\usepackage{subcaption}
\usepackage[pdftex]{hyperref}
\usepackage[capitalise,noabbrev]{cleveref}
\crefname{assumption}{Assumption}{Assumptions}
\usepackage{autonum} %
\hypersetup{
    unicode=false,          %
    pdftoolbar=true,        %
    pdfmenubar=true,        %
    pdffitwindow=false,      %
    pdfnewwindow=true,      %
    colorlinks=true,       %
    linkcolor=AquaGreen,          %
    citecolor=Greenish,        %
    filecolor=DarkRed,      %
    urlcolor=DarkBlue,          %
    pdftitle={},
    pdfauthor={},
}

\def\draft{1}

\def\submit{0}

\ifnum\draft=1 %
    \def\ShowAuthNotes{1}
\else
    \def\ShowAuthNotes{0}
\fi

\ifnum\submit=1
\newcommand{\forsubmit}[1]{#1}
\newcommand{\forreals}[1]{}
\else
\newcommand{\forreals}[1]{#1}
\newcommand{\forsubmit}[1]{}
\fi

\ifnum\ShowAuthNotes=1
\newcommand{\authnote}[2]{{ \footnotesize \bf{\color{DarkRed}[#1's Note:
{\color{DarkBlue}#2}]}}}

\newcommand{\rjsnote}[1]{{\color{purple}[Ricardo: #1]}}

\newcommand{\iwsnote}[1]{{\color{blue}[Ian: #1]}}
\newcommand{\todo}[1]{{\color{red}[To Do: #1]}}

\else
\newcommand{\authnote}[2]{}
\newcommand{\rjsnote}[1]{}
\newcommand{\iwsnote}[1]{}
\newcommand{\mjnote}[1]{}
\newcommand{\todo}[1]{}
\fi

\newtheorem{theorem}{Theorem}[section]
\newtheorem{remark}[theorem]{Remark}
\newtheorem{lemma}[theorem]{Lemma}
\newtheorem{corollary}[theorem]{Corollary}
\newtheorem{proposition}[theorem]{Proposition}
\newtheorem{fact}[theorem]{Fact}

\newtheorem{assumption}{Assumption}

\newtheorem{definition}[theorem]{Definition}%

\newtheorem*{definition*}{Definition}
\newtheorem*{proposition*}{Proposition}

\newtheorem{example*}{Example}
\newtheorem{example}[theorem]{Example}

\newtheoremstyle{example_contd}
{\topsep} {\topsep}%
{}%
{}%
{\bfseries}%
{.}%
{1em}%
{\thmname{#1} \thmnumber{ #2}\thmnote{#3} (continued)}%

\theoremstyle{example_contd}

\newcommand{\chapterref}[1]{\hyperref[ch:#1]{Chapter~\ref{ch:#1}}}

\newcommand{\claimref}[1]{\hyperref[claim:#1]{Claim~\ref{claim:#1}}}

\newcommand{\corollaryref}[1]{\hyperref[cor:#1]{Corollary~\ref{cor:#1}}}

\newcommand{\definitionref}[1]{\hyperref[def:#1]{Definition~\ref{def:#1}}}

\newcommand{\equationref}[1]{\hyperref[eq:#1]{Equation~\ref{eq:#1}}}

\newcommand{\factref}[1]{\hyperref[fact:#1]{Fact~\ref{fact:#1}}}

\newcommand{\figureref}[1]{\hyperref[fig:#1]{Figure~\ref{fig:#1}}}

\newcommand{\tableref}[1]{\hyperref[tab:#1]{Table~\ref{tab:#1}}}

\newcommand{\itemref}[1]{\hyperref[item:#1]{Item~(\ref{item:#1})}}

\newcommand{\lemmaref}[1]{\hyperref[lem:#1]{Lemma~\ref{lem:#1}}}

\newcommand{\propref}[1]{\hyperref[prop:#1]{Proposition~\ref{prop:#1}}}

\newcommand{\propositionref}[1]{\hyperref[prop:#1]{Proposition~\ref{prop:#1}}}

\newcommand{\remarkref}[1]{\hyperref[rem:#1]{Remark~\ref{rem:#1}}}

\newcommand{\sectionref}[1]{\hyperref[sec:#1]{Section~\ref{sec:#1}}}

\newcommand{\theoremref}[1]{\hyperref[thm:#1]{Theorem~\ref{thm:#1}}}

\newcommand{\assumptionref}[1]{\hyperref[ass:#1]{Assumption~\ref{ass:#1}}}

\newcommand{\exampleref}[1]{\hyperref[exmp:#1]{Example~\ref{exmp:#1}}}

\newcommand{\algoref}[1]{\hyperref[algo:#1]{Algorithm~\ref{algo:#1}}}

\usepackage[utf8]{inputenc}
\usepackage[T1]{fontenc}
\usepackage{microtype}
\usepackage[ttdefault=true]{AnonymousPro}

\newcommand{\Esymb}{\mathbb{E}}

\newcommand{\E}{\Esymb}

\newcommand{\Ex}[1]{\E\left[#1\right]}

\newcommand{\widebar}[1]{\overline{#1}}

\newcommand{\flatfrac}[2]{#1/#2}

\newcommand{\mper}{\,.}
\newcommand{\mcom}{\,,}

\newcommand{\cA}{{\cal A}}

\newcommand{\cF}{{\cal F}}

\newcommand{\cH}{{\cal H}}

\newcommand{\cO}{{\cal O}}
\newcommand{\cP}{{\cal P}}
\newcommand{\cQ}{{\cal Q}}
\newcommand{\cR}{{\cal R}}

\newcommand{\cW}{{\cal W}}

\newcommand{\paren}[1]{(#1 )}
\newcommand{\Paren}[1]{\left(#1 \right )}

\newcommand{\Brac}[1]{\left[#1 \right]}

\newcommand{\Set}[1]{\left\{#1\right\}}

\newcommand{\Abs}[1]{\left\lvert#1\right\rvert}

\newcommand{\R}{\mathbb{R}}

\newcommand{\N}{\mathbb N}

\usepackage{bm}

\newcommand{\trans}{\top}

\newcommand{\Ind}[1]{\mathds{1}\Set{#1}}

\newcommand{\independent}{\protect\mathpalette{\protect\independenT}{\perp}}
\def\independenT#1#2{\mathrel{\rlap{$#1#2$}\mkern2mu{#1#2}}}

\newcommand{\ignore}[1]{}

\DeclareMathOperator*{\argmax}{arg\,max}

\renewcommand{\epsilon}{\varepsilon}

\newcommand{\eps}{\epsilon}

\newcommand{\remove}[1]{}

\newcommand{\ate}{\psi}
\newcommand{\hate}{\widehat{\ate}}
\newcommand{\testate}{\delta}

\newcommand{\svert}{~\vert~}
\newcommand{\smvert}{~\middle\vert~}
\newcommand{\iid}{{\mathrm{i.i.d}}}
\newcommand{\iidtext}{i.i.d.}

\newcommand{\logp}[1]{\log\Paren{#1}}
\newcommand{\logs}[1]{\log\Set{#1}}

\newcommand{\bracka}{{(a)}}

\newcommand{\brackzero}{{(0)}}
\newcommand{\brackd}{{(d)}}

\newcommand{\UCB}{\mathrm{UCB}}
\newcommand{\LCB}{\mathrm{LCB}}

\newcommand{\exps}[1]{\exp\Set{#1}}
\newcommand{\ceil}[1]{\left\lceil#1\right\rceil}

\newcommand{\brackUCBastar}{{(\text{UCB-$\optarm$})}}
\newcommand{\brackLCBa}{{(\text{LCB-$a$})}}
\newcommand{\brackDeltaa}{{(\Delta_a)}}

\newcommand{\bE}{\mathbf{E}}
\newcommand{\bet}{\lambda}
\newcommand{\bbet}{\boldsymbol{\lambda}}

\newcommand{\optbbet}{\bbet_\Q}
\newcommand{\optarm}{a_{\Q}}

\newcommand{\SPRUCEtext}{\texttt{SPRUCE}}
\newcommand{\SPRUCE}{\hyperref[algorithm:spruce]{\texttt{SPRUCE}}}

\newcommand{\option}{\textsc{option}}

\newcommand{\MAB}{\mathrm{Alloc}}
\newcommand{\Port}{\mathrm{Port}}

\newcommand{\UP}{\mathrm{UP}}

\newcommand{\Bern}{\mathrm{Bern}}

\newcommand{\tW}{\widetilde{W}}
\newcommand{\tbbet}{\tilde{\bbet}}
\newcommand{\tA}{\widetilde{A}}

\newcommand{\barW}{\widebar{W}}

\newcommand{\Yobs}{Y^{\mathrm{obs}}}
\newcommand{\Prct}{\P_{\mathrm{RCT}}}
\newcommand{\Qrct}{\Q_{\mathrm{RCT}}}

\newcommand{\simplex}{\triangle_{d}}

\renewcommand{\emptyset}{\varnothing}

\algnewcommand{\Inputs}[1]{%
  \State \textbf{Inputs:}
  \Statex \hspace*{\algorithmicindent}\parbox[t]{.8\linewidth}{\raggedright #1}
}

\usepackage{authblk}

\title{Multi-Armed Sequential Hypothesis Testing by Betting}
\author{
  Ricardo J.~Sandoval$^{*\dagger}$,
  Ian Waudby-Smith$^{*\dagger}$, and Michael I.~Jordan$^{\dagger\ddagger}$\\
  $^\dagger$University of California, Berkeley\\
  $^\ddagger$\'Ecole Normale Sup\'erieure \& Inria, Paris
}
 
\date{\today}

\begin{document}

\maketitle
\def\thefootnote{*}\footnotetext{Equal contribution}
\def\thefootnote{\arabic{footnote}}

\begin{abstract}
  We consider a variant of sequential testing by betting where, at each time
step, the statistician is presented with multiple data sources (arms) and obtains data by choosing one of the arms. We consider the composite global null hypothesis $\Pcal$ that
\emph{all} arms are null in a certain sense (e.g.~all dosages of a treatment are
ineffective) and we are interested in rejecting $\Pcal$ in favor of a composite
alternative $\Qcal$ where \emph{at least one} arm is non-null (e.g.~there exists
an effective treatment dosage). We posit an optimality desideratum that we
describe informally as follows: even if several arms are non-null, we seek
$e$-processes and sequential tests  whose performance are as strong as the ones
that have oracle knowledge about which arm generates the most evidence against
$\Pcal$. Formally, we generalize notions of log-optimality and expected
rejection time optimality to more than one arm, obtaining matching lower and
upper bounds for both. A key technical device in this optimality analysis is a modified upper-confidence-bound-like algorithm for unobservable but sufficiently ``estimable'' rewards. In the design of this algorithm, we derive nonasymptotic concentration inequalities for optimal wealth growth rates in the sense of Kelly [1956]. These may be of independent interest.

\end{abstract}

\section{Introduction}\label{section:introduction}

Consider a statistical hypothesis testing setting where it is possible to
sequentially sample data from one of many arms, and it is of interest to quickly
gather evidence that at least one of those arms is non-null in some way. As an
example, suppose that a pharmaceutical company is interested in testing
whether a new treatment is effective, but there are many variations of that
treatment that could be administered---such as taking a drug at a certain dose,
at a certain time of day, in conjunction with or without an exercise routine,
and so on. In such cases, the goal may be to determine if \emph{any} combination
of interventions alongside the new treatment is significantly more effective
than that same combination but \emph{without} the treatment. Readers familiar with the
multiple testing literature will recognize the previous description as an
instantiation of ``global null'' hypothesis testing and we elaborate on a formal
description thereof in \cref{section:prelim-multi-armed}.
The setup we are considering here poses additional statistical challenges
due to the presence of a sequential data collection protocol where \emph{partial
information} arises. That is, when the statistician decides to collect data
from a particular arm (such as administering a drug at a particular dose at a
certain time of day), the counterfactual outcomes for that unit under any other
arm (e.g., with a different dose at another time of day) are never observed. Additionally, despite the partial information that arises, we will aim to develop test statistics
that enjoy the same asymptotic behavior as one that is constructed using \emph{a priori} knowledge of which arm
generates the most evidence against the global null. As such, this paper finds itself at
the junction of sequential hypothesis testing, multi-armed
bandits, and causal inference. We review the necessary background on these areas
as needed. The paper proceeds as follows:

\begin{enumerate}
\item In \cref{section:prelim-testing by betting}, we review (classical,
    single-arm) sequential hypothesis testing and testing by betting, as well as
    some optimality criteria used to evaluate the resulting test statistics.
    In \cref{section:prelim-multi-armed}, we formally outline the multi-arm data
    collection protocol that will be used throughout the
    paper, and describe the family of test
    statistics and sequential tests that will be studied in this paper. This
    family of test statistics reduces to at least three special cases that are
    commonly studied in the testing by betting literature when the arm set is
    taken to be a singleton.
\item \cref{section:multi-armed} contains a definition that generalizes the
    notion of optimality reviewed in \cref{section:prelim-testing by betting} to
    the multi-armed case. We also present the main algorithm of this work termed
    ``Sublinear Portfolio Regret Upper Confidence Estimation'' (\SPRUCEtext{}) in
    \cref{algorithm:spruce} and state its oracle-like log-optimality properties
    in \cref{theorem:log-optimality of spruce}.
\item In \cref{section:rejection times}, we provide a multi-armed generalization of a qualitatively different (but mathematically related) criterion of
    optimality that considers
    the expected value of the first time at which a sequential test rejects the
    null hypothesis. We provide a lower bound on the expected rejection time
    and demonstrate that \SPRUCE{} yields an expected rejection time matching
    this lower bound with exact constants in the high-confidence regime (i.e., as
    the type-I error $\alpha \to 0^+$). 
\item In \cref{section:ingredients}, we highlight that \SPRUCE{} has
    several nontrivial differences from typical upper confidence bound-type
    algorithms. Indeed, the setting considered in this paper forces a departure
    from the typical assumptions made in the multi-armed bandit literature.
    As a result, we develop some bespoke concentration inequalities that enable
    us to prove that \SPRUCE{} has the same qualitative behavior as classical
    upper
    confidence bound algorithms, namely having logarithmic allocation regret. The focus of
    \cref{section:ingredients} is to outline some of these constructions as results that may be of independent interest.
\item We instantiate a motivating example in \cref{section:treatment effects} 
    and show how one can use \SPRUCE~to test for the existence of a positive
    average treatment effect. As a corollary to the results from the preceding
    sections we state the log-optimality of \SPRUCE~in this setting and illustrate this result empirically through simulations.
\end{enumerate}

\paragraph{Notation.} Throughout, we fix a filtered measurable space $(\Omega, \Fcal)$ so that for a probability measure $\P$, the triplet $(\Omega, \Fcal, \P)$ is a filtered probability space. For a stochastic process $\infseqn{W_n}$ on $(\Omega,\Fcal)$, we will often write $W$ as shorthand. For a probability measure $\P$, we write the expectation of a random variable $X_1$ as
\begin{equation}
    \EE_\P [X_1] = \int x \dd \P(X \leq x)\mcom
\end{equation}
and we abuse notation slightly and write $\PP_\P(B) \coloneqq \EE_\P[\1_B]$ for the probability of an event $B \in \Fcal$ so that it reads as ``the probability of $B$ under $\P$''.

\section{Preliminaries}

In this section, we review some key aspects of testing by betting (\cref{section:prelim-testing by betting}), log-optimality of those tests (\cref{section:prelim-logopt}), the multi-armed data collection protocol that we will employ throughout the paper (\cref{section:prelim-multi-armed}), and global null hypothesis testing (\cref{section:prelim-global null}). Finally, in \cref{section:prelim-test supermartingales and e-processes} we present the class of test statistics that are our focus.

\subsection{Sequential hypothesis testing by betting}\label{section:prelim-testing by betting}

Testing by betting is a research program whose roots lie in the early work of Abraham Wald 
\citep{wald1945sequential,wald1947sequential} as well as Herbert Robbins and colleagues 
\citep{darling1967confidence,robbins1968iterated,robbins1970statistical,robbins1974expected}, 
and which has received considerable attention in the past few decades. For overviews, see the review paper of \citet{ramdas2023game}, the books of
\citet{shafer2005probability,shafer2019game} and \citet{ramdas2025hypothesis}, and other related papers \citep{shafer2011test,vovk2021values,shafer2021testing,waudby2024estimating,grunwald2019safe}.
In what follows we briefly review testing by betting with a particular focus on the aspects
that are directly relevant to the present paper.

Let $\infseqn{Y_n}$ be a sequence of \iidtext{} random variables on the filtered probability space $(\Omega, \Fcal)$ where $\Fcal
\equiv \infseqnz{\Fcal_n}$ is the filtration generated by $\infseqn{Y_n}$, with
$\Fcal_0 = \{\Omega,\emptyset\}$ being the trivial sigma-algebra. Suppose 
that it is of interest to test some null hypothesis---represented by a collection 
of distributions $\Pcal$---against an alternative hypothesis $\Qcal$ such that 
$\Pcal \cap \Qcal = \varnothing$. In testing by betting, one aims to 
fix a desired type-I error rate $\alpha \in (0, 1)$ and construct a
binary-valued $\Fcal$-adapted map $\infseqn{\phi_n^\brackalpha}$ referred to as
the ``level-$\alpha$ sequential test,'' where $\phi_n^\brackalpha = 1$ should be
interpreted as ``reject'' and $\phi_n^\brackalpha = 0$ as ``do not reject.''
The key guarantee that a level-$\alpha$ \emph{sequential} test must satisfy is
type-I error control at stopping times, formally meaning that
\begin{equation}\label{eq:type-I-err-control}
    \sup_\Pin \PP_{\P}\left ( \exists n \in \NN : \phi_n^\brackalpha = 1 \right ) \leq \alpha 
    \quad\text{or equivalently}\quad \sup_\Pin \PP_{\P} \left ( \phi_\tau^\brackalpha = 1 \right ) \leq \alpha,
\end{equation}
where $\tau$ is an arbitrary $\Fcal$-stopping time, meaning that $\{ \tau = n \} \in \Fcal_n$ for any $n \in \NN$.
One way to obtain such a sequential test is to devise so-called \emph{test
supermartingales} which are defined formally as follows.

\begin{definition}[Test supermartingales]\label{definition:test supermartingale}
  Let $\infseqn{\barW_n}$ be an
  $\Fcal$-adapted process, meaning that $\barW_n$ is $\Fcal_{n}$-measurable for each
  $n \in \NN$. For a given $\Pin$, we say that $\infseqn{\barW_n}$ is a \emph{test
  $\P$-supermartingale} if
  \begin{enumerate}
  \item $\infseqn{\barW_n}$ is a $\P$-supermartingale meaning that $\EE_\P[\barW_n \mid
      \Fcal_{n-1}] \leq \barW_{n-1}$ $\P$-almost surely,
  \item for every $n \in \NN$, $\barW_n \geq 0$ $\P$-almost surely, and
    \item $\EE_\P[\barW_1] \leq 1$.
  \end{enumerate}
  We say that $\infseqn{\barW_n}$ is a test $\Pcal$-supermartingale if it is a 
  test $\P$-supermartingale for every $\Pin$.
\end{definition}
Sequential tests can similarly be obtained through $e$-processes which should be
viewed as generalizations of test supermartingales for the purposes of
sequential testing since they satisfy a key inequality due to
\citet{ville1939etude} which will be discussed shortly. 

\begin{definition}[$e$-processes]
  We say that a process $\infseqn{W_n}$ is a $\Pe$-\emph{process} if it is
  nonnegative, adapted to the filtration $\Fcal$, and is a $\P$-almost sure
  lower bound for a test $\P$-supermartingale for each $\Pin$. That is, for
  every $n \in \NN$ and $\Pin$, there exists a test $\P$-supermartingale
  $\infseqn{\barW_n^{(\P)}}$ such that 
  \begin{equation}
    \PP_{\P} \left ( W_n \leq \barW_n^{(\P)} \right ) = 1. 
  \end{equation}
\end{definition}

\begin{remark}\label{remark:eproc composite}
  The superscript in $\barW_n^{(\P)}$ is used to emphasize that the
  upper-bounding test $\P$-supermartingale could in principle be different for each
  $\P\in \Pcal$, and $W_n$ would still form a well-defined
  $\Pe$-process. See \citet{wasserman2020universal} and
  \citet{ramdas2022testing} for explicit examples of such $\Pe$-processes for
  composite nulls. However, in the present paper, the $\Pe$-processes we
  construct will be upper-bounded by the same test $\cP$-supermartingale for
  each $\Pin$, for which reason we suppress the use of the superscript going forward.
\end{remark}
Clearly all test $\Pcal$-supermartingales are $\Pe$-processes.
The fact that \eqref{eq:type-I-err-control} can be achieved once provided access
to a $\Pe$-process follows immediately from Ville's inequality for
nonnegative supermartingales \citep{ville1939etude}. Specifically, we have that for
any $\alpha \in (0, 1)$ and process $\infseqn{\phi_n^\brackalpha}$ given by
$\phi_n^\brackalpha \coloneqq \1 \{ W_n \geq 1/\alpha \}$,
\begin{equation}
  \sup_\Pin \PP_{\P} ( \exists n \in \NN : \phi_n^\brackalpha = 1 ) \leq \sup_\Pin \PP_{\P}
  ( \exists n \in \NN : \barW_n^{(\P)} \geq 1/\alpha ) \leq \alpha
  \EE_\P[\barW_1^{(\P)}] \leq \alpha.
\end{equation}
In fact, \citet{ramdas2020admissible} demonstrate that there is a formal sense
in which thresholding a test $\P$-supermartingale at $1/\alpha$ is the only
admissible way to construct a level-$\alpha$ test. We do not dwell on the
details of this admissibility result but we will regard it as a reasonable
justification to focus exclusively on such tests for the remainder of the
paper---indeed, the majority of sequential tests are developed in exactly this
way. An incomplete list of examples can be found in 
\citep{robbins1970statistical,robbins1974expected,howard2021time,waudby2024estimating,orabona2023tight,waudby2025universal}.

\subsection{Log-optimality of e-processes}\label{section:prelim-logopt}

For a given composite null $\Pcal$, many $\Pe$-processes may exist, and it is
often of interest to claim that some are more ``powerful'' than others. Since
larger values of a $\Pe$-process signify more evidence against $\Pcal$, the
testing by betting literature designates processes that diverge to infinity with
large growth rates as powerful---this notion can be traced back to the seminal
work of \citet{kelly1956new}, \citet{breiman1961optimal}, and
\citet{long1990numeraire} with explicit connections to hypothesis testing in
several recent works
\citep{grunwald2019safe,waudby2024estimating,larsson2025numeraire,waudby2025universal,orabona2023tight,wang2025backtesting}.
Let us informally motivate the desire to maximize the growth rate. 
Without loss of generality, test supermartingales can always be written as
products of nonnegative multiplicative increments that are conditional
$\Pe$-values \citep[Proposition 3]{waudby2024estimating}, which we denote as
$\infseqn{E_n}$. Thus, under $\Qin$ and for some constant $\ell_{\Q} \in \R$,
\begin{equation}
  W_n \coloneqq \prod_{i=1}^n E_i = \exp \left \{ \sum_{i=1}^n \log(E_i) \right \} =
  \exp \left \{ n\ell_\Q + o(n) \right \} \quad\text{$\Q$-almost surely} \mcom
\end{equation}
where the final almost-sure relation follows from an application of some
appropriately chosen strong law of large numbers under correspondingly appropriate conditions on
$\infseqn{E_n}$ with respect to $\Q$.
As such, it is typically of interest to find those $\Pe$-processes with large
almost-sure limiting values of $n^{-1} \log(W_n)$ as $\nto$. The previous
informal discussion motivates the following definition of almost-sure asymptotic log-optimality \citep{waudby2025universal}.
\begin{definition}[Single-armed log-optimality]
\label{definition:log-optimality}
  Fix a null $\Pcal$ and an alternative $\Qcal$. For $\Qin$, a $\Pe$-process
  $\infseqn{W_n}$ is said to be $\Q$-\emph{log-optimal} in a class of $\Pe$-processes
  $\Wcal$ if for any $W' \in \Wcal$,
  \begin{equation}
    \liminf_{n \to \infty} \Paren{\frac{1}{n} \log(W_n) - \frac{1}{n} \log(W_n')} \geq 0 \quad\text{$\Q$-almost surely.}
  \end{equation}
  We say that $\infseqn{W_n}$ is $\Qcal$-\emph{log-optimal} if the same process is $\Q$-log-optimal for every $\Qin$.
\end{definition}

\cref{definition:log-optimality} states that a $\cP$-$e$-process $\infseqn{W_n}$ 
is log-optimal within a class $\cW$ if no other process within $\cW$ has an asymptotic almost sure growth rate that is strictly larger than $\infseqn{W_n}$. In
\cref{section:multi-armed} we generalize \cref{definition:log-optimality} to the multi-armed settings
with partial information (see \cref{definition:multi armed oracle logopt}).

In the next subsections we introduce the hypothesis testing setting and data
collection protocol that we consider throughout the
rest of this paper.  We also introduce the necessary notation and concepts that
will allow us to define a multi-armed version of the log-optimality
definition from \cref{definition:log-optimality}.

\subsection{The multi-armed data collection protocol}
\label{section:prelim-multi-armed}
Throughout the remainder of the paper, we will rely on the data collection protocol outlined in \cref{algorithm:data collection}. Informally, \cref{algorithm:data collection} states that the statistician must devise a rule for choosing an action $A_n$ at time $n \in \NN$ based on all of the data observed thus far, $(A_i, Y_i(A_i))_{i=1}^{n-1}$, and after committing to that action they observe $Y_n(A_n)$ but not $Y_n(a)$ for any $a \neq A_n$. 

\begin{algorithm}[!htbp]
    \caption{Multi-armed data collection}
    \label{algorithm:data collection}
    Nature selects a distribution $\P \in \cP \cup \cQ$.

    \For{each time step $n = 1, \dots$}{
      \begin{enumerate}
        \item Nature samples a $K$-vector of outcomes:
            $\Paren{Y_n(1), \dots, Y_n(K)} \sim \P$.
        \item The statistician chooses $A_n$ based on the history $(A_i, Y_i(A_i))_{i=1}^{n-1}$.
        \item The statistician observes $Y_n(A_n)$.
      \end{enumerate}
    }
\end{algorithm}

Despite the fact that the statistician does not observe $Y_n(a)$ for any $a \neq A_n$ at time $n$, we will frequently make use of a filtration that is based on oracle access to the full history of outcomes $(Y_i(1),\dots, Y_i(K))_{i=1}^{n-1}$, where we assume that $K \in \N$. We define this object formally now.
\begin{definition}[History-oracle filtration]
    Consider the sequence of tuples $(Y_n(1), \dots, Y_n(K))_{n \in \NN}$ generated according to Step 1 of \cref{algorithm:data collection}. With $\Hcal_0 = \{ \emptyset, \Omega \}$ being the trivial sigma-algebra, we refer to the nondecreasing sequence of sigma-algebras $\Hcal \equiv (\Hcal_n)_{n \in \NN \cup \{ 0 \} }$ given by
    \begin{equation}
        \Hcal_n \coloneqq \sigma \left ( (Y_i(1), \dots, Y_i(K))_{i=1}^n \right )
    \end{equation}
    as the \emph{history-oracle filtration}.
\end{definition}
We will often refer to objects that are $\cH_{n-1}$-measurable for each $n \in \N$ as being $\cH$-\emph{predictable}.
We remark that the history-oracle filtration $\Hcal$ is a purely mathematical object that will be central to many of our technical results; for example, test supermartingales, $e$-processes, and stopping times will all be defined with respect to $\Hcal$. However, it is important to keep in mind that the statistician observing data according to the protocol in \cref{algorithm:data collection} will only ever have access to a \emph{strict subset} of $\Hcal_{n-1}$ at time $n$ as they will only see those outcomes $(Y_i(A_i))_{i=1}^{n-1}$ along the path of actions taken $(A_i)_{i=1}^{n-1}$. It is in this sense that the statistician operates in a \emph{partial information} setting. As we will see in \cref{section:multi-armed}, it is possible to devise $e$-processes based on partial information that have the same asymptotic behavior as those with access to the entire history-oracle filtration.

\subsection{Global null hypothesis testing}
\label{section:prelim-global null}

Following \cref{algorithm:data collection}, suppose that at
each time step $n \in \NN$ the vector $(Y_n(1), \dots, Y_n(K))$ is drawn from some distribution $\P$ in a collection $\widebar{\cP}$ but the statistician only observes $Y_n(A_n)$ for some chosen $A_n \in \cA \coloneqq \Set{1, \dots, K}$.
Within this multi-armed data collection setup, we focus exclusively
on deriving test statistics for a \emph{global null} hypothesis. To elaborate, let $\widebar \cP$ be
a collection of joint distributions on $Y_1(1),\dots, Y_1(K)$. 
Let $\Pcal_1, \dots, \Pcal_K \subseteq \widebar \Pcal$ be sub-collections of
distributions so that for each $a \in \Acal$, $\Pcal_a$ encodes some property of
$Y_1(a)$---e.g., $\EE_\P[Y_1(a)] = \mu_a$; $\mu_a \in \RR$---but is uninformative
about the marginal distributions of $Y_1(a')$ for $a' \neq a$.
The \emph{global} null hypothesis consists of the intersection of
$\cP_{1}, \dots, \cP_{K}$. Formally, consider the global null $\Pcal$
and its alternative $\Qcal$ as
\begin{equation}\label{eq:prelim-global null}
    \cP \coloneqq \bigcap_{a \in \cA} \cP_{a} \quad\text{and}\quad
    \cQ \coloneqq \bigcup_{a \in \cA} \cP_{a}^{c} \mcom
\end{equation}
respectively. In words, if $\Pcal_1,\dots, \Pcal_K$ contains distributions that encode properties of $Y_1(1),\dots, Y_1(K)$, respectively, then $\Pcal$ consists of those distributions where \emph{all} of those properties hold, while $\Qcal$ consists of those distributions where \emph{at least one of them} does not hold. Relating this discussion to the one in \cref{section:prelim-logopt}, this paper considers the task of deriving $\Pe$-processes under the global null $\Pcal$ that will enjoy optimality guarantees (to be defined in \cref{section:multi-armed}) under $\Qcal$.

One may wonder whether the multi-armed setting introduces
complexity to the validity of global null hypothesis tests. As illustrated by the following proposition, this is not the case.

\begin{proposition}[Type-I error control under multi-armed data collection]\label{proposition:validity}
    Let $\Pcal$ be a global null hypothesis in the sense of \eqref{eq:prelim-global null}.
    Suppose that $\infseqn{(Y_n(1), \dots, Y_n(K))} \sim \P$ are 
    \iidtext~draws from the joint distribution $\P \in \cP$. Furthermore, suppose that $\infseqn{f_n}$ 
    is a sequence of nonnegative $\Hcal$-predictable maps from $\RR$ to $[0, \infty)$, meaning that $f_n$ is $\Hcal_{n-1}$-measurable. 
    Suppose that for every $n \in \NN$, every $\Pin$, and every $a \in \Acal$,
    \begin{equation}
        \EE_\P [f_n(Y_n(a)) \mid \Hcal_{n-1}] \leq 1\quad \P\text{-almost
        surely}\mper
    \end{equation}
    Then, for any $\Hcal$-predictable $\infseqn{A_n}$, the process $\infseqn{M_n}$ given by
    \begin{equation}
    \label{eq:validity}
        M_n \coloneqq \prod_{i=1}^n f_i(Y_i(A_i))
    \end{equation}
    is a test $\Pcal$-supermartingale. Consequently, for any $\alpha \in (0, 1)$, 
    $\phi_n^\brackalpha \coloneqq \1 \left\{ M_n \geq 1/\alpha \right\}$ forms a 
    level-$\alpha$ sequential hypothesis test for $\Pcal$.
\end{proposition}
The conclusion one should draw from \cref{proposition:validity} is that
adaptive, multi-armed, data collection does not add substantial complication to the
single-arm case insofar as \emph{time-uniform type-I error} control under the
global null is concerned. Indeed the simplicity of the proof of
\cref{proposition:validity} (found in \cref{proof:validity}) reflects this
conclusion. However, complexity naturally arises when we pose the question:
{``What is a \emph{powerful} rule for choosing the sequence of arms
$\infseqn{A_n}$?''}, especially since we will operate in a restricted partial
information setting where $(K-1)$ of the outcomes $(Y_n(1),\dots, Y_n(K))$
are unobservable. This paper provides a comprehensive answer to that
question, both by articulating notions of ``power'' and proposing explicit
algorithms attaining those notions under partial information.

Let us briefly discuss two concurrent works that consider a thematically similar problem of testing global nulls under multi-armed data collection.
\begin{remark}[On the related work of \citet{hsu2025active}]
   In \citet{hsu2025active}, the authors focus on a global null generalization
   of two-sample testing when finitely many data sources are available for
   comparison.
   They demonstrate that their nonasymptotic sequential test will reject with
   probability one under the alternative and they derive an upper bound on the
   expected time to rejection. When viewed in contrast to the present work,
   their bound is loose in general as it relies on a commonly employed tuning
   parameter selection strategy (a ``betting strategy'') that is suboptimal for
   certain alternatives even in the single-armed case (see
   \citep{waudby2025universal} for discussions). For reasons that are made clear
   in \cref{section:ingredients}, the use of an ``optimal'' strategy introduces
   some mathematical complexity that the present work focuses on tackling. We
   also note that by virtue of studying two-sample testing,
   \citet{hsu2025active} must simultaneously estimate a so-called ``witness
   function''---an important object when using integral probability metrics as
   the authors do---while balancing exploration and exploitation. Because of
   this additional estimation task, the analysis of the present work would
   require nontrivial refinements to handle the testing problem they consider,
   and this is a possible direction for future work.
\end{remark}
\begin{remark}[On the related work of \citet{imbens2026demonstration}]
   The work of \citet{imbens2026demonstration} has similar motivations to
   \citet{hsu2025active} and the present work in the sense that the authors are
   interested in exploring arms to combine evidence across them in an effort to
   discredit a global null hypothesis. Their main results provide asymptotic (in
   the spirit of the central limit theorem) tests of the global null at a finite
   sample size through the use of strong approximations and asymptotic
   sequential inference tools
   \citep{robbins1970boundary,waudby2024time,waudby2023distribution}. They also
   make sub-Gaussian tail assumptions common in upper confidence bound-type
   analyses. The present work studies a different statistical inference
   desideratum of \emph{nonasymptotic} sequential testing with a focus on its
   optimality, and considers distributional assumptions that disallow those
   sub-Gaussian tail assumptions.
\end{remark}

We additionally highlight the work of \citet{yang2017framework}, who study a multiple testing problem under bandit feedback with the goal of
controlling the false discovery rate. \citet{cho2024peeking}
derive simultaneous confidence regions for means of random vectors under bandit feedback. Their work, nonetheless, focuses on testing hypotheses for each 
arm and not on testing global null hypothesis. A recent work from \citet{bharti2026global}
focuses on testing a global null hypothesis but considers a full information setting.

\subsection{A nonparametric class of test supermartingales and e-processes}\label{section:prelim-test supermartingales and e-processes}

Throughout the rest of the paper we focus on test $\cP$-supermartingales that 
can be constructed by taking convex combinations of $\Pe$-values, which are themselves
functions of the observed random variable in each round. As we elaborate on in \cref{example:two-sided mean testing} and later in \cref{section:multi-armed-testing-examples}, this abstract setup includes several nonparametric testing problems of interest in the literature. Formally, fix $d \in \N$ and let $\infseqn{A_n}$ be an $\Hcal$-predictable sequence of actions taking values in $\{1,\dots, K\}$. For each time step $n \in \N$ we define the 
vector of $\Hcal_{n-1}$-conditional $\Pe$-values under the global null hypothesis $\cP$ as the $(d+1)$-vector
\begin{equation}
 \bE_n(A_n) 
\coloneqq \left (E_n^{(0)}(Y_n(A_n)), \dots, E_n^{(d)}(Y_n(A_n))\right )
\end{equation}
 for some $\Hcal_{n-1}$-measurable maps $E_n^\brackzero(\cdot), \dots, E_n^\brackd(\cdot)$ from $\RR$ to $[0, \infty)$. Let $\infseqn{\bbet_n}$ be $\Hcal$-predictable $\simplex$-valued random variables which we view as the ``portfolios'' at each time---a term we will justify further in \cref{section:portfolio-and-allocation-regret}. Define the test statistic $\widebar W \equiv \infseqn{\barW_n}$ as follows:
\begin{equation}\label{eq:prelim-test supermartingale}
    \barW_n \coloneqq \prod_{i=1}^n \bbet_i^\trans \bE_i(A_i) \mper
\end{equation}
Here we have instantiated the test $\cP$-supermartingale from \eqref{eq:validity}
using the map
\begin{equation}
  f_n(y) = \bbet_n^\trans (E_n^\brackzero(y(0)) \cdots E_n^\brackd(y(d))^\trans; \quad y \in \RR^{d+1},
\end{equation}
and hence we conclude that $\widebar W_n$ forms a test $\Pcal$-supermartingale. We will make the following set of additional assumptions on the $\Pe$-values that make up the multiplicative increments of $\barW$.

\begin{assumption}\label{assumption:class of eprocesses}
    For each $a \in \Acal$, assume that $\infseqn{\bE_n(a)}$ 
    are \iidtext{} $(d+1)$-vectors of $\Pe$-values, meaning that $E_n^{(j)}(Y_n(a)) \geq 0$
    for each $j \in [0:d] \coloneqq \Set{0, \dots, d}$ with $\P$-probability one for each $\Pin$ and
    \begin{equation}
        \sup_{\Pin} \E_{\P}\Brac{E^{(j)}_n(Y_n(a))} \leq 1 \quad\text{for each }
        j \in [0:d] \mper
    \end{equation}
    Furthermore, assume that there exists some $b > 1$ so that for each $n \in \N$, 
    $a \in \cA$, and under all $\Pin$ and $\Qin$,
    \begin{equation}
        \sup_{\bbet \in \simplex} \Set{\bbet^\trans \bE_n(a)} \leq b
        \quad\text{almost surely}\mcom
    \end{equation}
    and that there exists some $\tilde{\bbet} \in \simplex$ so that
    under all $\Qin$,
    \begin{equation}\label{eq:assump-wealth-increment-equal-to-one}
        \tilde{\bbet}^\trans \bE_n(a) = 1
        \quad\text{$\Q$-almost surely}\mper
    \end{equation}
\end{assumption}
Whenever we impose \cref{assumption:class of eprocesses} without including the
arm index, one should think of this assumption as holding for the single-arm
setting (i.e., $K = 1$). We emphasize that the \iidtext{} assumption is made
with respect to how the vectors are sampled at each time step $n \in \N$, but
that the $\Pe$-values within each vector $\bE_n(a)$ for each $a \in \cA$ and $n
\in \N$ can be highly dependent of each other. Moreover, we remark that
\eqref{eq:assump-wealth-increment-equal-to-one} is a particularly weak
assumption since if it does not hold, one can simply consider an additional
$e$-value $E_n^{(d+1)} = 1$ whereby it suffices to take $\tilde \bbet = (0, 0,
\dots, 0, 1) \in \triangle_{d+1}$. However, there do exist interesting vectors
of $e$-values for which no single element deterministically takes the value one
but nevertheless satisfy \eqref{eq:assump-wealth-increment-equal-to-one} (see
\cref{example:two-sided mean testing} below).

The optimality definitions and results to come are stated with respect to the following
comparator class of $\Pe$-processes.
\begin{definition}[The oracle-history comparator class of $\Pe$-processes]
\label{def:comparator-class}
    Fix a null hypothesis $\Pcal$, an integer $d \in \NN$, and a constant $b > 1$. Let $\Wcal$ be a collection of stochastic processes satisfying the following conditions for any $W \in \Wcal$.
    \begin{enumerate}
        \item For every $\Pin$, $W$ is upper bounded by a test $\P$-supermartingale $\barW$ of the form \eqref{eq:prelim-test supermartingale}.
        \item The upper-bounding test supermartingale $\barW$ satisfies \cref{assumption:class of eprocesses} with the constant $b$.
    \end{enumerate}
    We refer to $\Wcal$ as the \emph{oracle-history comparator class}.
\end{definition}
We note that there are many
processes in $\Wcal$ that could not have been generated according to
\cref{algorithm:data collection} because $\Hcal$-predictability implies that
$\infseqn{\bbet_n}$ and $\infseqn{A_n}$ could have been constructed in a
\emph{full information} setting.

As alluded to previously, the conditions of \cref{assumption:class of eprocesses} 
hold in many practical settings that have been of interest in the literature. 
We exhibit this in the following example that considers a two-sided bounded mean
testing problem, as well as in \cref{section:treatment effects} where we present
an example of positive treatment effect testing. We defer additional examples to
\cref{section:multi-armed-testing-examples}, and refer the reader to 
\citet{orabona2023tight} for
an application to the two-sided mean testing setup described in the example
below.

\begin{example}[Two-sided bounded mean testing]\label{example:two-sided mean testing}
    Consider the problem of testing whether the mean of a bounded
    random variable is equal to some $\mu_0 \in [0,1]$. This problem has been
    previously studied in the single-arm setting in several works; see
    \citep{hendriks2021test,waudby2024estimating,orabona2023tight,ryu2024confidence,shekhar2023nonparametric,chugg2023auditing,podkopaev2023sequentialKernelized,podkopaev2023sequentialTwoSample,ryu2025improved}.
    We now show how it can be instantiated in the multi-armed testing by betting
    setting.
    Suppose that $\infseqn{Y_n(1), \dots, Y_n(K)}$ are sampled $\iid$ and are
    supported on $[0,1]^K$. The statistician is interested in testing the
    following two-sided global null $\cP^{=}$ versus alternative $\cQ^{\neq}$:
    \begin{equation}
        \cP^= = \Set{\P \svert \forall a \in \cA,~ \E_{\P}\Brac{Y(a)} = \mu_0}
        \quad\mathrm{versus}\quad \cQ^{\neq} = \Set{\P \svert \exists a \in
        \cA,~ \E_{\P}\Brac{Y(a)} \neq \mu_0} \mcom
    \end{equation}
    for some $\mu_0 \in [0,1]$. Using a $[0,1]$-valued predictable sequence
    $\infseqn{\bet_n}$ of portfolios and a predictable sequence $\infseqn{A_n}$ of
    arm pulls, the statistician can construct the following test
    $\cP^{=}$-supermartingale:
    \begin{equation}\label{eq:prelim two sided}
        W_n^{=} \coloneqq \prod_{i=1}^n \Brac{(1 - \bet_i) \frac{1 - Y_i(A_i)}
        {1 - \mu_0} + \bet_i \frac{Y_i(A_i)}{\mu_0}} \mper
    \end{equation}
    In this two-sided bounded mean testing setting, the bound on the
    multiplicative increments is $b = \max\Set{1/(1-\mu_0), 1/\mu_0}$ 
    and we also have that $\tilde{\bet}$ from \cref{assumption:class of eprocesses}
    takes the form $\tilde{\bet} = \mu_0$.
\end{example}

Later on in the paper we focus on $e$-processes that are not test
supermartingales, and we do this purely for computational resasons (see the
discussion following \cref{lem:spruce-test-statistic-eprocess}). Those
$e$-processes take the form of pathwise lower bounds on the process in
\eqref{eq:prelim-test supermartingale} using certain regret bounds due to
\citep{cover1996universal} that are elaborated on in
\cref{section:portfolio-and-allocation-regret}.  Nonetheless, for now, we focus on \emph{defining}
multi-armed log-optimality and defer explicit constructions of $e$-processes and
test supermartingales to \cref{algorithm:spruce}.

\section{Multi-Armed Log-Optimality}\label{section:multi-armed}

In this section, we first introduce a definition of multi-armed log-optimality and
then present sufficient conditions for achieving this definition. The section
culmiantes with explicit $e$-processes that satisfy those sufficient
conditions.

\begin{definition}[Multi-armed log-optimality]\label{definition:multi armed oracle logopt}
  Fix the global null $\Pcal$ and and its alternative $\Qcal$. Let $\infseqn{W_n}$ be a $\Pe$-process constructed via the protocol in \cref{algorithm:data collection}. We say that $\infseqn{W_n}$ is \emph{multi-armed $\Qcal$-log-optimal}
  if for all $\Qin$ and any other process $\tW \in \Wcal$  
  \begin{equation}
    \liminf_{n \to \infty} \left ( \frac{1}{n} \log (W_n) - \frac{1}{n} \log(\tW_n)  \right ) \geq 0 \quad \text{$\Q$-almost surely.}
  \end{equation}
\end{definition}
Let us briefly parse the above. First, notice that when $K=1$, \cref{definition:multi armed oracle logopt} recovers \cref{definition:log-optimality} for the class $\Wcal$. In the case of $K \geq 2$, \cref{definition:multi armed oracle logopt} states that a multi-armed log-optimal process $W$ does not have a limiting growth rate that is smaller than \emph{any other} process $\tW$ of the form \eqref{eq:prelim-test supermartingale}, even those that can construct $\infseqn{\bbet_n}$ and $\infseqn{A_n}$ with oracle access to the \emph{entire history} $(A_i, (Y_i(1), \dots, Y_i(K)))_{i=1}^{n-1}$ at time $n$. In particular, $W$ is not asymptotically outperformed by the process that has oracle knowledge of both the optimal arm $a_\Q$ \emph{and} the optimal portfolio $\bbet_\Q(a_\Q)$ under that arm, both of which are given by
\begin{equation}
\label{eq:log-optimal-strategy}
    \optarm \coloneqq \argmax_{a \in \cA}
    \E_{\Q}\Brac{\logp{\optbbet(a)^\trans \bE(a)}} \quad \text{where}\quad \bbet_\Q(a) = \argmax_{\bbet \in \simplex} \EE_\Q \left [ \log (\bbet^\trans \bE(a)) \right ].
\end{equation}
Consequently, a multi-armed log-optimal process $W$---should it exist---would have the property that
\begin{equation}
    \liminf_{n \to \infty}  \frac{1}{n}\log(W_n) \geq \max_{(a, \bbet) \in \Acal \times \simplex} \EE_\Q \left [\log \left ( \bbet^\trans \bE(a) \right ) \right ]
\end{equation}
with $\Q$-probability one, a simple consequence of the strong law of large numbers when considering the comparator process given by $W_\nQ^\star \coloneqq \prod_{i=1}^n \bbet_\Q(a_\Q)^\trans \bE(a_\Q)$.
Crucially, a multi-armed $\Qcal$-log-optimal process would attain such asymptotic growth rates without any prior knowledge of $\Q, a_\Q$, nor $\bbet_\Q(a)$ for any $a \in \Acal$. From this vantage point, while the process given by $W_\nQ^\star$ may be $\{ \Q \}$-multi-armed log-optimal in a distribution-pointwise sense, it is not when $\Qcal$ consists of suitably different distributions.

\begin{remark}[On a multi-armed interpretation of Breiman's favorable games \citep{kelly1956new,breiman1961optimal}]\label{remark:favorable games}
    The goal of maximizing the expected log-wealth (i.e., expected log-increment) 
    dates back to the work of \citet{kelly1956new} and is often referred to as the ``Kelly criterion.'' 
    \citet{breiman1961optimal} provided a study of so-called ``favorable games''---stochastic repeated games for which there exists a gambling strategy
    that allows the wealth of a gambler playing that game to diverge to infinity---and showed that the unique fixed gambling
    strategy satisfying the Kelly criterion is optimal in a few different senses. The data collection protocol in \cref{algorithm:data collection} can be viewed as a repeated stochastic game where at each time step $n \in \NN$, the gambler must not only place a bet $\bbet \in \simplex$ but most also choose \emph{which sub-game $a \in \Acal$ to play} prior to placing that bet. \cref{definition:multi armed oracle logopt} can be viewed as the growth rate of that gambler matching the one that knows both which sub-game is \emph{most favorable} and how to bet optimally on that most favorable sub-game.
\end{remark}

As mentioned above, taking $\Acal = \{a\}$ to be a singleton,
\cref{definition:multi armed oracle logopt} recovers
\cref{definition:log-optimality} \citep[Definition 1]{waudby2025universal}. In that case, \citet[Theorem 2.1]{waudby2025universal} show that a sufficient condition for achieving log-optimality is obtained when the $\Pe$-process $\infseqn{W_n}$ exhibits so-called \emph{sublinear portfolio regret}. We introduce an arm-wise version of portfolio regret in the following subsection.
\subsection{Portfolio and allocation regret}
\label{section:portfolio-and-allocation-regret}

Recall the (single-armed) notion of log-optimality in \cref{definition:log-optimality}. Even in the single-armed case where $K=1$, it is not immediately clear how one should choose the $\simplex$-valued vectors $\infseqn{\bbet_n}$ to achieve log-optimality. \citet[Theorem 2.1]{waudby2025universal} show that if $\log(W)$ enjoys a particular deterministic regret bound, then it will be log-optimal in the stochastic sense of \cref{definition:log-optimality}, and that there are several algorithms achieving such a bound. In the language of the present paper, the authors show that in the case of $K=1$, if $W$ is a $\Pe$-process for which the \emph{portfolio regret}
\begin{equation}
    \Rcal^\Port_n \coloneqq \max_{\bbet \in \simplex} \sum_{i=1}^n \log \left ( \bbet^\trans \bE_i \right ) - \log (W_n)
\end{equation}
is pathwise sublinear (i.e., for every $\omega \in \Omega$), then $W$ is $\Qcal$-log-optimal in the sense of \cref{definition:log-optimality}. Note that $\Rcal^\Port_n$ is precisely the form of regret considered by \citet{cover1991universal} and \citet{cover1996universal}. Moreover, those authors provide and analyze an algorithm attaining \emph{logarithmic} portfolio regret which we now describe. Suppose that $\bbet_n^\UP$ is chosen according to the rule
\begin{equation}
    \bbet_n^\UP \coloneqq \frac{\int_{\bbet \in \simplex} \bbet \barW_{n-1}(\bbet) \dd F(\bbet)}{\int_{\bbet \in \simplex} \barW_{n-1}(\bbet)\dd F(\bbet)},
\end{equation}
where $\barW_n(\bbet) = \prod_{i=1}^n \bbet^\trans \bE_i$ and $F(\bbet)$ is a $d$-dimensional Dirichlet($1/2,\dots, 1/2$) distribution. Then for the process $\barW_n^\UP \coloneqq \prod_{i=1}^n (\bbet^\UP_i)^\trans \bE_i$, the portfolio regret is bounded by $d\log(n + 1) / 2 + \log(2)$ in a pathwise sense \citep{cover1996universal}. Using \citet[Theorem 2.1]{waudby2025universal}, $\barW$ is single-arm $\Qcal$-log-optimal, as is any $\Pe$-process that lower-bounds it; see also \citet{orabona2023tight} and \citet[Corollary 2.2]{waudby2025universal} for examples of such lower-bounding $e$-processes.

For the purposes of this work, we rely on a notion of \emph{arm-wise} portfolio
regret. This is a straightforward extension, but it is worth defining rigorously
because, as discussed in \cref{proposition:oracle log-optimality abstract}, it
is key to formulating one of the two sufficient conditions for attaining
multi-armed log-optimality. In order to make the ``arm-wise'' aspect of arm-wise
portfolio well-defined, we introduce an assumption on the form the
$e$-processes take.
\begin{assumption}[$E$-processes as arm-wise products]\label{assumption:arm-wise products}
   Let $\infseqn{W_n}$ be a $\Pe$-process constructed from data collected according to \cref{algorithm:data collection}. Assume that $W_n$ can be written as
   \begin{equation}
       W_n = \prod_{a = 1}^K W_n(a)
   \end{equation}
   where for each $a \in \Acal$, $W_n(a)$ depends only on $\1 \{ A_1 = a \} Y_1(a), \dots, \1 \{ A_n = a \} Y_n(a)$. 
\end{assumption}

Every $e$-process and test supermartingale considered throughout the paper
satisfies \cref{assumption:arm-wise products}. With \cref{assumption:arm-wise
products} in mind, we define the arm-wise portfolio regret.

\begin{definition}[Arm-wise portfolio regret]\label{definition:arm-wise portfolio regret}
    Let $\infseqn{A_n, Y_n(A_n)}$ be collected according to the protocol in \cref{algorithm:data collection}. Let $\infseqn{W_n}$ be a $\Pe$-process satisfying \cref{assumption:arm-wise products}.
    For any $n \in \NN$ and $a \in \Acal$, we define the \emph{arm-wise portfolio regret} of $W_n$ on arm $a$ to be
    \begin{equation}
        \Rcal_n^\Port(a) \coloneqq \max_{\bbet \in \simplex} \sum_{i=1}^n \1 \{ A_i = a \} \log \left ( \bbet^\trans \bE_i(a) \right ) - \log \left ( W_n(a) \right ).
    \end{equation}
\end{definition}
This definition should be interpreted as the portfolio regret obtained when considering only the data that has been generated from a specific arm $a \in \Acal$. Clearly, any algorithm that enjoys a portfolio regret of $\Rcal_n^\Port$ in the single-arm case will enjoy an arm-wise portfolio regret of $\Rcal^\Port_{N_a(n)}(a)$ for arm $a \in \Acal$ where $N_a(n) = \sum_{i=1}^n \1 \{ A_i = a \}$. For example, if the universal portfolio algorithm $\infseqn{\bbet_n^\UP}$ is run for each arm $a \in \Acal$, i.e.,
\begin{equation}\label{eq:arm-wise UP}
    \bbet_n^\UP(a) \coloneqq \frac{\int_{\bbet \in \simplex} \bbet \barW_{n-1}(\bbet; a) \dd F(\bbet)}{\int_{\bbet \in \simplex} \barW_{n-1}(\bbet; a)\dd F(\bbet)},
\end{equation}
where $\barW_n(\bbet; a) = \prod_{i=1}^n \left (\bbet^\trans \bE_i \right )^{\1 \{ A_i = a \}}$ and $F(\cdot)$ is the Dirichlet measure as before,
then the resulting test supermartingale will have an arm-wise portfolio regret of $d\log(N_a(n) + 1) / 2 + \log (2)$. 
We rely on this fact later on. In light of the previous discussion, it is clear
that attaining arm-wise portfolio regret does not present any unique technical
challenges. The main technical challenges arises when attaining sublinear
\emph{allocation regret}, a concept that we now introduce.

\begin{definition}[Portfolio-oracle allocation regret]\label{definition:allocation regret}
Fix a distribution $\Qin$ and for any $a \in \Acal$, define $\bbet_\Q(a) \coloneqq \argmax_{\bbet \in \simplex} \EE_\Q [\log (\bbet^\trans \bE(a)]$ as the log-optimal portfolio for arm $a$ under $\Q$. Define $a_\Q \coloneqq \argmax_{a \in \Acal} \EE_\Q [\log(\bbet_\Q(a)^\trans \bE(a)]$ as the arm with the largest log-optimal portfolio. For any $n \in \NN$, we define the \emph{portfolio-oracle allocation regret} $\Rcal_\nQ^\MAB$ as 
    \begin{equation}
        \Rcal_{n,\Q}^\MAB \coloneqq n \EE_\Q \left [\log(\bbet_\Q(a_\Q)^\trans \bE(a_\Q)) \right ] - \sum_{i=1}^n \EE_\Q  \left [ \log \left (\bbet_\Q(A_i)^\trans \bE(A_i) \right ) \right ].
    \end{equation}
\end{definition}
In words, $\Rcal_\nQ^\MAB$ is the difference between the sum of expected logarithmic increments under the optimal arm $a_\Q$ and the optimal portfolio $\optbbet(a_\Q)$ versus that of the sequence of arms $A_1, \dots, A_n$, under their respective optimal portfolios $\optbbet(A_1),\dots, \optbbet(A_n)$.
Notice that if viewing the ``reward'' at time $n$ as $\log(\optbbet(A_n)^\trans
\bE(A_n))$, then $\Rcal_\nQ^\MAB$ is precisely the notion of regret appearing in
the stochastic multi-armed bandit literature \citep{lai1985asymptotically},
\citep[\S 7]{lattimore2020bandit}. We use the qualifier ``portfolio-oracle''
because the second term in the definition of $\Rcal_\nQ^\MAB$ considers the
expectation of a random variable without knowledge of $a_\Q$ but with oracle
knowledge of $\bbet_\Q(a)$ for each $a \in \Acal$. For succinctness, however, we
often drop the qualifier and refer to $\cR_{\nQ}^\MAB$ as the ``allocation
regret.''

In the following proposition, we connect sublinear portfolio-oracle allocation regret to our overall goal of attaining multi-armed log-optimality, showing that it provides one of two sufficient conditions that yield such optimality.

\begin{proposition}[Multi-armed log-optimality from sublinear portfolio and allocation regrets]\label{proposition:oracle log-optimality abstract}
  Let $\infseqn{W_n}$ be a $\Pe$-process satisfying \cref{assumption:class of eprocesses,assumption:arm-wise products}.
    If $\infseqn{W_n}$ has sublinear arm-wise portfolio regret as well as sublinear allocation regret, then $\infseqn{W_n}$ is multi-armed log-optimal. Moreover,
    \begin{equation}
      \lim_\nto \frac{1}{n} \log(W_n) = \max_{(a, \bbet) \in \Acal\times \simplex} \EE_\Q \left [ \log \left ( \bbet^\trans \bE_1(a) \right ) \right ]
    \end{equation}
    $\Q$-almost surely.
\end{proposition}
A sketch of the proof of the second part of the proposition can be made short and illustrative so we provide one here. The full proof can be found in \cref{proof:oracle log-optimality abstract}.
\begin{proof}[Proof sketch of \cref{proposition:oracle log-optimality abstract}]
    Fix $\Qin$. Define the difference between the optimal portfolio under the optimal arm and the empirical growth rate: $\Rcal_n \coloneqq \EE_\Q \left [ \log(\optbbet(a_\Q)^\trans \bE_1(a_\Q) \right ] - n^{-1}\log(W_n)$. Notice that $\Rcal_n$ can be decomposed as
    \begin{align}
    \cR_n =\ &\EE_\Q \left [ \log \left ( \optbbet(\optarm)^\trans
    \bE(\optarm) \right ) \right ] - \frac{1}{n}\sum_{i=1}^n \EE_\Q \left [ \log \left (
    \optbbet(A_i)^\trans \bE_1(A_i) \right ) \right
            ]\\ 
                &+ \frac{1}{n}\sum_{i=1}^n \EE_\Q \left [ \log
            \left ( \optbbet(A_i)^\trans \bE_i(A_i) \right ) \right ] -
            \frac{1}{n} \sum_{i=1}^n \sum_{a=1}^K \1 \{A_i = a \} \log \left (
            \optbbet(a)^\trans \bE_i(a) \right )\\
                &+ \frac{1}{n}\sum_{i=1}^n \sum_{a=1}^K \1
                  \{A_i = a \} \log \left (
                  \optbbet(a)^\trans \bE_i(a) \right )
                  - \frac{1}{n} \sum_{a=1}^K \log(W_n(a))
                  .
    \end{align}
    Now, the first and third lines vanish by sublinearity of the allocation and arm-wise portfolio regrets, respectively. The second line vanishes by a concentration inequality plus an application of the Borel-Cantelli lemma. This last justification can be thought of as applying $K$ separate arm-wise strong laws of large numbers and is made rigorous in \cref{lemma:exponential concentration for bandit sequence - centered with independent copies}.
\end{proof}

In light of \cref{proposition:oracle log-optimality abstract}, it is of interest
to devise algorithms with sublinear portfolio-oracle allocation regret. However,
it is precisely because of the ``portfolio-oracle'' qualification that
off-the-shelf regret minimization algorithms from the bandit literature cannot
be applied. Indeed, they would assume access to (and place certain assumptions
on) the random variables $\left (\log(\bbet_\Q(A_i)^\trans \bE(A_i))\right)_{i=1}^n$ 
which cannot be observed under the assumptions that we make. To
address this key issue, in the next subsection we present an algorithm that
employs a nontrivial modification of upper-confidence-bound-type algorithms
\citep{lai1985asymptotically} \citep[\S 7]{lattimore2020bandit}. This
algorithm is designed to accommodate the unknown---but in a sense,
``learnable''---optimal portfolios. Rather than dwell further on these
technical considerations let us focus on presenting the aforementioned
algorithm and its optimality guarantees. Those technical considerations are
the sole subject of \cref{section:ingredients}.

\subsection{Achieving multi-armed log-optimality via \SPRUCEtext{}}
Let us introduce some notation that will aid in the exposition of \cref{algorithm:spruce}.
For any $n \in \NN$ and any $a \in \Acal$, let $N_a(n) \coloneqq \sum_{i=1}^{n}
\Ind{A_i = a}$ denote the (random) number of times arm $a$ has been pulled up until and
including time $n$. Furthermore, define the user-chosen parameters $\gamma > 2$
and $\zeta > 0$ that play the role of exploration incentives and disincentives,
respectively, and let $b > 1$ be the almost sure upper bound on the convex
combination of $\Pe$-values from \cref{assumption:class of eprocesses}. For any
$n \in \NN$, define
$R_{n}^{\CO} \coloneqq d\log(n +
1) / 2 + \log (2)$. Using these quantities, we define the upper confidence bound $\UCB_n(a)$ as follows:
\begin{align}
    \UCB_a(n) \coloneqq &\max_{\bbet \in \simplex} \frac{1}{N_a(n-1)}
                          \sum_{i=1}^n \Ind{A_i = a}\log(\bbet^\trans \bE_i(A_i)) \\
    &+ \sqrt{\frac{8 b \gamma \log(\zeta n + 1)}{N_a(n-1)}} 
    + \frac{4 \gamma \log(\zeta n + 1)}{N_a(n-1)} + \frac{R_{N_a(n-1)}^{\CO}}{N_a(n-1)}.
\end{align}
Based on this definition, we present
the algorithm \SPRUCEtext{} (Sublinear Portfolio Regret Upper Confidence Estimation) in \cref{algorithm:spruce}.

\begin{algorithm}[!htbp]
    \caption{: Sublinear Portfolio Regret Upper Confidence Estimation (\SPRUCEtext)}
    \label{algorithm:spruce}

    Select $\option \in \{ \UP, \CO \}$
    
    Pull each arm $a \in \cA$ once.

    \For{$n = K + 1, \dots$}{
      \begin{enumerate}
        \item Select $A_n = \argmax_{a \in \cA} \UCB_a(n)$, breaking ties arbitrarily.
        \item Observe $Y_n(A_n)$ from the $K$-vector of outcomes: $\Paren{Y_n(1), \dots, Y_n(K)} \sim \P$.
          \item Construct the $\Pcal$-$e$-values $\bE_n(Y_n(A_n))$.
          \item If $\option == \CO$, construct
            \begin{center}
              $\displaystyle W_n^\CO \coloneqq \prod_{a = 1}^K
              \exp \left \{ \max_{\bbet \in \simplex}\sum_{i=1}^n \1 \{A_i =
                a\} \log \left ( \bbet^\trans \bE_i(a) \right ) - R_{N_a(n)}^{\CO}
              \right \}.$
            \end{center}

            If $\option == \UP$, construct
            \begin{center}
             $\displaystyle \barW_n^\UP \coloneqq \prod_{i=1}^n \bbet_i^\UP (A_i)^\trans \bE(A_i)$,\quad where $\bbet_i^\UP$ is as in \eqref{eq:arm-wise UP}.
            \end{center}
          \end{enumerate}
        }
      \end{algorithm}

The process $W^\CO$ defined in \cref{algorithm:spruce} is visually distinct from the test supermartingales described earlier in the paper, including Cover's universal portfolio algorithm $\barW^\UP$, but it, like the universal portfolio algorithm, forms an $\Pe$-process under \cref{algorithm:spruce}, a fact that we state formally in the following proposition.

\begin{proposition}
\label{lem:spruce-test-statistic-eprocess}
The process $\barW^\UP$ is a test $\Pcal$-supermartingale and $W^\CO$ is a $\Pe$-process.
\end{proposition}
\begin{proof}
  The fact that $\barW^\UP$ is a test $\Pcal$-supermartingale follows immediately from \cref{proposition:validity} (see the discussion following \eqref{eq:prelim-test supermartingale}). The fact that $W^\CO$ is a $\Pe$-process follows from the regret bound of \citet{cover1996universal}, allowing us to conclude that for any $n \in \NN$,
  \begin{equation}
    W_n^\CO \leq \barW_n^\UP\quad \text{pathwise.}
  \end{equation}
  Since $\Pe$-processes are those nonnegative processes that are $\P$-almost surely upper-bounded by a test $\P$-supermartingale for each $\Pin$, this completes the proof.
\end{proof}
The use of an in-hindsight maximum less a regret bound to form an $e$-process was introduced by \citet{orabona2023tight} to derive sharp confidence sequences for means of bounded random variables. The same idea was employed in \citet{waudby2025universal} for the purposes of deriving log-optimality guarantees of a more abstract class of $e$-processes.

The reason that a practitioner may wish to use $W^\CO$ in place of $\barW^\UP$
is because the integral in the definition of $\bbet_n^\UP(\cdot)$ can be
computationally expensive and numerically unstable, while the maximum in the
exponential of $W_n^\CO$ can be computed efficiently via off-the-shelf
root-finding algorithms. In the results that follow, we state guarantees for
both $\barW^\UP$ and $W^\CO$ interchangeably.
We now show that \SPRUCE{} yields multi-armed log-optimal $e$-processes.

\begin{theorem}[Multi-armed log-optimality of \SPRUCEtext]
\label{theorem:log-optimality of spruce}
    Let $\infseqn{W_n}$ be one of the $\Pe$-processes constructed via \emph{\SPRUCE}
    (\cref{algorithm:spruce}). Then under \cref{assumption:class of eprocesses}, $\infseqn{W_n}$ is multi-armed $\Qcal$-log-optimal in
    the sense of \cref{definition:multi armed oracle logopt}. Furthermore, the following two properties hold.
    \begin{enumerate}[(i)]
        \item Fix $\Q \in \cQ$ and let $W_\nQ^{\star} \coloneqq \prod_{i=1}^n \optbbet(\optarm)^\trans \bE_i(\optarm)$. $W_n$ and $W_\nQ^\star$ are $\Q$-almost surely asymptotically equivalent, meaning that
        \begin{equation}
            \lim_{n \to \infty} \left (\frac{1}{n} \log(W_n) - \frac{1}{n}\log (W_\nQ^\star)\right ) = 0\quad\text{$\Q$-almost surely.}
        \end{equation}
        \item For every $\Q \in \cQ$, $W_n$ has an asymptotic growth rate given by
            \begin{equation}
                \lim_{n \to \infty} \frac{1}{n} \logp{W_n} = 
                \max_{(a, \bbet) \in \Acal \times
                \simplex} \EE_\Q \Brac{\log \Paren{\bbet^\trans \bE_1(a)}} 
                \quad\text{$\Q$-almost surely.}
            \end{equation}
            Finally, this result is unimprovable in the sense that for any other $\tW$  as in \cref{definition:multi armed oracle logopt} (with oracle access to the entire history $\Hcal_{n-1}$ at time $n \in \NN$),
            \begin{equation}
                \limsup_{n \to \infty} \frac{1}{n} \logp{\tW_n} \leq 
                \max_{(a, \bbet) \in \Acal \times
                \simplex} \EE_\Q \Brac{\log \Paren{\bbet^\trans \bE_1(a)}} 
                \quad\text{$\Q$-almost surely.}
            \end{equation}
    \end{enumerate}
\end{theorem}
The proof of \cref{theorem:log-optimality of spruce} can be found in
\cref{proof:log-optimality of spruce}. It relies on deriving a nonasymptotic concentration inequality for $|\log(W_{\nQ}^\star / W_n)|$ and applying the Borel-Cantelli lemma.
\begin{figure}
    \centering
    \begin{subcaptionblock}{0.485\textwidth}
        \centering 
        \includegraphics[width=\linewidth]{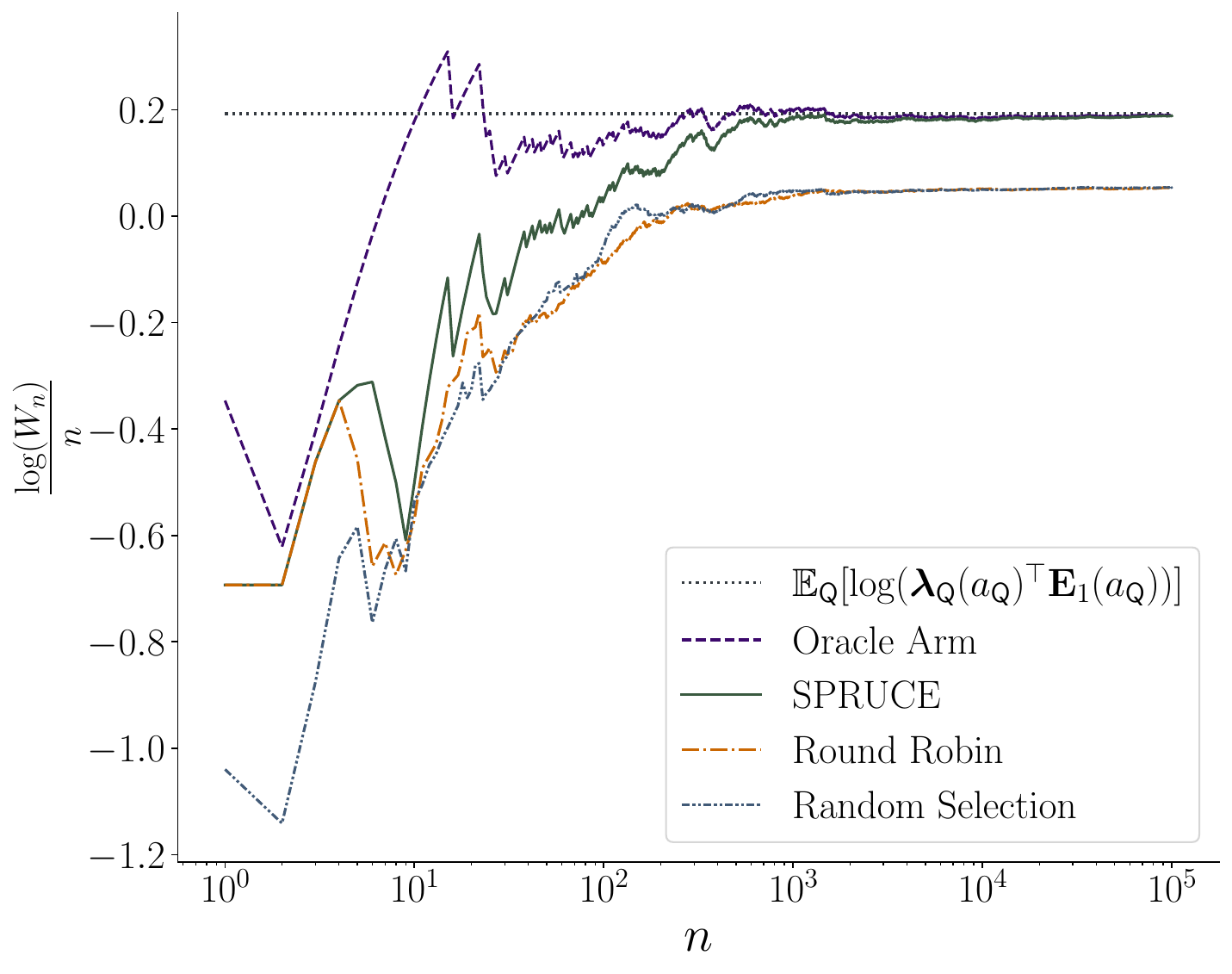}
    \end{subcaptionblock}
    \hfill
    \begin{subcaptionblock}{0.485\textwidth}
        \centering 
        \includegraphics[width=\linewidth]{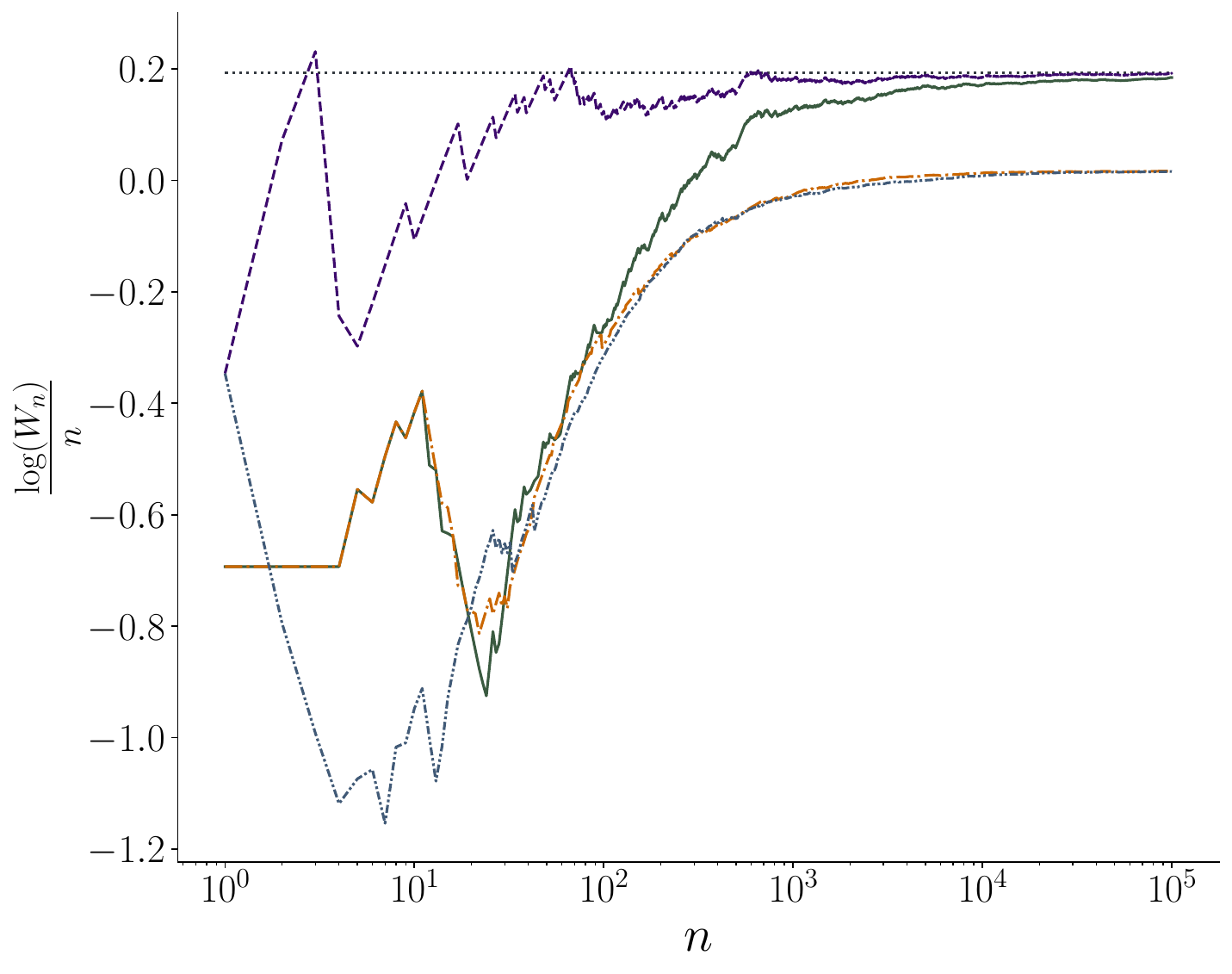}
    \end{subcaptionblock}
    \caption{Empirical growth rates for the one-sided bounded mean testing problem from \cref{example:one-sided mean testing} under ``easy'' (left) and ``hard'' (right) data generating processes. We consider three algorithms in addition to \SPRUCE. \emph{Oracle Arm} has oracle access to and solely pulls the optimal arm. \emph{Round Robin} pulls the arms one-by-one until
    all of them have been selected, and starts the process over again. 
    \emph{Random Selection} samples uniformly at random the arm to be played in 
    round $n$. 
    All four algorithms employ the regret-based test statistic
    from \cref{algorithm:spruce} and only differ in the way the arms are
    selected. Lastly, we note that the empirical growth rates of Round Robin and
    Random Selection are close to but nevertheless strictly greater than zero.}
    \label{fig:bounded-mean-testing}
\end{figure}
The exact inequality employed is a technical device, but its derivation involves another inequality that describes how $\frac{1}{n} \log(W_n)$ concentrates around the optimal growth rate. We present the latter inequality here and subsequently discuss its interpretation.
\begin{lemma}
\label{lemma:exponential concentration for bandit sequence - optimal}
  Let $\bE_1, \dots, \bE_n$ be $(d+1)$-vectors of $\Pe$-values satisfying 
  \cref{assumption:class of eprocesses}. Let $\infseqn{A_n}$ and $\infseqn{\bbet_n(A_n)}$ be chosen according to \cref{algorithm:data collection} (not necessarily according to \SPRUCE{}). Let $\infseqn{W_n}$ be a $\Pe$-process satisfying \cref{assumption:arm-wise products}. Furthermore, let $R_{\nQ}^{\MAB}$ and
  $R_n^{\mathrm{Port}}(a)$ defined for each $a \in \cA$ be monotonic sequences bounding the allocation and arm-wise portfolio regrets, respectively: $\Rcal_\nQ^\MAB \leq R_\nQ^\MAB$ and $\Rcal^{\Port}_n(a) \leq R_n^\Port(a)$ for each $a \in \cA$ and $n \in \NN$. Then, for any $\eps > 0$ and any $m \in \NN$,
  \begin{align}
    &\PP_\Q \Paren{\sup_{n \geq m } \left \lvert \frac{1}{n} \log(W_n) - \max_{(a, \bbet) \in \Acal \times \simplex}\EE_\Q \Brac{\log
    \Paren{\bbet^\trans \bE(a)}} \right
      \rvert \geq \eps}\\
        &\quad \leq
         \sum_{n=m}^\infty \1 \left \{ R_{\nQ}^\MAB \geq
        \epsilon n / 3 \right \} 
        + \sum_{n=m}^\infty \1  \left \{ \sum_{a=1}^K R^\Port_n(a) \geq \eps n / 3\right  \} \\
        &\quad+ 2K \sum_{n=m}^\infty n \left ( \exp \left \{ - \frac{\eps^2n }{72bK^2} \right \} + \exp \left \{ -\frac{\eps n }{12K} \right \} \right )
        + \sum_{n=m}^\infty \exp \left \{ -n\eps / 3 \right \} \mper
  \end{align}
\end{lemma}
The proof of \cref{lemma:exponential concentration for bandit sequence - optimal} can be found in \cref{proof:exponential concentration for bandit sequence - optimal}. Notice that the summands in the third and fourth terms are summable so the series from $n = m$ to $\infty$ vanishes as $m \to \infty$. Inspecting the summands of the first two terms, notice that if $R_{\nQ}^\MAB$ and $R_n^\Port$ are sublinear, then only finitely many summands are non-zero. Therefore, for $m$ sufficiently large, the first two terms would be zero. Taken together, we have that if the allocation and arm-wise portfolio regrets are sublinear, then the right-hand side of the inequality in \cref{lemma:exponential concentration for bandit sequence - optimal} vanishes as $m \to \infty$, and hence
\begin{equation}
    \PP_\Q \left ( \lim_{n \to \infty} \frac{1}{n} \log(W_n) = \max_{(a, \bbet) \in \Acal \times \simplex}\EE_\Q \Brac{\log
    \Paren{\bbet^\trans \bE(a)}} \right ) = 1.
\end{equation}
Indeed we can see such behavior in \cref{fig:bounded-mean-testing} which depicts how \SPRUCE~converges to the expected log-increment under the optimal arm and its optimal portfolio, and eventually matches the empirical growth rate of an algorithm that has oracle access to the optimal arm and employs the regret-based test statistic from \cref{algorithm:spruce}.

In the following section, we turn to a related problem, deriving lower and upper bounds on the expected number of samples required to reject the global null. While rejection time and growth rate are different goals, the bounds we obtain for the former turn out to depend on exactly the same quantity (the optimal expected log-increment) that appeared in \cref{theorem:log-optimality of spruce}.

\section{Analyzing the Expected Time to Rejection}
\label{section:rejection times}

Let us now consider another notion of ``power'' for sequential hypothesis tests: having a small expected number of
samples to reject some null hypothesis. Concretely, for a given $\Pe$-process $\infseqn{W_n}$ and a desired type-I error level $\alpha \in (0, 1)$, consider the stopping time $\tau_{\alpha} \coloneqq \inf
\Set{n \in \N : W_n \geq 1/\alpha}$ given by the first sample size for which that process exceeds the threshold $1/\alpha$. In this section, our aim is to study the expectation of this stopping time $\EE_\Q[\tau_\alpha]$ for distributions $\Qin$, deriving lower and upper bounds with respect to the space of arm-and-portfolio combinations.

In the single-armed case, the study of lower and upper bounds on such stopping
times has been of interest since the advent of sequential analysis with
\citet{wald1947sequential}, in the context of Breiman's favorable games
\citep{breiman1961optimal}, and especially in the best-arm identification
literature; see the works of
\citet{kaufmann2016complexity,garivier2016optimal,kaufmann2021mixture} and
\citet{agrawal2020optimal,agrawal2021regret,agrawal2021optimal}. For examples in
sequential hypothesis testing see the works of
\citet{chugg2023auditing,shekhar2023nonparametric}, and the
information-theoretic lower and upper bounds on stopping times for general
classes of testing problems found in \citet{agrawal2025stopping}. Lower and
upper bounds over the class of nonparametric portfolios of the form
\eqref{eq:prelim-test supermartingale} can be found in
\citet{waudby2025universal} for the case of $K=1$. The results to follow are
multi-armed analogues of those earlier results, in the setting where vanishing
allocation regret is of central importance.

For a fixed $\Qin$, we begin by deriving lower bounds on $\EE_\Q[\tau_\alpha]$
and we later show that they match an upper bound in the small-$\alpha$
regime for the stopping time $\tau_{\alpha,\Q}^\star \coloneqq \inf\{n \in \NN :
W_\nQ^\star \geq 1/\alpha \}$, where $W_\nQ^\star$ forms the oracle
$\Pe$-process  
constructed using both $a_\Q$ and $\bbet_\Q(a_\Q)$ and which was characterized in
\cref{theorem:log-optimality of spruce}.
\begin{proposition}[A lower bound on the expected rejection time]\label{proposition:stopping-timelower-bound}
  Fix an alternative distribution $\Qin$ and a desired type-I error rate $\alpha
  \in (0, 1)$. Let $\infseqn{\tW_n}$ be any process of the form
  \eqref{eq:prelim-test supermartingale}, that satisfies \cref{assumption:class of eprocesses}, 
  and for which $\bbet_n$ and $A_n$ are
  $\Hcal_{n-1}$-measurable for each $n \in \N$. Define the rejection time
  $\widetilde \tau_\alpha \coloneqq \inf\{ n \in \NN : \tW_n \geq 1/\alpha \}$.
  We then have the following lower bound on the expected time to rejection
  \begin{equation}
    \frac{\EE_\Q[\widetilde \tau_\alpha]}{\log(1/\alpha)} \geq \left ( \max_{(a, \bbet) \in
    \Acal \times \simplex} \EE_\Q \Brac{\log \left (\bbet^\trans \bE_1(a) \right )}\right
    )^{-1} \mper
  \end{equation}
  Moreover, for a given arm $a \in \Acal$, this lower bound holds for the
  stopping time $\tau_\alpha^\bracka \coloneqq \inf \{ n \in \NN :
  \tW_{n}^\bracka \geq 1/\alpha \}$ where $\tW_n^\bracka$ is any
  $\Pe$-process of the form \eqref{eq:prelim-test supermartingale} that
  satisfying \cref{assumption:class of eprocesses}, that pulls the $a$-th arm
  at each time step and selects an $\Hcal$-predictable portfolio $\infseqn{\bbet_n(a)}$.
\end{proposition}
The proof of \cref{proposition:stopping-timelower-bound} can be found in 
\cref{sec:proof-stopping-time-lower-bound}. The argument relies on comparing
$\tW_n$ to another process that pulls the same sequence of arms but has oracle
access to the log-optimal portfolios $\infseqn{\optbbet(A_n)}$. We show that the
difference between $\tW_n$ and the oracle process forms a supermartingale whose
mean is upper bounded by zero. The proof then follows by Doob's optional stopping
theorem and an application of Wald's identity. We highlight that version of
Doob's optional stopping theorem we employ in the proof of
\cref{proposition:stopping-timelower-bound} only assumes boundedness of the
expected absolute log-increments when we condition on the filtration $\cF_{n-1}$. We provide this version of
Doob's optional stopping theorem in 
\cref{lem:optional-stopping-supermartingale}---and whose proof can be found in 
\cref{proof:optional_stopping-supermartingale}---as this result seems to be new to the literature.

\begin{figure}
    \centering
    \begin{subcaptionblock}{0.485\textwidth}
        \centering 
        \includegraphics[width=\linewidth]{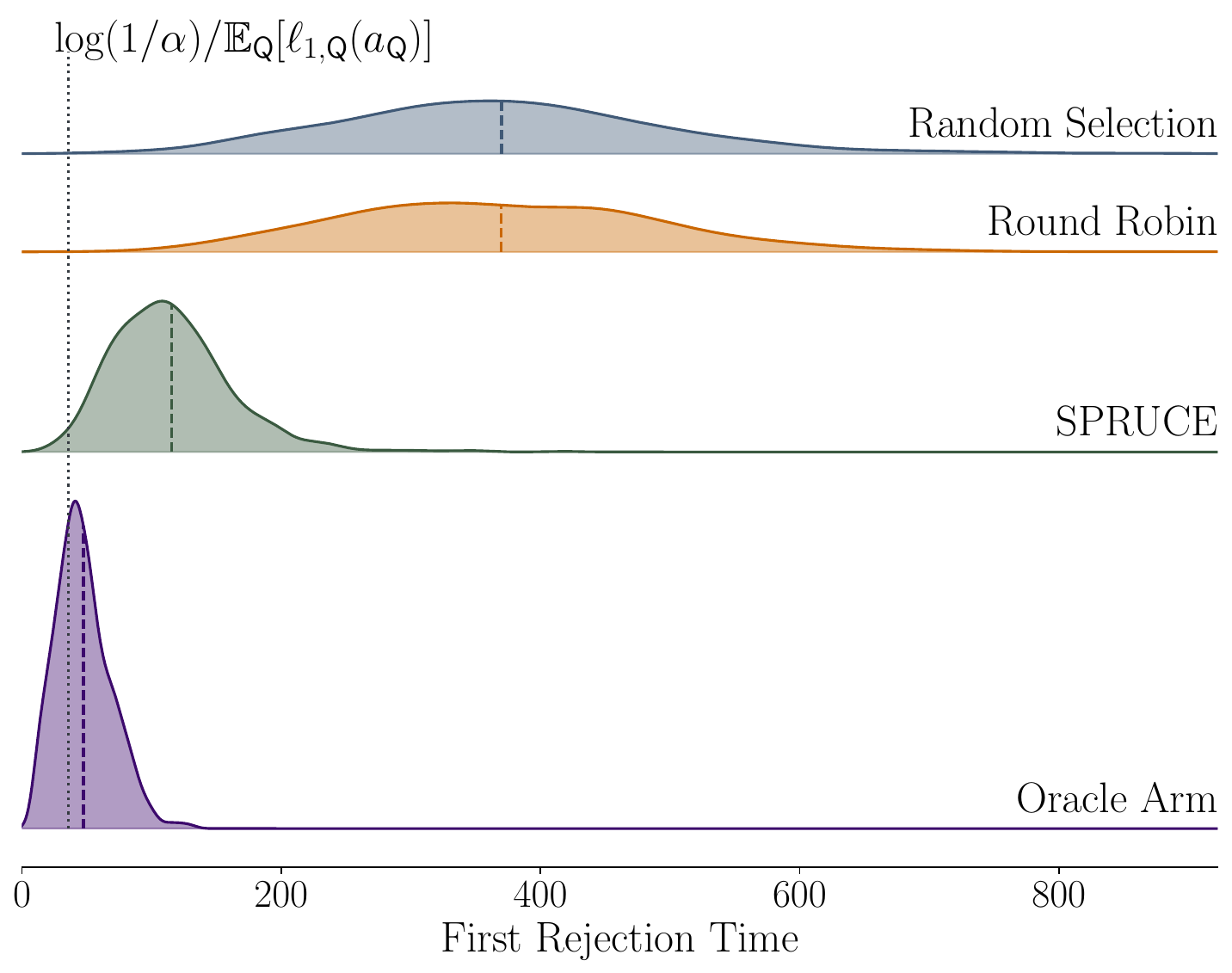}
    \end{subcaptionblock}
    \hfill
    \begin{subcaptionblock}{0.485\textwidth}
        \centering 
        \includegraphics[width=\linewidth]{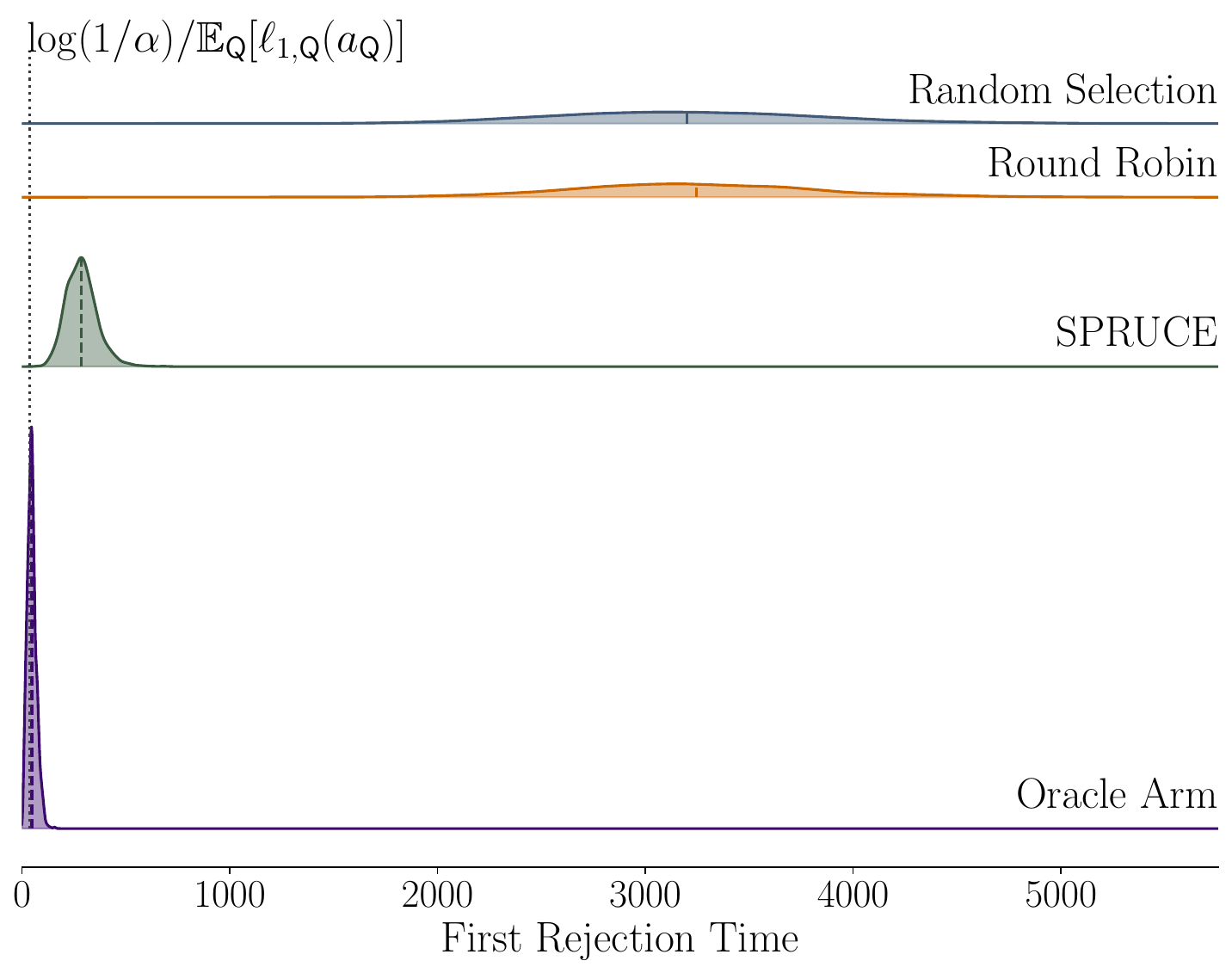}
    \end{subcaptionblock}
    \caption{Distribution of stopping times when $\alpha = 0.001$ for the one-sided bounded mean testing problem from \cref{example:one-sided mean testing} under ``easy'' (left) and ``hard'' (right) data generating processes. In addition to \SPRUCE{}, we evaluate three algorithms whose description can be found the caption of \cref{fig:bounded-mean-testing}.
    We use the following shorthand for the log-increment under the optimal arm and its optimal portfolio: $\ell_{1,\Q}\Paren{\optarm} \coloneqq \logp{\optbbet(\optarm)^\trans \bE_1\Paren{\optarm}}$.}
    \label{fig:bounded-mean-testing-expected-stopping-time}
\end{figure}

In the following theorem, we establish two facts. First, the process $\infseqn{W_\nQ^\star}$ with oracle access to both $a_\Q$ and $\optbbet(a_\Q)$ achieves the lower bound in \cref{proposition:stopping-timelower-bound} in the ``high-confidence'' regime where $\alpha \to 0^+$. Second, the two processes generated according to \SPRUCE{} achieve the same bound in the same high-confidence regime.
\begin{theorem}[An upper bound on the expected rejection time in the high-confidence regime]
\label{thm:expected-stopping-time-upper-bound}
    Fix an alternative $\Qin$. Let $\infseqn{W_n}$ be one of the processes generated according to \SPRUCE{} (\cref{algorithm:spruce}) and define $\tau_\alpha^\SPRUCEtext \coloneqq \inf\{ n \in \NN : W_n \geq 1/\alpha \}$. Moreover, let $\infseqn{W_\nQ^\star}$ be the process given by
    \begin{equation}
        W_\nQ^\star \coloneqq \prod_{i=1}^n\optbbet(a_\Q)^\trans \bE_i(a_\Q)
    \end{equation}
    and define $\tau_{\alpha,\Q}^\star \coloneqq \inf \{ n \in \NN : W_\nQ^\star \geq 1/\alpha \}$.  Then we have
\begin{equation}\label{eq:expected stopping time spruce and oracle}
        \lim_\alphato \frac{\EE_\Q[\tau_\alpha^\SPRUCEtext{}]}{\log(1/\alpha)} = \lim_\alphato \frac{\EE_\Q[\tau_{\alpha, \Q}^\star]}{\log(1/\alpha)} = \left ( \max_{(a, \bbet) \in
    \Acal \times \simplex} \EE_\Q \Brac{\log \left (\bbet^\trans \bE_1(a) \right )}\right
    )^{-1} 
    \end{equation}
\end{theorem}

The proof of \cref{thm:expected-stopping-time-upper-bound} can be found in 
\cref{sec:proof-expected-stopping-time}. Note that \cref{thm:expected-stopping-time-upper-bound} follows from nonasymptotic upper bounds on $\EE_\Q[\tau_\alpha^\SPRUCEtext{}]$ and $\EE_\Q[\tau_{\alpha, \Q}^\star]$ that hold for any $\alpha \in (0, 1)$ and reduce to \eqref{eq:expected stopping time spruce and oracle} when $\alphato$. \cref{thm:expected-stopping-time-upper-bound} demonstrates that in the high-confidence regime, the expected time to rejection $\EE_\Q[\tau_\alpha^\SPRUCEtext]$ of \SPRUCE{} will match that of an oracle $e$-process that has access to $(a_\Q,\optbbet(a_\Q))$ from the outset. We depict the empirical distribution of \SPRUCE's rejection times in \cref{fig:bounded-mean-testing-expected-stopping-time}, and compare it to the empirical distributions of two other allocation algorithms as well as an algorithm that has oracle access to the optimal arm but employs the same regret-based test statistic as \SPRUCE.

In the next section we present the technical ingredients needed to arrive at the main results of \cref{section:multi-armed,section:rejection times}. In particular, we present several results on the derivation of logarithmic allocation regret.

\section{Proof Ingredients for Sublinear Allocation Regret}
\label{section:ingredients}

Recall the definition of portfolio-oracle allocation regret from \cref{definition:allocation regret}:
\begin{equation}
    \Rcal_\nQ^\MAB \coloneqq n \EE_\Q \left [ \log(\optbbet(a_\Q) ^\trans \bE(a_\Q) \right ] - \sum_{i=1}^n \EE_\Q \left [ \log( \optbbet(A_i)^\trans \bE(A_i) \right ],
\end{equation}
and recall from \cref{proposition:oracle log-optimality abstract} that when used in conjunction with sublinear arm-wise portfolio regret algorithms played on each arm, it suffices to find an allocation rule so that $\Rcal_\nQ^\MAB$ is sublinear in $n$.
In such a pursuit, one might be tempted to look to
standard multi-armed bandit analyses such as those for upper confidence
bound (UCB) algorithms \citep{lai1985asymptotically}, but
in our setting we lack two key ingredients typically present in such
analyses: ($i$) we do not observe \iidtext{} copies of
$\log(\optbbet(a)^\trans \bE_n(a))$ for any arm $a \in \Acal$ since
$\optbbet(a)$ is always unknown; ($ii$) even if such \iidtext{} copies were
available, we are not willing to assume that $\log\paren{\optbbet(a)^\trans \bE_n(A_n)}$
is $\sigma$-sub-Gaussian for some \emph{a priori} known $\sigma \in (0, \infty)$ given that the log-increments $\log\paren{\optbbet(a)^\trans \bE_n(a)}$ could take arbitrarily large negative values (in principle). Nevertheless, it is still possible to derive sufficiently sharp confidence intervals for the means $\EE_\Q[\log(\optbbet(a)^\trans \bE_1(a)]$ for $a \in \Acal$ using only
the \iidtext{} random variables $\infseqn{\bE_n(a)}$. The approach relies on incorporating a pathwise regret bound into the width of a confidence interval and demonstrating that the random variable $\log\paren{\optbbet(a)^\trans \bE_1(a)}$ is sub-\emph{exponential}, from which UCB-type analyses can be amended to still satisfy logarithmic allocation regret (see e.g., \citep{jia2021multi}). The fact that this random variable is sub-exponential is not immediately obvious and its proof relies on some properties of so-called \emph{numeraire portfolios} \citep{long1990numeraire,karatzas2007numeraire,larsson2025numeraire} which we will introduce as needed.

We carry out the analysis through some lemmas and propositions, ultimately culminating in \cref{theorem:expected suboptimal pulls} which demonstrates that the allocation regret $\Rcal_{\nQ}^\MAB$ is in fact \emph{logarithmic} under \SPRUCE{}. We begin by demonstrating that for each $a \in \Acal$, $\log(\optbbet(a)^\trans \bE_n(a))$ has a finite moment generating function.

\begin{lemma}[Log-increments have finite moment generating functions under optimal portfolios]
\label{lemma:sub-exponential-log-wealth-increments}
Fix $\Q \in \Qcal$ and let $K=1$. Let $\infseqn{\bE_n}$ be ($d+1$)-vectors of $\Pe$-values that satisfy 
\cref{assumption:class of eprocesses} with the constant $b > 1$.
Let $\optbbet \coloneqq \argmax_{\bbet \in \simplex} \E_{\Q}\Brac{
\logp{\bbet^\trans \bE_1}}$.
Then the logarithm of the $\optbbet$-weighted tuples is sub-exponential, meaning
that
\begin{equation}\label{eq:q-a.s. upper bound assumption}
    \forall \theta \in [-1, 1], \quad \E_{\Q} \Brac{\exps{\theta \Paren{
    \logp{\optbbet^\trans \bE_1} - \E_{\Q}\Brac{\logp{\optbbet^\trans \bE_1}}}}}
    \leq b \mper
\end{equation}
\end{lemma}
The proof of \cref{lemma:sub-exponential-log-wealth-increments} can be found in
\cref{proof:sub-exponential-log-wealth-increments}. We remark that the proof
crucially relies on the fact that $\optbbet$ is the optimal portfolio under
$\Q$, and the inequality \eqref{eq:q-a.s. upper bound assumption} may not hold
if some other sub-optimal portfolio $\bbet \in \simplex$ were considered
instead. The lemma is used to deduce the following tail bound on the
deviation between the average of the log-increments under the best in hindsight
portfolio---i.e., $\max_{\bbet \in \simplex} \frac{1}{n} \sum_{i=1}^n \logp{\bbet^\trans
\bE_i}$---and the expected log-increment under the log-optimal portfolio.

\begin{proposition}\label{proposition:regret based confidence interval}Fix $\Qin$ and let $K=1$.
  Let $\infseqn{\bE_n}$ be $(d+1)$-vectors of $\Pe$-values satisfying
  \cref{assumption:class of eprocesses} with the constant $b > 1$, and
  suppose that there exists an algorithm for selecting the portfolios
  $\infseqn{\bbet_n}$ which has a pathwise portfolio regret bound of
  $\Rcal_n^\Port \leq R_n^\Port$ for each $n \in \NN$. Then for any $\alpha \in
  (0, 1)$,
  \begin{equation}
    \PP_\Q \left ( \max_{\bbet \in \simplex}\frac{1}{n} \sum_{i=1}^n \log \Paren{\bbet^\trans \bE_i} - \EE_\Q\Brac{\log\Paren{\optbbet^\trans \bE_1}} \geq \sqrt{\frac{8b\log(1/\alpha)}{n}} + \frac{5 \log (1/\alpha)}{n} + \frac{R_n^\Port}{n}\right ) \leq \alpha.
  \end{equation}
  Moreover,
  \begin{equation}
    \PP_\Q \left ( \EE_\Q\Brac{\log\Paren{\optbbet^\trans \bE_1}} - \max_{\bbet
    \in \simplex} \frac{1}{n} \sum_{i=1}^n \log \Paren{\bbet^\trans \bE_i} \geq \sqrt{\frac{8b\log(1/\alpha)}{n}} + \frac{4 \log (1/\alpha)}{n}\right ) \leq \alpha.
  \end{equation}
\end{proposition}
The proof can be found in
\cref{proof:regret based confidence interval}. 
The argument proceeds by first relying on
\cref{lemma:sub-exponential-log-wealth-increments} to establish that the
log-increments have finite moment generating functions in the neighborhood
$\theta \in [-1,1]$. It then follows that for any integer $p \in \NN$, the
$p^\tth$ absolute moment is uniformly bounded by $bp!$
(\cref{lemma:bound-on-pth-moment-of-log-wealth-increments}), and hence its
moment generating function can be bounded by that of a Gaussian random variable
but only in the neighborhood $\theta \in [-1/2, 1/2]$ (\cref{lemma:range of
subgaussianity}). Using the Chernoff method combined with a pathwise bound on
the portfolio regret yields \cref{proposition:regret based confidence interval}.
The relationship between finite moment generating functions and sub-exponential
tail bounds is well-known; e.g., it is summarized in
\citet[2.7.1]{vershynin2018high}. However, the constants therein are implicit so
we carry out the proofs to derive explicit confidence bounds that can be used in
\SPRUCE{}.

We remark that one does not need to \emph{run} the portfolio selection algorithm
or compute $\infseqn{\bbet_n}$---the existence of an algorithm with a regret
bound of $R_n^\Port$ suffices. For the sake of concreteness, one can take
$R_n^\Port = d \log (n + 1) / 2 + \log (2)$ as in the universal portfolio
algorithm of \citet{cover1996universal} but in some special cases, such as when
$d = 1$, sharper regret bounds exist \citep{orabona2023tight}.
Concentration inequalities for optimal growth rates such as
\cref{proposition:regret based confidence interval} that do not rely on
knowledge of $\optbbet$ are not easily found in the literature. To the best of
our knowledge, the only other result similar in spirit are inequalities of
\citet[Appendix A.2]{agrawal2025stopping}.

We previously alluded to the fact that the setting we consider does not permit
some of the usual assumptions made for UCB-type algorithms (e.g.,
sub-Gaussianity or being in a parametric family
\citep{lai1985asymptotically,cappe2013kullback,lattimore2020bandit}).
Nevertheless, a key property that several of the UCB-based regret analyses
exploit is the fact that the width of the UCB scales as $\sqrt{\log (1/\alpha )
/ n}$, where $\alpha \in (0, 1)$ and $n$ are the miscoverage probability and the
sample size, respectively. The concentration inequalities of
\cref{proposition:regret based confidence interval} satisfy this condition
modulo some extra $\log(1/\alpha) / n$ and $R_n / n$ terms, both of which turn
out to affect the downstream analysis in benign ways. The following theorem uses
the concentration inequalities of \cref{proposition:regret based confidence
interval} to arrive at a bound on the portfolio-oracle allocation regret.

\begin{theorem}\label{theorem:expected suboptimal pulls}
    Fix a distribution $\Q \in \Qcal$. Let $n \in \NN$, $\gamma > 2$, and $\zeta >
    0$. Fix an arm $a\in \Acal$ and define the oracle suboptimality gap for arm $a
    \in \cA$ under the distribution $\Q$ as 
    \begin{equation}
     \Delta_{a,\Q} \coloneqq \E_{\Q}\Brac{\logp{\optbbet(\optarm)^\trans
    \bE(\optarm)}} - \E_{\Q}\Brac{\logp{\optbbet(a)^\trans \bE(a)}}.  
    \end{equation}
    Suppose that we are running \SPRUCE{} to select the arms and
    construct the test statistic $W_n$, and that \cref{assumption:class of eprocesses} 
    holds with the constant $b > 1$. Define $R_n^\CO \coloneqq d \log (n + 1) / 2 + \log (2)$.  Then the expected number of arm pulls $\EE_\Q[N_a(n)]$ is bounded by
    \begin{equation}
      \EE_\Q \left [ N_a(n) \right ] \leq 1 + \max \left \{ \frac{72 b \gamma
      \log(\zeta n + 1)}{\Delta_{a,\Q}^2}, \frac{15\gamma \log(\zeta n +
      1)}{\Delta_{a,\Q}}, \frac{3R_{n-1}^\CO}{\Delta_{a,\Q}} \right \} + \frac{4
        \zeta^{-\gamma}}{\gamma - 2}.
    \end{equation}
    Consequently, the portfolio-oracle allocation regret $\Rcal_n^{\MAB}$ can
    be bounded in the large-$n$, small-$\Delta_{a,\Q}$ regime as
    \begin{equation}
        \Rcal_{\nQ}^{\MAB} = \sum_{a\in\Acal} \Delta_{a,\Q} \EE_\Q \Brac{N_a(n)} = \cO
        \Paren{\sum_{a \in \Acal} \frac{\log(n)}{\Delta_{a,\Q}}}.
    \end{equation}
\end{theorem}
The proof of \cref{theorem:expected suboptimal pulls} can be found in
\cref{proof:expected suboptimal pulls}.
Combined with an arm-wise sublinear portfolio regret algorithm and invoking \cref{proposition:oracle log-optimality abstract}, we see that \SPRUCE{} yields $e$-processes that are multi-armed log-optimal and enjoy oracle-like expected rejection time behavior.
In what follows, we demonstrate how the technology we have
developed thus far can be used to tackle the problem of testing for the existence of a
positive treatment effect, which was the example that we provided as motivation in the introduction.

\section{Testing for the Existence of a Treatment Effect}
\label{section:treatment effects}

Suppose that a pharmaceutical company is running a randomized
control trial to determine whether there exists a variation of a new treatment
that is significantly more effective than the control.
In this section we formalize this motivating example
and show that \SPRUCE{} provides one way to solve the pharmaceutical company's
question in an optimal manner.

Here, $\cA$ should be thought of as representing the set of possible treatment variations (e.g., different dosage levels). We adopt the potential outcomes 
framework~\citep{neyman1923applications,rubin1974estimating} and assume that each
experimental unit $n \in \N$ has $K+1$ potential outcomes
\begin{align}
    Y_n(0), Y_n(1), \dots, Y_n(K),
\end{align}
where $Y_n(0)$ is their potential outcome under the control group, which will always be used as a baseline to which treatments $1, \dots, K$ will be compared. Suppose that for every $n \in \NN$ and $a \in \Acal \cup \Set{0}$, $Y_n(a)$ takes values in the unit interval $[0, 1]$. 
For a treatment variation $a \in \Acal$ and a distribution $\P$, define the \emph{average treatment effect} $\ate_\P(a)$
as
\begin{equation}
  \ate_{\P}(a) \coloneqq \E_{\P}\Brac{Y(a) - Y(0)}  .
\end{equation}
We consider the problem of testing whether \emph{any} of the treatment variations $a \in \cA$ have an average treatment effect greater than a pre-determined threshold $\delta \in [-1, 1]$. In the case of testing for a positive treatment effect, set $\delta$ to zero. Formally, consider the global null $\cP^\brackdelta$ versus the
alternative $\cQ^\brackdelta$:
\begin{equation}
    \cP^\brackdelta = \Set{\P \svert \forall a \in \cA,~\ate_{\P}(a) \leq \testate} 
    \quad\mathrm{versus}\quad
    \cQ^\brackdelta = \Set{\P \svert \exists a \in \cA,~\ate_{\P}(a) > \testate} \mper
\label{eq:ate-testing-problem}
\end{equation}
In words, the null $\Pcal^\brackdelta$ contains distributions for which every treatment variation has a $\delta$-small treatment effect while the alternative contains those for which at least one treatment variation is $\delta$-large.

As the experiment is run sequentially, the statistician collects data according to a protocol that is a slight variation on \cref{algorithm:data collection} to account for randomization between treatment and control groups. We present this protocol in \cref{algorithm:randomized expt data collection}.

\begin{algorithm}[!htbp]
    \caption{Data collection in multi-armed randomized experiments}
    \label{algorithm:randomized expt data collection}
    The statistician chooses a propensity score $\pi \in (0, 1)$.
    
    Nature selects a distribution $\P \in \cP^\brackdelta \cup \cQ^\brackdelta$.

    \For{each time step $n = 1, \dots$}{
      \begin{enumerate}
        \item Nature samples a $(K+1)$-vector of outcomes:
            $\Paren{Y_n(0), Y_n(1), \dots, Y_n(K)} \sim \P$.
        \item The statistician chooses $A_n \in \{1, \dots, K\}$ based on the data gathered thus far.
        \item The statistician draws $Z_n \sim \text{Bernoulli}(\pi)$ 
        \item Subject $n$ is assigned to treatment $A_n$ if $Z_n = 1$ or the control group if $Z_n = 0$.
        \item The statistician observes $\Yobs_n(A_n) = Z_n Y_n(A_n) + (1-Z_n) Y_n(0)$.
      \end{enumerate}
    }
\end{algorithm}

The primary difference between the protocols in \cref{algorithm:data
collection,algorithm:randomized expt data collection} is that while the
statistician may choose to inspect promising treatment variations $1,\dots, K$
adaptively at each step, the subject always has a marginal probability of
$1-\pi$ of being assigned to the control group (see Steps 3 and 4). 
One could use a sequence of adaptively chosen propensity scores $\infseqn{\pi_n}$
without affecting the type-I error guarantees, but we fix the propensity scores
to $\pi \in (0,1)$ for the optimality guarantees to come.

The observation $\Yobs_n(A_n)$ taking the value of $Z_nY_n(A_n) + (1-Z_n) Y_n(0)$ is implicitly imposing an assumption of no interference across units. This assumption is commonly known as stable unit treatment value assumption (SUTVA) but is also referred to as ``consistency.''
No interference stipulates that the observed outcome of unit $n$ only depends on 
the treatment assignment they received, and not on the treatment assignment
any other unit received. Indeed, this assumption along two others are sufficient for \emph{identification} of the average treatment effects $\psi_\P(a); a \in \Acal$. Here, the term ``identification'' simply means that the counterfactual functional $\psi_\P(a)$ can be written in terms of a functional of the observable data distribution; informally, $\psi_\P(1), \dots, \psi_\P(K)$ can be estimated from the data arising in \cref{algorithm:randomized expt data collection}. While these assumptions are ubiquitous throughout the causal inference literature, we state them formally for completeness.

\begin{assumption}\label{assumption:causal identification}
    Suppose that the data arising from \cref{algorithm:randomized expt data collection} satisfies the following three conditions:
    \begin{enumerate}
        \item \emph{SUTVA}: For each $n \in \NN$, $\Yobs_n(A_n) = Z_n Y_n(A_n) + (1-Z_n) Y_n(0)$,
        \item \emph{No confounding}: For each $n \in \NN$, $(Y_n(0), \dots, Y_n(K), A_n) \independent Z_n$, and
        \item \emph{Positivity}: $\pi \coloneqq \PP_{\Bern}(Z_1 = 1) \in (0, 1)$.
    \end{enumerate}
\end{assumption}
We remark that the second and third conditions are trivially satisfied by the fact that $\infseqn{Z_n}$ are independent Bernoullis with success probability $\pi \in (0, 1)$. Under \cref{assumption:causal identification}, it holds that for any $n \in \NN$,
\begin{equation}
    \ate_{\P}(a) = \EE_{\Prct}[\Yobs_n(a) \svert Z_n = 1] - \EE_{\Prct}[\Yobs_n(a) \svert Z_n = 0] \mcom
\end{equation}
where we use $\Prct$ with the additional ``RCT'' subscript to indicate that the expectation is taken over the randomness in both the treatment assignments $\infseqn{Z_n}$ and $\P$.
To estimate this quantity, we consider the Horvitz-Thompson estimator \citep{horvitz1952generalization} of the individual treatment effect given by
\begin{equation}
\label{eq:horvitz-thompson-estimator}
    \hate_n(a) \coloneqq \Yobs_n(a)\Paren{\frac{Z_n}{\pi} - \frac{1-Z_n}
    {1-\pi}} \mper
\end{equation}
Importantly, the Horvitz-Thompson estimator from \eqref{eq:horvitz-thompson-estimator}
is an unbiased estimator for the average treatment effect.
\begin{fact}
\label{fact:horvitz-thompson-unbiased-estimation}
    Fix a distribution $\P \in \Pcal^\brackdelta \cup \Qcal^\brackdelta$. Under \cref{assumption:causal identification}, we have that for any $n \in \NN$ and any $a \in \Acal$
    \begin{equation}
        \EE_{\Prct}\left [\hate_n(a) \right ] = \ate_{\P}(a).
    \end{equation}
\end{fact}
Furthermore, we note that under the assumption that potential outcomes are bounded, the Horvitz-Thompson estimator from 
\eqref{eq:horvitz-thompson-estimator} is almost surely bounded as
$\hate_n(a) \in [-1/(1-\pi), 1/\pi]$.
We apply the transformation $x \mapsto \pi(1 + x (1-\pi))$
to $(\hate_n(a), \delta)$ so that the latter two variables lie in the unit interval:
\begin{equation}
    \underline{\hate}_n(a) \coloneqq \pi(1 + \hate_n(a) (1-\pi)) \quad\mathrm{and}\quad
    \underline{\delta} \coloneqq \pi(1 + \delta (1-\pi)) \mper
\end{equation}
Given that $\underline{\hate}_n(a)$ is a $[0, 1]$-bounded random variable with $\P$-mean $\pi(1 + \ate_\P(a)(1-\pi))$, we can employ techniques from prior work on bounded mean testing \citep{hendriks2021test,waudby2024estimating,orabona2023tight,waudby2024anytime} (see also \cref{example:one-sided mean testing}). Specifically, take $d = 1$ and for each $n \in \NN$, define the $\cP^\brackdelta$-$e$-values $\bE_n$ by
\begin{equation}
    \bE_n \coloneqq \left (1, \frac{\underline{\hate}_n(A_n)}{\underline{\delta}} \right )
\end{equation}
and let $\infseqn{W_n^\brackdelta}$ be the process given by
\begin{equation}
\label{eq:ate-test-statistic}
    W_n^\brackdelta = \prod_{a \in \cA} \exps{\max_{\bet \in [0,1]} 
        \sum_{i=1}^n \Ind{A_i = a}\logs{1-\lambda + \lambda
        \flatfrac{\underline{\hate}_i(A_i)}{\underline{\delta}}}
    - R_{N_a(n)}^\CO} \mcom
\end{equation}
where $R_{n}^\CO \coloneqq \logp{n + 1} / 2 + \logp{2}$ for $n \in \NN$ as in the definition of \SPRUCE{}.
Importantly, the test statistic from \eqref{eq:ate-test-statistic} forms a
$\cP^\brackdelta$-$e$-process with  growth-rate and expected-rejection-time guarantees that we summarize here.

\begin{figure}
    \centering
    \begin{subcaptionblock}{0.485\textwidth}
        \centering 
        \includegraphics[width=\linewidth]{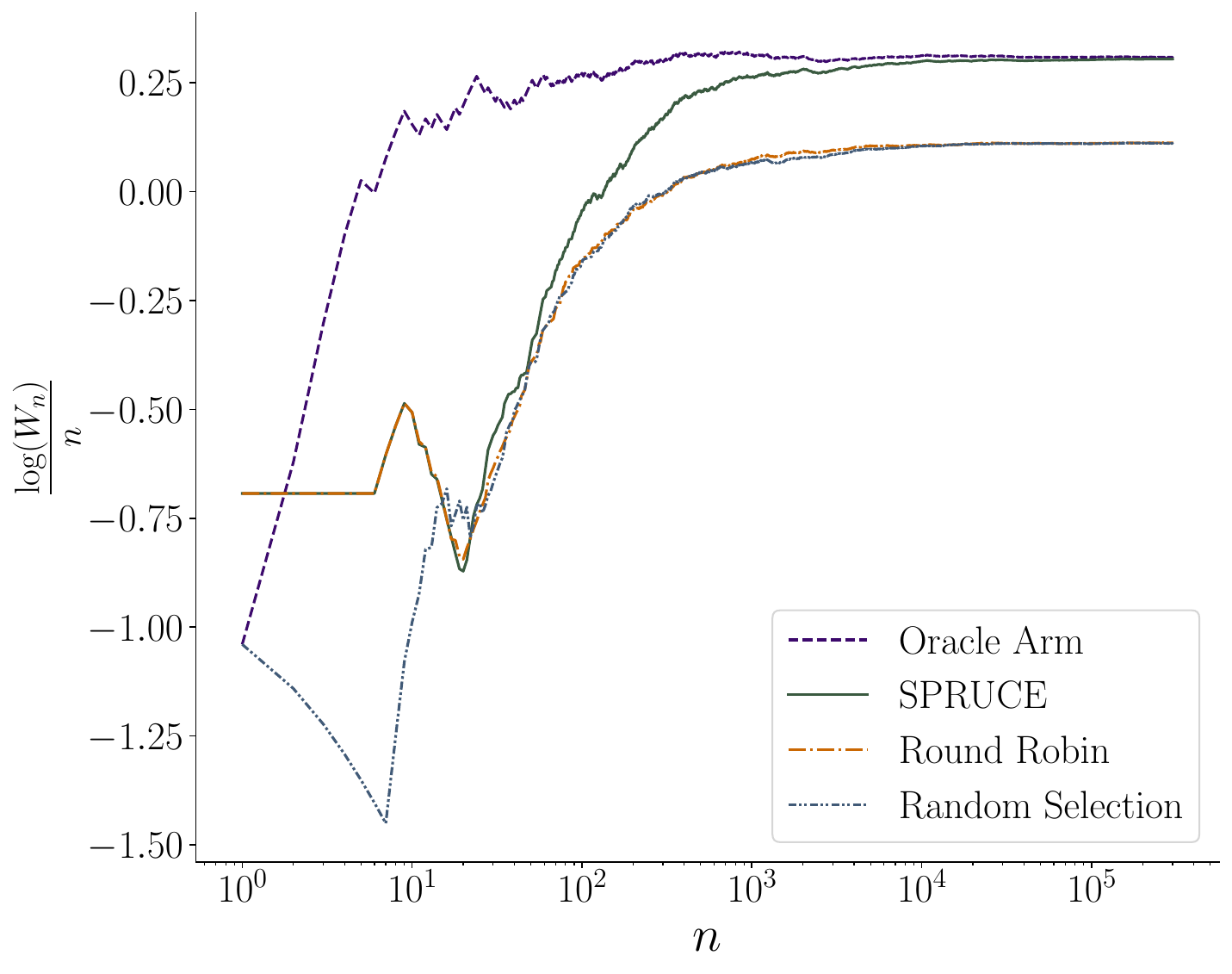}
    \end{subcaptionblock}
    \hfill
    \begin{subcaptionblock}{0.485\textwidth}
        \centering 
        \includegraphics[width=\linewidth]{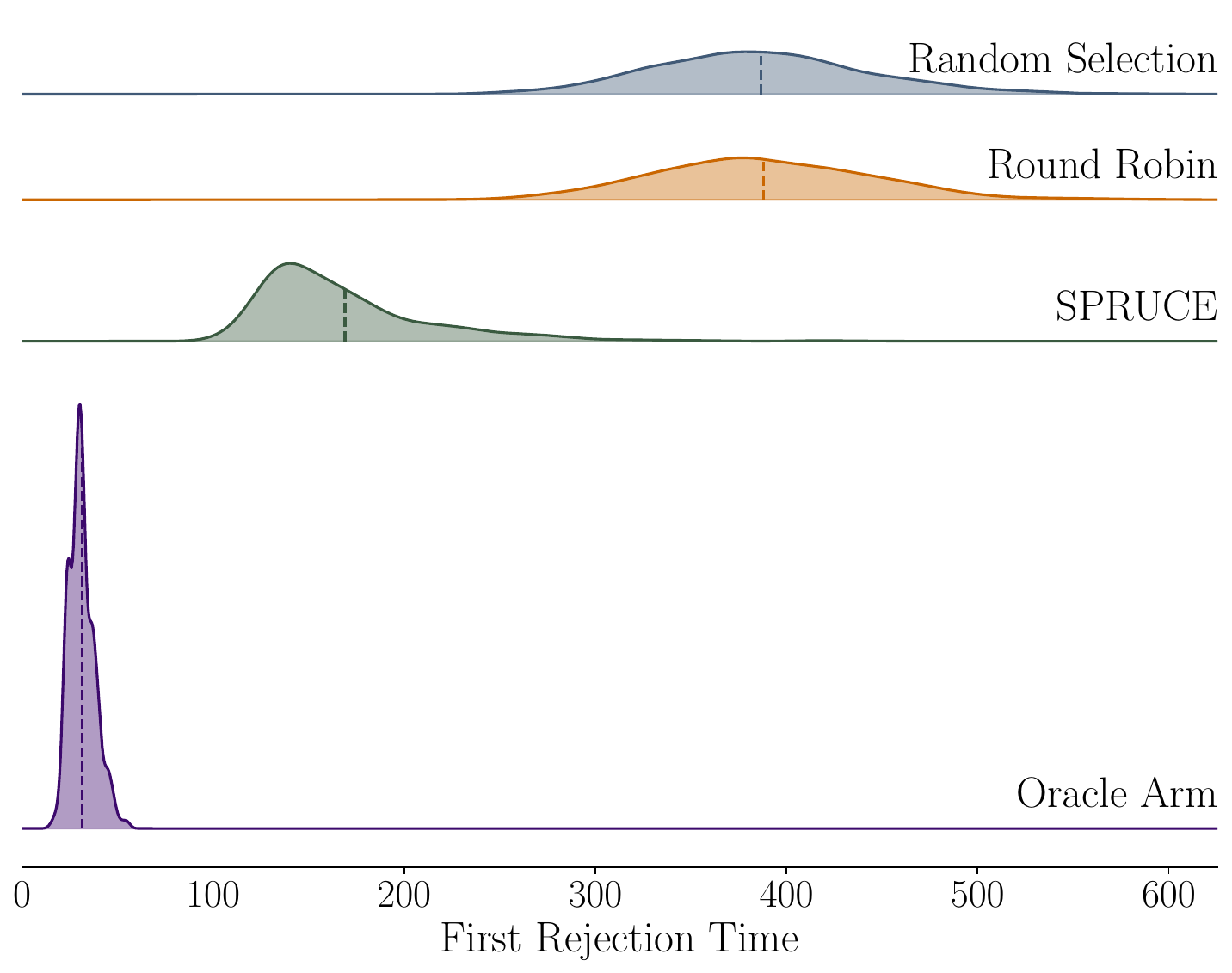}
    \end{subcaptionblock}
    \caption{Empirical growth rates (left) and distribution of stopping times when $\alpha=0.001$ (right)
    for the average treatment effect testing problem from \eqref{eq:ate-testing-problem}
    for four different algorithms. \emph{Round Robin} pulls the arms one-by-one until
    all of them have been selected, and starts the process over again. 
    \emph{Random Selection} samples uniformly at random the arm to be played in 
    round $n$.
    All four algorithms compute their test statistic following the form given
    in \cref{eq:ate-test-statistic}; that is, they only differ in the way they
    select the arm to pull (i.e., treatment variation to test) at each time step.}
    \label{fig:ate-testing}
\end{figure}

\begin{corollary}
\label{corollary:ate-e-process}
    Let $\infseqn{W_n^\brackdelta}$ be the process given by
    \eqref{eq:ate-test-statistic}. Then $\infseqn{W_n^\brackdelta}$ is a $\cP^\brackdelta$-$e$-process so that for any $\Q \in
    \cQ^\brackdelta$,
    \begin{equation}
        \lim_{n \to \infty} \frac{1}{n} \log \left ( W_n^\brackdelta \right ) = \max_{(a, \bet) \in
        \cA \times [0,1]} \E_{\Qrct}\Brac{\logp{1 - \lambda + \lambda
        \flatfrac{\underline{\hate}_1(a)}{\underline{\delta}}}}
        \quad \Qrct \text{-almost surely}\mper
    \end{equation}
    Furthermore, for any $\alpha \in (0, 1)$, define the first rejection time $\tau_\alpha \coloneqq \inf \{ n \in \NN : W_n^\brackdelta \geq 1/\alpha \}$. Then,
    \begin{equation}
        \lim_{\alpha \to 0^+} \frac{\E_{\Qrct}\Brac{\tau_{\alpha}}}{\logp{1/\alpha}}
         = \Paren{\max_{(a, \bet) \in \cA \times [0,1]} \E_{\Qrct}\Brac{\logp{1 - \lambda + \lambda
         \flatfrac{\underline{\hate}_1(a)}{\underline{\delta}}}}}^{-1} .
    \end{equation}
\end{corollary}
\cref{corollary:ate-e-process} follows immediately from \cref{theorem:log-optimality of spruce,thm:expected-stopping-time-upper-bound}. Moreover, it follows that these growth rates and expected rejection times are unimprovable in the sense that no other rule for selecting treatment arms to test $\infseqn{A_n}$ and choosing portfolios $\infseqn{\lambda_n}$ could lead to larger growth rates or smaller expected rejection times.

\cref{fig:ate-testing} illustrates empirically the conclusions from
\cref{corollary:ate-e-process}. It depicts how the empirical growth rate of \SPRUCE~converges to that of an $e$-process with oracle access to the optimal arm and which uses the regret-based test statistic from \eqref{eq:ate-test-statistic}. Moreover, it can be seen that the distribution of rejection times under \SPRUCE{} tends to include smaller values than those of Round Robin or Random Selection. Indeed, the average time to rejection under \SPRUCE{} is roughly 44\% that of the processes employing Round Robin or Random Selection.

\section{Conclusions}
In this work we presented a generalization of sequential hypothesis testing by betting wherein
the statistician must choose to draw data from one of many possible distributions (arms) and then must proceed to select a portfolio (via a ``betting strategy'') to play on that arm. We considered global nulls in which each arm satisfies some property of interest, and it is desired to (quickly) accumulate evidence that \emph{at least one} arm does not satisfy that property. We showed that
this multi-armed data collection protocol does not add any extra complexity when it comes to
time-uniform type-I error control, but argued that designing powerful
tests under the alternative motivates new algorithms and analyses. We
introduced a notion of multi-armed log-optimality for the growth rate of e-processes under the alternative hypotheses, and derived an algorithm that
achieves this new and strong notion of optimality. We analyzed the expected
time to rejection under that same algorithm, leading to matching lower and upper bounds in the high-confidence regime. In summary, we observed that the (achievable) optimal growth rate
and expected times to rejection depend on the same quantity: the largest expected logarithmic increment over all arms and all portfolios. Along the way, we derived some new concentration inequalities for these maximal expected logarithmic increments that may be of independent interest. 

We anticipate this work having several
applications beyond testing for the existence of treatment effects. For example, one may consider testing properties characterizing aspects of the performance of large language models, where the different ``arms'' could be prompts or candidate models themselves. Alternatively, the problem of quantum state certification (and related problems) \citep{buadescu2019quantum,martinez2021quantum,zecchin2025quantum} can be framed in terms of a sequential hypothesis testing problem where partial information naturally arises; this partial information has not yet been exploited in the sense of the present paper.
It may also be worth exploring how side 
information (e.g., covariates or features) can be used to better choose which arms to pull or portfolios to construct at each time step. 
We intend to pursue these directions in followup work.

\subsection*{Acknowledgments}
The authors would like to thank Sivaraman Balakrishnan, Ron Boger, Avi Feller,
Paula Gradu, Keegan Harris, Christian Ikeokwu, Jivat Neet Kaur, Aaditya Ramdas,
David Wu, and Tijana Zrnic for insightful conversations.
Funded by the European Union (ERC-2022-SYG-OCEAN-101071601).
Views and opinions expressed are however those of the author(s) only and do not
necessarily reflect those of the European Union or the European Research Council
Executive Agency. Neither the European Union nor the granting authority can be
held responsible for them.

\bibliographystyle{plainnat}
\bibliography{refs}

\newpage
\appendix
\section{Proofs of the Main Results}
In what follows we present the proofs of our main results as well as our auxiliary results and their proofs. Nevertheless, to simplify the exposition of certain proofs we will be employing some shorthands for the log-increments. We will denote the log-increment corresponding to round $n \in \N$ under any arm $a \in \cA$ as:
\begin{equation}
    \ell_n(a) \coloneqq \logp{\bbet_n(a)^\trans \bE_n(a)} \mper
\end{equation}
Following the above notation, we will further denote the log-increments under any arm $a \in \cA$ and under the optimal portfolio for that arm as:
\begin{equation}
    \ell_{n,\Q}(a) \coloneqq \logp{\optbbet(a)^\trans \bE_n(a)} \mcom
\end{equation}
for any $n \in \N$. Whenever we suppress the dependency on the arm $a \in \cA$, the proof should be thought of as holding for the single-arm ($K = 1$) setting. Having defined these shorthands, we now proceed to state the proofs of our main results.

\subsection{Proof of \cref*{proposition:validity}}\label{proof:validity}
\begin{proof}[Proof of \cref{proposition:validity}]
    We will show that for each $n \in \NN$,
    $\barW_n$ is nonnegative almost surely, that it forms a supermartingale, and that its mean is upper bounded by one. Indeed, nonnegativity follows by the assumption that $f_n(\cdot)$ takes values in $[0, \infty)$ for each $n \in \NN$. To show that $\barW$ is a $\Pcal$-supermartingale, observe that for any $n \in \NN$ and $\Pin$,
    \begin{equation}
        \EE_\P[\barW_n \mid \Hcal_{n-1}] = \barW_{n-1} \EE_\P [f_n(Y_n(A_n)) \mid \Hcal_{n-1}] \leq \barW_{n-1}\quad \text{$\P$-almost surely,}
    \end{equation}
    where the inequality follows by assumption. Instantiating the above for $n = 1$, we have that $\sup_\Pin \EE_\P[\bar W_1] \leq 1$. This completes the proof.
\end{proof}

\subsection{Proof of \cref*{proposition:oracle log-optimality abstract}}\label{proof:oracle log-optimality abstract}
\begin{proof}
    Fix $\Qin$. Similar to the proof sketch, decompose the difference between the optimal portfolio under the optimal arm and the empirical growth rate: 
    \begin{align}
    \cR_n =\ &\EE_\Q \left [ \log \left ( \optbbet(\optarm)^\trans
    \bE(\optarm) \right ) \right ] - \frac{1}{n}\sum_{i=1}^n \EE_\Q \left [ \log \left (
    \optbbet(A_i)^\trans \bE_1(A_i) \right ) \right
            ]\label{eq:decomposition-ucb}\\ 
                &+ \frac{1}{n}\sum_{i=1}^n \EE_\Q \left [ \log
            \left ( \optbbet(A_i)^\trans \bE_i(A_i) \right ) \right ] -
            \frac{1}{n} \sum_{i=1}^n \sum_{a=1}^K \1 \{A_i = a \} \log \left (
            \optbbet(a)^\trans \bE_i(a) \right )\label{eq:decomposition-lil}\\
                &+ \frac{1}{n}\sum_{i=1}^n \sum_{a=1}^K \1
                \{A_i = a \} \log \left (
                \optbbet(a)^\trans \bE_i(a) \right )
                - \frac{1}{n} \sum_{a=1}^K \log(W_n(a))
                .
                \label{eq:decomposition-universal-portfolio}
    \end{align}
    By sublinearity of the allocation regret, \eqref{eq:decomposition-ucb} vanishes. Letting $\Rcal_n^\Port(a)$ be the arm-wise portfolio regret of arm $a \in \Acal$, we have that 
    \begin{equation}
        \eqref{eq:decomposition-universal-portfolio} \leq \frac{1}{n} \sum_{a=1}^K \Rcal^\Port_{n}(a) \leq \frac{1}{n} K \max\{ \Rcal^\Port_n(1) ,\dots, \Rcal^\Port_n(K) \}
    \end{equation}
    pathwise, and hence \eqref{eq:decomposition-universal-portfolio} is sublinear.
    By \cref{lemma:exponential concentration for bandit sequence - centered with independent copies}, we have that for any $\eps > 0$,
  \begin{align}
      &\PP_{\Q}\Paren{\frac{1}{n} \sum_{i=1}^n \Paren{\log \left ( \optbbet(A_i)^\trans \bE_i(A_i) \right ) -
        \E_{\Q}\Brac{\log \left ( \optbbet(A_i)^\trans \bE_i(A_i) \right )}} \geq \epsilon} \\
        &\quad\leq Kn \Paren{\exps{-\frac{\epsilon^2 n}{8bK^2}} + \exps{-\frac{\epsilon n}{4K}}}.
  \end{align}
  Applying the Borel-Cantelli lemma, \eqref{eq:decomposition-lil} vanishes $\Q$-almost surely, completing the proof.
\end{proof}

\subsection{Proof of \cref*{theorem:log-optimality of spruce}}\label{proof:log-optimality of spruce}
\begin{proof}[Proof of \cref{theorem:log-optimality of spruce}]
    We begin by deriving a time-uniform concentration inequality for the difference between the $\infseqn{W_n}$ and that of an oracle with access to $a_\Q$ and $\optbbet(a_\Q)$. This inequality will be integral to the proofs of $(i)$ and $(ii)$.
    
    \paragraph{A time-uniform concentration inequality.}
    Fix a $\Q \in \cQ$ and let $W^\star_\nQ \coloneqq \prod_{i = 1}^n \Paren{\optbbet(\optarm)^\trans
    \bE_i(\optarm)}$.
    Our goal is to provide an upper bound on the following probability
    for any $\epsilon > 0$,
    \begin{equation}
        \PP_{\Q}\Paren{\sup_{n \geq m}~ \Abs{\frac{1}{n} \logp{W_\nQ^\star} -
        \frac{1}{n}\logp{W_n}} \geq \epsilon} \mper
        \label{eq:proof-log-opt-target probability}
    \end{equation}
    Applying the triangle inequality, we obtain
    \begin{align}
        \eqref{eq:proof-log-opt-target probability} 
        \leq \underbrace{\PP_{\Q}\Paren{\sup_{n \geq m}~ \Abs{\frac{1}{n}
        \logp{W_\nQ^\star} - \E_{\Q}\Brac{\ell_{1,\Q}(\optarm)}} \geq
        \frac{\epsilon}{2}}}_{
        \eqref{eq:proof-log-opt-target-prob-decomp}.(1)} + 
        \underbrace{\PP_{\Q}\Paren{\sup_{n \geq m}~ \Abs{\E_{\Q}\Brac{\ell_{1,\Q}(\optarm)}
        - \frac{1}{n}\logp{W_n}} \geq \frac{\epsilon}{2}}}_{
        \eqref{eq:proof-log-opt-target-prob-decomp}.(2)}.
        \label{eq:proof-log-opt-target-prob-decomp}
    \end{align}
    We first focus on upper bounding  
    $\eqref{eq:proof-log-opt-target-prob-decomp}.(1)$. To do so, we apply a
    union bound and obtain
    \begin{align}
        \PP_{\Q}\Paren{\sup_{n \geq m}~ \Abs{\frac{1}{n}\logp{W_\nQ^\star} -
        \E_{\Q}\Brac{\ell_{1,\Q}(\optarm)}} \geq \frac{\epsilon}{2}}
        &\leq \sum_{n=m}^{\infty} \PP_{\Q}\Paren{\Abs{\frac{1}{n}
        \logp{W_\nQ^\star} - \E_{\Q}\Brac{\ell_{1,\Q}(\optarm)}} \geq
        \frac{\epsilon}{2}}.\label{eq:borel cantelli bound continued}
    \end{align}
    Now, by \cref{lemma:sub-exponential-log-wealth-increments} we have that 
    \begin{equation}
    \forall \theta \in [-1, 1], \quad \E_{\Q} \Brac{\exps{\theta \Paren{
    \logp{\optbbet^\trans \bE_1} - \E_{\Q}\Brac{\logp{\optbbet^\trans \bE_1}}}}}
    \leq b,
    \end{equation}
    which when used in conjunction with the fact that random variables with finite moment generating functions have sub-exponential tails (\cref{lem:concentration-sub-exponential-rvs}), the right-hand side of \eqref{eq:borel cantelli bound continued} can be bounded as
    \begin{equation}
        \sum_{n=m}^{\infty} \PP_{\Q}\Paren{\Abs{\frac{1}{n}
        \logp{W_\nQ^\star} - \E_{\Q}\Brac{\ell_{1,\Q}(\optarm)}} \geq
        \frac{\epsilon}{2}}
        \leq 2 \sum_{n=m}^{\infty} \Paren{\exps{-\frac{\epsilon^2 n}{32 b}} + \exps{-\frac{\epsilon n}{8}}}.
    \end{equation}
    We now turn our attention to upper bounding
    $\eqref{eq:proof-log-opt-target-prob-decomp}.(2)$. 
    To do so, let $R_\nQ^\MAB$ be an upper bound on $\Rcal_{\nQ}^\MAB$ for each $n \in \NN$ and appeal to \cref{lemma:exponential concentration for bandit sequence - optimal} to obtain the following upper bound:
     \begin{align}
       &\PP_{\Q}\Paren{\sup_{n \geq m}\Abs{\E_{\Q}\Brac{\ell_{1,\Q}(\optarm)} - \frac{1}{n}\logp{W_n}} \geq \frac{\epsilon}{2}}\\
       &\quad \leq
         \sum_{n=m}^\infty \1 \left \{ R_\nQ^\MAB \geq
        \epsilon n / 6 \right \} + 2K \sum_{n=m}^\infty n \left ( \exp \left \{ - \frac{\eps^2n }{288bK^2} \right \} + \exp \left \{ -\frac{\eps n }{24K} \right \} \right )\\
        &\quad + \sum_{n=m}^\infty \1  \left \{ \sum_{a=1}^K R^\Port_n(a) \geq \eps n / 6\right  \} + \sum_{n=m}^\infty \exp \left \{ -n\eps / 6 \right \}.
     \end{align} 
    Putting these steps together and consolidating some terms, we end up with the following inequality on \eqref{eq:proof-log-opt-target probability}
    for all $\Q \in \cQ$ and $\epsilon > 0$,
    \begin{align}
      \PP_{\Q}\Paren{\sup_{n \geq m} \Abs{\frac{1}{n} \logp{W_\nQ^\star} -
        \frac{1}{n}\logp{W_n}} \geq \epsilon} &\leq 5 K \sum_{n=m}^{\infty} n \Paren{\exps{-\frac{\epsilon^2 n}{288 b K^2}} + \exps{-\frac{\epsilon n}{24K}}}\\
      &\quad +  \sum_{n=m}^\infty \1 \left \{ R_\nQ^\MAB \geq
        \frac{\epsilon n}{6} \right \} + \sum_{n=m}^\infty \1  \left \{ \sum_{a=1}^K R^\Port_n(a) \geq \frac{\eps n}{6}\right  \}.
    \end{align}
    We will use this upper bound to prove properties
    $(i)$ and $(ii)$.

    \paragraph{Proof of property $(i)$.} We start by recalling the well known 
    fact that for a stochastic process $\infseqn{X_n}$ and a distribution $\P$,
    \begin{equation}
        \PP_{\P}\Paren{\lim_{n \to \infty} \Abs{X_n} = 0} = 1 \quad\text{if and
        only if}\quad \forall \epsilon > 0\mcom~ \lim_{m \to
        \infty}\PP_{\P}\Paren{\sup_{n \geq m} \Abs{X_n} \geq \epsilon} = 0 \mper
    \end{equation}
    Indeed, as we take $m \to \infty$ the previously derived upper bound on  
    \eqref{eq:proof-log-opt-target probability} vanishes, implying that
    \begin{equation}
        \frac{1}{n} \logp{W_n / W_\nQ^\star} \to 0 \quad\text{with
        $\Q$-probability one} \mcom
    \end{equation}
    completing the proof of property $(i)$.

    \paragraph{Proof of property $(ii)$.} By the law of large numbers we have
    that with $\Q$-probability one,
    \begin{equation}
        \frac{1}{n} \logp{W_\nQ^\star} \to \E_{\Q}\Brac{\logp{\optbbet(\optarm)^\trans 
        \bE(\optarm)}}\equiv \E_{\Q}\Brac{\ell_{1,\Q}(\optarm)} \mper
    \end{equation}
    Appealing to property $(i)$ it follows that
    \begin{equation}
        \frac{1}{n} \logp{W_n} = \frac{1}{n} \logp{W_n / W_\nQ^\star} +
        \frac{1}{n} W_n^\star \to \E_{\Q}\Brac{\ell_{1,\Q}(\optarm)} \mcom
    \end{equation}
    $\Q$-almost surely for every $\Q \in \cQ$.

    \paragraph{Proof that $W$ is multi-armed $\Qcal$-log-optimal.} Let $\widetilde W$ be any process of the form 
    \begin{equation}
        \widetilde W_n := \prod_{i=1}^n \widetilde \bbet_i^\trans \bE_i(\widetilde A_i),
    \end{equation}
    for $\Hcal$-predictable $\infseqn{\tbbet_n}$ and $\infseqn{\widetilde A_n}$.
    Consider the following limit inferior:
    \begin{align}
        \liminf_{n \to \infty} \frac{1}{n} \Paren{\logp{W_n} - \logp{\tW_n}}
        &= \liminf_{n \to \infty} \frac{1}{n} \Paren{\logp{W_n / W_\nQ^\star}
        + \logp{W_\nQ^\star / \tW_n}} \\
        &\geq \underbrace{\liminf_{n \to \infty} \frac{1}{n} \logp{W_n /
        W_\nQ^\star}}_{\eqref{eq:proof-log-opt-prop-i-decomp-lower-bound}.(1)}
        + \underbrace{\liminf_{n \to \infty} \frac{1}{n} \logp{W_\nQ^\star /
        \tW_n}}_{\eqref{eq:proof-log-opt-prop-i-decomp-lower-bound}.(2)} \mper
        \label{eq:proof-log-opt-prop-i-decomp-lower-bound}
    \end{align}
    Using property $(i)$, we have that 
    $\eqref{eq:proof-log-opt-prop-i-decomp-lower-bound}.(1)$ is zero $\Q$-almost
    surely. Focusing now on 
    $\eqref{eq:proof-log-opt-prop-i-decomp-lower-bound}.(2)$, we have
    \begin{align}
        \eqref{eq:proof-log-opt-prop-i-decomp-lower-bound}.(2)  
        &= \liminf_{n \to \infty} \left(\frac{1}{n} \logp{W_\nQ^\star}
          - \E_{\Q}\Brac{\ell_{1,\Q}(\optarm)} + \E_{\Q}\Brac{\ell_{1,\Q}(\optarm)}
          - \frac{1}{n} \sum_{i=1}^n \E_{\Q}\Brac{\ell_{i, \Q}(\tA_i)}\right. \\
        &\qquad\qquad \left. + \frac{1}{n} \sum_{i=1}^n \E_{\Q}\Brac{\ell_{i, \Q}(\tA_i)}
        -\frac{1}{n} \logp{\tW_n}\right) \\
        &\geq \liminf_{n \to \infty} \left(\frac{1}{n} \logp{W_\nQ^\star} -
        \E_{\Q}\Brac{\ell_{1,\Q}(\optarm)} + \frac{1}{n} \sum_{i=1}^n \E_{\Q}\Brac{
        \ell_{i, \Q}(\tA_i)} - \frac{1}{n} \logp{\tW_n} \right),
        \label{eq:proof-log-opt-prop-i-lower-bound-opt-arm-def}
    \end{align}
    where the inequality follows from the definition of $\optarm$. 
    We now continue to analyze 
    \eqref{eq:proof-log-opt-prop-i-lower-bound-opt-arm-def}:
    \begin{align}
        \eqref{eq:proof-log-opt-prop-i-lower-bound-opt-arm-def}
        &\geq \liminf_{n \to \infty} \Paren{\frac{1}{n}\logp{W_\nQ^\star} -
        \E_{\Q}\Brac{\ell_{1,\Q}(\optarm)}} 
        + \liminf_{n \to \infty} \Paren{\frac{1}{n} \sum_{i=1}^n
        \E_{\Q}\Brac{\ell_{i, \Q}(\tA_i)} - \frac{1}{n} \logp{\tW_n}} \\
        &= \liminf_{n \to \infty} \Paren{\frac{1}{n} \sum_{i=1}^n
        \E_{\Q}\Brac{\ell_{i, \Q}(\tA_i)} - \frac{1}{n} \logp{\tW_n}}
        \label{eq:proof-log-opt-prop-i-equality-lln-opt-process} \mcom
    \end{align}
    where \eqref{eq:proof-log-opt-prop-i-equality-lln-opt-process} follows from
    the fact that $\lim_{n \to \infty} n^{-1} \logp{W_{\nQ}^\star} =
    \E_{\Q}\Brac{\ell_{1,\Q}(\optarm)}$ with $\Q$-probability one by the strong 
    law of large numbers.
    Focusing on \eqref{eq:proof-log-opt-prop-i-equality-lln-opt-process}, define the process $\widetilde W_\Q$ given by 
    \begin{equation}
        \tW_\nQ := \prod_{i=1}^n \optbbet(\widetilde A_i)^\trans \bE_i(\widetilde A_i).
    \end{equation}
    and notice that \eqref{eq:proof-log-opt-prop-i-equality-lln-opt-process} can be further lower bounded as
    \begin{equation}
        \eqref{eq:proof-log-opt-prop-i-equality-lln-opt-process}
        \geq \underbrace{\liminf_{n \to \infty}\Paren{\frac{1}{n} \sum_{i=1}^n
        \E_{\Q}\Brac{\ell_{i, \Q}(\tA_i)} - \frac{1}{n}\logp{\tW_\nQ}}}_{
        \eqref{eq:proof-log-opt-prop-i-final-decomp-concentration-term}.(1)}
        + \underbrace{\liminf_{n \to \infty} \Paren{\frac{1}{n}\logp{\tW_\nQ} -
        \frac{1}{n} \logp{\tW_n}}}_{
        \eqref{eq:proof-log-opt-prop-i-final-decomp-concentration-term}.(2)}.
        \label{eq:proof-log-opt-prop-i-final-decomp-concentration-term}
    \end{equation}
    Notice that $\eqref{eq:proof-log-opt-prop-i-final-decomp-concentration-term}.(1)$ vanishes to zero $\Q$-almost surely by \cref{lemma:exponential concentration for bandit sequence - centered with independent copies} combined with the Borel-Cantelli lemma.
    Similarly, \eqref{eq:proof-log-opt-prop-i-final-decomp-concentration-term} is nonnegative in its limit inferior with $\Q$-probability one by \cref{corollary:arm-wise-portfolio-concentration} combined with the Borel-Cantelli lemma.
    
    Putting all of the previous steps together we conclude that
    \begin{equation}
        \liminf_{n \to \infty} \frac{1}{n} \Paren{\logp{W_n} - \logp{\tW_n}}
        \geq 0 \quad\text{with $\Q$-probability one}\mcom
    \end{equation}
    completing the proof of the property that $W$ is multi-armed $\cQ$-log-optimality and hence the proof of 
    \cref{theorem:log-optimality of spruce}.
\end{proof}

\subsection{Proof of \cref*{lemma:exponential concentration for bandit sequence - optimal}}
\label{proof:exponential concentration for bandit sequence - optimal}
\begin{proof}[Proof of \cref{lemma:exponential concentration for bandit sequence - optimal}]
    Fix $\eps > 0$. 
    For any $n \in \N$ we define $\cR_n$ as
    \begin{equation}
        \cR_n \coloneqq \E_{\Q}\Brac{\logp{\optbbet(\optarm)^\trans
        \bE(\optarm)}} - \frac{1}{n} \sum_{a=1}^K\log( W_n(a)).
    \end{equation}
    Decompose $\cR_n$ into the following three terms:
    \begin{align}
        \cR_n = 
        \ & \E_{\Q}\Brac{\logp{\optbbet(\optarm)^\trans \bE(\optarm)}}
        - \frac{1}{n}\sum_{i=1}^n \E_{\Q}\Brac{\ell_{i, \Q}(A_i)} \\
        & + \frac{1}{n}\sum_{i=1}^n \E_{\Q}\Brac{\ell_{i, \Q}(A_i)} - \sum_{i=1}^n
        \ell_{i, \Q}(A_i) \\
        & + \frac{1}{n} \sum_{i=1}^n \ell_{i, \Q}(A_i) - 
        \frac{1}{n}\sum_{a=1}^K \log(W_n(a)).
    \end{align}
    By a union bound over integers $n \geq m$, we have that 
        $\PP_\Q \left ( \sup_{n \geq m} |\Rcal_n| \geq \eps \right ) \leq \sum_{n=m}^\infty \PP_\Q \left ( |\Rcal_n| \geq \eps \right ),$
    and hence for any $n \geq m$, the summand in the right-hand side can therefore be upper-bounded as
    \begin{align}
        \PP_{\Q}\Paren{\Abs{\cR_n} \geq \epsilon}
        &\leq \1 \left \{ \Abs{ \E_{\Q}\Brac{\logp{\optbbet(\optarm)^\trans 
            \bE(\optarm)}} - \frac{1}{n}\sum_{i=1}^n \E_{\Q}\Brac{\ell_{i, \Q}(A_i)}} \geq
        \epsilon / 3 \right \} \\
        &\quad + \PP_{\Q}\Paren{\Abs{\frac{1}{n}\sum_{i=1}^n \E_{\Q}\Brac{\ell_{i, \Q}(A_i)} 
            - \frac{1}{n}\sum_{i=1}^n \ell_{i, \Q}(A_i)} \geq \epsilon / 3} \\
        &\quad + \PP_{\Q}\Paren{\Abs{\frac{1}{n}\sum_{i=1}^n \ell_{i, \Q}(A_i) 
        - \frac{1}{n} \sum_{a=1}^K \logp{W_n(a)}} \geq \epsilon / 3} 
        \label{eq:proof-exponential-concentration-for-bandit-sequence-portfolio-prob}\mper
    \end{align}
    Using \cref{lemma:exponential concentration for bandit sequence - centered with independent copies} and \cref{corollary:arm-wise-portfolio-concentration} to control the second and third terms, respectively, we have
    \begin{align}
        \PP_\Q \left ( |\Rcal_n| \geq \eps \right ) &\leq \1 \left \{ R_{\nQ}^\MAB \geq
        \epsilon n / 3 \right \} + 2Kn \left ( \exp \left \{ - \frac{\eps^2n }{72bK^2} \right \} + \exp \left \{ -\frac{\eps n }{12K} \right \} \right )\\
        &+ \1  \left \{ \sum_{a=1}^K R^\Port_n(a) \geq \eps n / 3\right  \} + \exp \left \{ -n\eps / 3 \right \}.
    \end{align}
    Summing over $n \geq m$
   completes the proof.
\end{proof}

\subsection{Proof of \cref*{proposition:stopping-timelower-bound}}
\label{sec:proof-stopping-time-lower-bound}
\begin{proof}[Proof of \cref{proposition:stopping-timelower-bound}]
    Let $\tW$ be an arbitrary test $\Pcal$-supermartingale of the form \eqref{eq:prelim-test supermartingale} constructed with any $\Hcal$-predictable values of $\infseqn{\widetilde \bbet_n, \widetilde A_n}$. Let $\tau \equiv \widetilde \tau_{\alpha} \coloneqq \inf 
    \{ n \in \N : \tW_n \geq 1/\alpha \}$ be its associated stopping time for reaching the threshold $1/\alpha$. 
    
    Let $\tW_\nQ$
    be the process that in each time step $n \in \NN$, selects the arm $\tA_n$ (i.e., the same arm as $\tW_n$)
    but always with the optimal portfolio $\bbet_\Q(\cdot)$ for that selected arm. Concretely, define $\tW_\nQ$ by
    \begin{equation}
        \tW_\nQ
    \coloneqq \prod_{i=1}^n \optbbet(\tA_i)^\trans \bE_i(\tA_i)
    \end{equation}
    Now, define the process $\infseqn{W_\nQ^\star}$, which always plays the optimal arm, and its optimal portfolio,
    \begin{equation}
        W_\nQ^\star := \prod_{i=1}^n \optbbet(a_\Q)^\trans \bE_i(a_\Q).
    \end{equation}
    We now write $\E_{\Q}[\log (\tW_{\tau})]$ as follows:
    \begin{align}
        \E_{\Q}\Brac{\logp{\tW_{\tau}}}
        &= \E_{\Q}\Brac{\logp{\tW_{\tau} / \tW_{\tauQ}}} 
        + \E_{\Q}\Brac{\logp{\tW_{\tauQ}}} \\
        &= \underbrace{\E_{\Q}\Brac{\logp{\tW_{\tau} / \tW_{\tauQ}}}}_{
        \eqref{eq:proof-stopping-time-lower-bound-decomp}.(1)}
        + \underbrace{\E_{\Q}\Brac{\logp{\tW_{\tauQ}}
        - \logp{W_{\tauQ}^\star}}}_{
        \eqref{eq:proof-stopping-time-lower-bound-decomp}.(2)} 
        + \E_{\Q}\Brac{\logp{W_{\tauQ}^\star}} \mper
        \label{eq:proof-stopping-time-lower-bound-decomp}
    \end{align}
    We will subsequently upper bound both \eqref{eq:proof-stopping-time-lower-bound-decomp}.(1) and \eqref{eq:proof-stopping-time-lower-bound-decomp}.(2) by zero.
    Turning to the first term,
    \begin{equation}
        \E_{\Q}\Brac{\logp{\tW_{\tau} / \tW_{\tauQ}}}
        \leq \logp{\E_{\Q}\Brac{\tW_{\tau}/\tW_{\tauQ}}}
        \leq \logp{1} = 0 \mcom
    \end{equation}
    where the first inequality follows from Jensen's inequality and the second
    inequality follows from the allocation-wise numeraire property of $\tW_{\tauQ}$ 
    (i.e., \cref{lemma:allocation-wise numeraire portfolio}).

    Turning now to the second term, define the process $\infseqn{L_n}$ given for any $n \in \NN$ by
    \begin{equation}
        L_n \coloneqq \log(\tW_\nQ) - \log(W_\nQ^\star) = \sum_{i=1}^n \Brac{\logp{\optbbet(\tA_i)^\trans
        \bE_i(\tA_i)} - \logp{\optbbet(\optarm)^\trans \bE_i(\optarm)}} \mcom
    \end{equation}
    noting that $\EE_\Q[L_\tau] = \eqref{eq:proof-stopping-time-lower-bound-decomp}$.(2) by definition.
    We will now show that $\infseqn{L_n}$ is a $\Q$-supermartingale with respect to the filtration $\Hcal$ and with mean zero. For any $n \in
    \N$ we have that
    \begin{align}
        \E_{\Q}\Brac{L_n \smvert \cH_{n-1}} 
        &= L_{n-1} + \E_{\Q}\Brac{\logp{\optbbet(\tA_i)^\trans
        \bE_n(\tA_n)} - \logp{\optbbet(\optarm)^\trans \bE_n(\optarm)} \smvert
        \cH_{n-1}} \\
        &\leq L_{n-1} \mcom
    \end{align}
    where the inequality follows from the definition of $\optarm$. 
    Together \cref{lem:bounded-expected-absolute-log-difference,lem:optional-stopping-supermartingale}
    allow us to conclude that $\E_\Q\Brac{L_\tau} \leq 0$.
    Returning to the equality in
    \eqref{eq:proof-stopping-time-lower-bound-decomp}, and applying Wald's
    identity,
    \begin{equation}
        \EE_\Q[\log(\tW_\tau)] \leq \EE_\Q[\log(W_\tauQ^\star)] = \EE_\Q[\tau] \EE_\Q[\log(\optbbet(a_\Q)^\trans \bE_1(a_\Q)].
    \end{equation}
    By definition of $\tau$ being the first time for which $\log(\tW_\tau)$ exceeds $\log(1/\alpha)$, we have that
    \begin{equation}
    \EE_\Q[\widetilde \tau_\alpha] \EE_\Q[\log(\optbbet(a_\Q)^\trans \bE_1(a_\Q)] \geq \log(1/\alpha).
    \end{equation}
    Dividing both sides by the expected log-increment yields the first part of the proposition.

    Moving on to the second part of the proposition, fix $a \in \Acal$ and let $\tW_n^\bracka$ form any $\Pe$-process of the form described in \cref{proposition:stopping-timelower-bound}. Define $\tau \equiv \widetilde \tau_\alpha^\bracka := \inf \{ n \in \NN : \tW_n^\bracka \geq 1/\alpha \}$. Following a similar approach to the first part of the proof, define $\tW_\nQ^\bracka$ as the process given by
    \begin{equation}
        \tW_\nQ^\bracka := \prod_{i=1}^n \optbbet(a)^\trans\bE_i(a)
    \end{equation}
    and observe that
    \begin{align}
        \EE_\Q \left [ \log(\tW_\tau^\bracka) \right ] &= \EE_\Q \left [ \log \left (\tW_\tau^\bracka / \tW_\tauQ^\bracka\right ) \right ] + \EE_\Q \left [ \log(\tW_\tauQ^\bracka) \right ] \\
        &\leq \underbrace{\log \left ( \EE_\Q \left [ \tW_\tau^\bracka / \tW_\tauQ^\bracka \right ]\right )}_{\leq 0} + \EE_\Q \left [ \log(\tW_\tauQ^\bracka) \right ].
    \end{align}
    Since $\log(\tW_\tau^\bracka) \geq \log(1/\alpha)$, we have by an application of Wald's identity that
    \begin{equation}
        \EE_\Q[\tau] \EE_\Q\left [ \log(\optbbet(a)^\trans \bE_1(a)) \right ] \geq \log(1/\alpha).
    \end{equation}
    By definition of $a_\Q$, it follows that
    \begin{equation}
        \EE_\Q[\widetilde \tau_\alpha^\bracka] \geq \frac{\log(1/\alpha)}{\EE_\Q\left [ \log(\optbbet(a)^\trans \bE_1(a)) \right ]} \geq \frac{\log(1/\alpha)}{\EE_\Q\left [ \log(\optbbet(a_\Q)^\trans \bE_1(a_\Q)) \right ]},
    \end{equation}
    completing the proof.
\end{proof}

\subsection{Proof of \cref*{thm:expected-stopping-time-upper-bound}}
\label{sec:proof-expected-stopping-time}
\begin{proof}[Proof of \cref{thm:expected-stopping-time-upper-bound}]
    Fix $\Q \in \cQ$ and $\alpha \in (0,1)$.
    For an arbitrary $\delta \in (0,1)$, let
    \begin{align}
        \epsilon \coloneqq \frac{\delta}{1 + \delta}
        \E_{\Q}\Brac{\ell_{1,\Q}(\optarm)}
        \quad\mathrm{and}\quad m \coloneqq \ceil{\frac{\logp{1/\alpha}}
        {\E_{\Q}\Brac{\ell_{1,\Q}(\optarm) } - \epsilon}} 
        = \ceil{\frac{(1 + \delta)
        \logp{1/\alpha}}{\E_{\Q}\Brac{\ell_{1,\Q}(\optarm)}}} \mper
    \end{align}
    
    We begin by rewriting the expected stopping time through the tail sum 
    formula:
    \begin{align}
        \E_{\Q}\Brac{\tau} 
        &= \sum_{t=1}^\infty \PP_{\Q}\Paren{\tau \geq t} \\
        &\leq m + \sum_{t=m}^\infty \PP_{\Q}\Paren{\tau \geq t} \\
        &\leq 1 + m + \sum_{t=m}^\infty \PP_{\Q}\Paren{\tau > t} \mcom
        \label{eq:proof-stopping-time-ub-series-probs}
    \end{align}
    where the first inequality follows by upper bounding the first $m$ probabilities
    in the series by one, and in the second inequality we rewrite the event
    inside the probability $\tau \geq t$ as $\tau > t$ while upper bounding the
    corresponding $m-1$th probability by one. We now recall our definition of 
    the stopping time $\tau \equiv \tau_{\alpha} \coloneqq \Set{n \in \N \smvert 
    W_n \geq 1/\alpha}$. This definition allows us to rewrite the sum in 
    \eqref{eq:proof-stopping-time-ub-series-probs} as
    \begin{align}
        &\sum_{t = m}^\infty \PP_{\Q}\Paren{\tau > t}\\ 
        &\qquad= \sum_{t = m}^\infty \PP_{\Q}\Paren{W_t < 1/\alpha}\\
        &\qquad= \sum_{t=m}^\infty \PP_{\Q}\Paren{t^{-1}\logp{W_t} < 
        t^{-1} \logp{1/\alpha}} \\
        &\qquad \leq \sum_{t=m}^\infty
        \Brac{\PP_{\Q}\Paren{\E_{\Q}\Brac{\ell_{1,\Q}(\optarm)} 
        - \epsilon \leq t^{-1} \logp{1/\alpha}} 
        + \PP_{\Q}\Paren{t^{-1}\logp{W_t} - \E_{\Q}\Brac{\ell_{1,\Q}(\optarm)} < -\epsilon}} \\
        &\qquad \leq \sum_{t=m}^\infty \Brac{\Ind{\E_{\Q}\Brac{\ell_{1,\Q}(\optarm)} 
        - \epsilon \leq m^{-1} \logp{1/\alpha}} 
        + \PP_{\Q}\Paren{\E_{\Q}\Brac{\ell_{1,\Q}(\optarm)} - t^{-1}\logp{W_t} \geq \epsilon}} 
        \label{eq:proof-stopping-time-ub-indicator-upper-bound}
    \end{align}
    where the first inequality follows from a union bound, and the second
    inequality follows from upper bounding the probability by an indicator and
    from the fact that $t^{-1} \leq m^{-1}$ whenever $t \geq m$.
    By plugging in the definition of $m$ we can see how 
    \eqref{eq:proof-stopping-time-ub-indicator-upper-bound} simplifies to
    \begin{align}
        \eqref{eq:proof-stopping-time-ub-indicator-upper-bound}
        &= \sum_{t=m}^\infty \PP_{\Q}\Paren{\E_{\Q}\Brac{\ell_{1,\Q}(\optarm)} 
        - t^{-1} \logp{W_t} \geq \epsilon} \\
        &\leq \sum_{t=m}^\infty
        \PP_{\Q}\Paren{\Abs{\E_{\Q}\Brac{\ell_{1,\Q}(\optarm)} 
        - t^{-1} \logp{W_t}} \geq \epsilon} \mper
        \label{eq:proof-stopping-time-bandit-sequence-concentration-term}
    \end{align}
    Appealing to the proof of 
    \cref{lemma:exponential concentration for bandit sequence - optimal}, we have 
    \begin{align}
        \eqref{eq:proof-stopping-time-bandit-sequence-concentration-term} &\leq \1 \left \{ R_\tQ^\MAB \geq \eps t / 3 \right \} + 3Kt \left ( \exp \left \{ - \frac{\eps^2 t }{72bK^2} \right \} + \exp \left \{ -\frac{\eps t }{12 K} \right \} \right ) \\
        &\quad + \1 \left \{ \sum_{a=1}^K R_t^\CO(a) \geq \eps t/3 \right \}.
    \end{align}
    Putting these inequalities together and substituting in the definition of $m$ we obtain the following upper bound on the expected stopping time:
    \begin{align}
        \E_{\Q}\Brac{\tau} 
        \leq 2 + \frac{(1+\delta) \log(1/\alpha)}{\EE_\Q [\ell_{1,\Q}(\optarm)]} &+ \sum_{t=m}^\infty \1 \left \{ R_\tQ^\MAB \geq \eps t / 3 \right \} \\
        &+ \sum_{t=m}^\infty 3Kt \left ( \exp \left \{ - \frac{\eps^2 t }{72bK^2} \right \} + \exp \left \{ -\frac{\eps t }{12 K} \right \} \right ) \\
        & + \sum_{t=m}^\infty \1 \left \{ \sum_{a=1}^K R_t^\CO(a) \geq \eps t/3 \right \} .
    \end{align}
    Dividing both sides of the inequality by $\logp{1/\alpha}$, noting the sublinearity of $R_t^\CO(a)$ for each $a \in \Acal$ and the sublinearity of $R_\tQ^\MAB$ as in \cref{theorem:expected suboptimal pulls}, and taking the limit superior as $\alpha \to 0^+$, we have 
    \begin{equation}
        \limsup_{\alphato}\frac{\EE_\Q[\tau_\alpha]}{\log(1/\alpha)} \leq \frac{1+\delta}{\EE_\Q[\ell_\Q(a_\Q)]}.
    \end{equation}
    Since $\delta \in (0, 1)$ was arbitrary, the above holds with a one in the numerator, completing the proof.
\end{proof}

\begin{lemma}[Sub-optimal arm-pull event decomposition]\label{lemma:suboptimal arm selection implications}
  Fix $a \in \Acal$ with $a \neq \optarm$. Suppose that actions $\infseqn{A_n}$ are chosen according to \SPRUCE{}. For any $t \in \NN$, let $B_t^\bracka
  \coloneqq \{ A_t = a \}$ be the event that suboptimal arm $a$ is selected at time $t$. For any $t \geq K+1$, it holds that
  \begin{equation}
    B_t^\bracka \subseteq B_t^{(\text{UCB-$\optarm$})} \cup B_t^{(\text{LCB-$a$})} \cup B_t^{(\Delta_a)},
  \end{equation}
  where
  \begin{align}
      B_t^{\brackUCBastar} &:= \left \{ \UCB_{\optarm}(t) <
          \E_{\Q}\Brac{\ell_{1, \Q}(\optarm)} \right \},\\
    B_t^{\brackLCBa} &:= \left \{ \LCB_a(t) > \E_{\Q}\Brac{\ell_{1, \Q}(\optarm)} \right \}, \text{ and}\\
    B_t^{\brackDeltaa} &:= \left \{ \Delta_a < 2 \left ( \sqrt{\frac{8b \gamma \log(\zeta t + 1)}{N_a(t-1)}} + \frac{5 \gamma \log (\zeta t + 1)}{N_a(t-1)} + \frac{R_{N_a(t-1)}}{N_a(t-1)} \right ) \right \}.
  \end{align}
  In words, choosing a suboptimal arm implies that either the UCB for the
  optimal arm $\optarm$ miscovered, the LCB for the suboptimal arm $a$ miscovered, or the suboptimality gap is small relative to the confidence bound width.
\end{lemma}
    \begin{proof}
      Suppose that $B_t^\bracka \not \subseteq B_t^\brackUCBastar$ and that $B_t^\bracka \not \subseteq B_t^\brackLCBa$. We will show that $B_t^\bracka \subseteq B_t^\brackDeltaa$. Indeed, we have that on the event $B_t^\bracka$,
      \begin{align}
        \UCB_a(t) &\geq \UCB_{\optarm}(t)\\
                  &> \EE_{\Q} \left [ \ell_{1, \Q}(\optarm) \right ] 
                  - \EE_{\Q} \left [ \ell_{1, \Q}(a) \right ] + \EE_{\Q} \left [
                  \ell_{1, \Q}(a) \right ]\\
                  &= \Delta_a + \EE_{\Q} \left [ \ell_{1, \Q}(a) \right ] \\
                          &> \Delta_a + \LCB_a(t),
      \end{align}
      where the first inequality is precisely the event $B_t^\bracka$, the second uses the fact that $B_t^\bracka \not \subseteq B_t^\brackUCBastar$, and the third uses the fact that $B_t^\bracka \not \subseteq B_t^\brackLCBa$. We therefore have that on the event $B_t^\bracka$,
      \begin{equation}
	 2 \left ( \sqrt{\frac{8b \gamma \log(\zeta t + 1)}{N_a(t-1)}} + \frac{5 \gamma \log (\zeta t + 1)}{N_a(t-1)} + \frac{R_{N_a(t-1)}}{N_a(t-1)} \right ) \equiv \UCB_a(t) - \LCB_a(t) > \Delta_a.
      \end{equation}
      It follows that $B_t^\bracka \subseteq B_t^\brackDeltaa$, and hence
      \begin{equation}
	B_t^\bracka \subseteq B_t^{(\text{UCB-$\optarm$})} \cup B_t^{(\text{LCB-$a$})} \cup B_t^{(\Delta_a)},
      \end{equation}
      which completes the proof.
\end{proof}

\subsection{Proof of \cref*{lemma:sub-exponential-log-wealth-increments}}\label{proof:sub-exponential-log-wealth-increments}

\begin{proof}[Proof of \cref{lemma:sub-exponential-log-wealth-increments}]
    Let us first consider the case where $\theta \in [0, 1]$. 
    Through a direct calculation, it holds that
    \begin{align}
        \EE_\Q\Brac{\exps{\theta \Paren{\ell_{1, \Q} - \EE_\Q[\ell_{1, \Q}]}}} 
        &= \E_\Q\Brac{\Paren{\frac{\exps{\ell_{1, \Q}}}{\exps{\EE_\Q[\ell_{1,\Q}]}}}^{\theta}} \\
        &= \E_\Q\Brac{\Paren{\frac{\optbbet^\trans \bE_1}
            {\exps{\E_\Q\Brac{\logp{\optbbet^\trans \bE_1}}}}}^\theta} \\
        &\leq \E_\Q\Brac{\Paren{\optbbet^\trans \bE_1}^\theta}\\
      &\leq b^\theta \leq b,
    \end{align}
    where in the first equality, we used the fact that $\EE_\Q[\log
    (\optbbet^\trans \bE_1 )] \geq 0$ by construction, and the second inequality
    uses the almost-sure upper bound in \eqref{eq:q-a.s. upper bound assumption}
    and the fact that $\theta \in [0,1]$ and $b > 1$.

    Consider now the case where $\theta \in [-1, 0)$ and let $\beta := -\theta \in (0, 1]$. We have
    \begin{align}
        \E_\Q\Brac{\exps{\theta \Paren{ \ell_{1, \Q} - \EE_\Q[\ell_{1, \Q}]}}}
        &= \E_\Q\Brac{\Paren{\frac{\exps{\E_\Q\Brac{\logp{\optbbet^\trans \bE_1}}}}
        {\optbbet^\trans \bE_1}}^\beta} \\
        &\leq \E_\Q\Brac{\Paren{\frac{\E_\Q \left [ \exps{\logp{\optbbet^\trans \bE_1}} \right ]}
        {\optbbet^\trans \bE_1}}^\beta}
        \label{eq:proof-sub-exponential-upper-bound-side-two-jensens} \\
        &\leq b^\beta \E_\Q\Brac{\Paren{\frac{1}
          {\optbbet^\trans \bE_1}}^\beta},
    \end{align}
    where the first inequality follows from Jensen's inequality applied to the
    convex function $x \mapsto \exp \{ x\}$ and the second follows from the
    almost sure upper bound on $\optbbet^\trans \bE_1$. By another application of
    Jensen's inequality but now applied to the concave function $x \mapsto
    x^\beta$, we have
    \begin{equation}
      b^\beta \E_\Q\Brac{\Paren{\frac{1}
          {\optbbet^\trans \bE_1}}^\beta} \leq b^\beta \left ( \EE_\Q \left [ \frac{1}{\optbbet^\trans \bE_1} \right ] \right )^\beta.
    \end{equation}
    Using the fact that there exists some $\lambda_1$ for which $(1-\lambda_1)
    \dot E_1 + \lambda_1 \ddot E_1 = 1$ alongside the numeraire property of
    $\optbbet$, we have
    \begin{equation}
      b^\beta \left ( \EE_\Q \left [ \frac{1}{\optbbet^\trans \bE_1} \right ] \right )^\beta \leq b^\beta \leq b.
    \end{equation}
    Combining the above with the previous upper bound for the case where $\theta
    \in [0, 1]$, we conclude that for any $\theta \in [-1, 1]$,
    \begin{equation}
        \EE_\Q \left [ \exp \left \{ \theta \left ( \ell_{1, \Q} - \EE_\Q
        [\ell_{1, \Q}] \right ) \right \} \right ] \leq b
    \end{equation}
    which completes the proof.
\end{proof}

\subsection{Proof of \cref*{proposition:regret based confidence interval}}
\label{proof:regret based confidence interval}
\begin{proof}[Proof of \cref{proposition:regret based confidence interval}]
    We start by using the following portfolio regret inequality,
    \\$\max_{\bbet \in \simplex} \sum_{i=1}^n \logp{\bbet^\trans \bE_i} 
    - \sum_{i=1}^n \logp{\bbet_i^\trans \bE_i} \leq R_n^\Port$, and use it to obtain
    the following upper bound,
    \begin{align}
      &\PP_\Q \left ( \max_{\bbet \in \simplex}\frac{1}{n} \sum_{i=1}^n
          \logp{\bbet^\trans \bE_i(a)} - \EE_\Q\left[\logp{\optbbet^\trans \bE_1}\right] \geq \sqrt{\frac{8b\log(1/\alpha)}{n}}
        + \frac{5 \log (1/\alpha)}{n} + \frac{R_n}{n}\right )\\
      &\quad\leq \PP_\Q \left ( \frac{1}{n} \sum_{i=1}^n \logp{\bbet_i^\trans \bE_i} -
        \EE_\Q\left[\logp{\optbbet^\trans \bE_1}\right] \geq \sqrt{\frac{8b\log(1/\alpha)}{n}} + \frac{5 \log
       (1/\alpha)}{n}\right ) \leq 2\alpha,
    \end{align} 
    where the second inequality follows from \cref{proposition:root-n logwealth}. 
    Analyzing the same deviation but from below, we have
    \begin{align}
      &\PP_\Q \left ( \EE_\Q\left[\logp{\optbbet^\trans \bE_1}\right] - \max_{\bbet \in \simplex}
      \frac{1}{n} \sum_{i=1}^n \logp{\bbet^\trans \bE_i} \geq
      \sqrt{\frac{8b\log(1/\alpha)}{n}} + \frac{4 \log (1/\alpha)}{n}\right ) \\
      &\quad\leq \PP_\Q \left ( \EE_\Q\left[\logp{\optbbet^\trans \bE_1}\right] - \frac{1}{n}
      \sum_{i=1}^n \logp{\optbbet^\trans \bE_i} \geq \sqrt{\frac{8b\log(1/\alpha)}{n}} +
    \frac{4 \log (1/\alpha)}{n}\right ) \leq \alpha,
    \end{align}
    where the final inequality follows from \cref{lemma:root-n confidence intervals of sub-exp}.
    This completes the proof.
\end{proof}

\subsection{Proof of \cref*{theorem:expected suboptimal pulls}}\label{proof:expected suboptimal pulls}
\begin{proof}
    Fix $m > 0$ whose value will be chosen later. 
    Let $N_a(n)$ be the number of times that arm
    $a \in \Acal$ has been pulled up until and including time $n$. Consider its
    expectation and recall that each arm is played exactly once in the first
    $K$ rounds:
    \begin{align}
        \E_\Q\Brac{N_a(n)} &= 1 + \E_\Q\Brac{\sum_{t = K+1}^n \Ind{A_t = a}} \\
                    &= 1 + \Ex{\sum_{t = K+1}^n \Ind{A_t = a ~\mathrm{and}~ 
                    N_a(t-1) \leq m} + \Ind{A_t = a 
                      ~\mathrm{and}~ N_a(t - 1) > m}}. \label{eq:mabproof m-indicators}
    \end{align}
    Notice that $\sum_{t=K+1}^n \1 \{ A_t = a \text{ and } N_a(t-1) \leq m \} \leq m$ with probability one since in the most extreme case, $N_a(n-1) = m$, in which case at most $m$ of those indicators can be nonzero. As such, we have that
    \begin{align}
                    \eqref{eq:mabproof m-indicators} &\leq 1 + m + \E_\Q\Brac{\sum_{t = K+1}^n \Ind{A_t = a ~\mathrm{and}~
                    N_a(t-1) > m}} \\
                    &\leq 1 + m + \E_\Q\Brac{\sum_{t = m+1}^n \Ind{A_t = a ~\mathrm{and}~
                      N_a(t-1) > m}}, \label{eq:mabproof mceil}
      \end{align}
      where the second inequality follows from the fact that the indicator inside the sum is only positive when $N_a(t-1) > m$, which implies that $t-1 > m$.
      For 
      \begin{equation}
            m \coloneqq \max \left \{ \frac{72 b \gamma \log(\zeta n +
            1)}{\Delta_a^2}, \frac{15\gamma \log(\zeta n + 1)}{\Delta_a},
        \frac{3R_{n-1}}{\Delta_a} \right \} \mcom
      \end{equation}
    notice that on the event $\{ N_a(t-1) > m \}$, it holds that
    \begin{align}
      &2 \left ( \sqrt{\frac{8b \gamma \log(\zeta t + 1)}{N_a(t-1)}} + \frac{5 \gamma \log (\zeta t + 1)}{N_a(t-1)} + \frac{R_{N_a(t-1)}}{N_a(t-1)} \right ) \\
      &\quad \leq 2 \left ( \sqrt{\frac{8b \gamma \log(\zeta t + 1)}{N_a(t-1)}} + \frac{5 \gamma \log (\zeta t + 1)}{N_a(t-1)} + \frac{R_{t-1}}{N_a(t-1)} \right ) \\
    &\quad \leq 2 \left ( \sqrt{\frac{ \Delta_a^2 \cancel{8b \gamma \log(\zeta t + 1)}}{9\cdot \cancel{8b\gamma \log(\zeta t + 1)}}} + \frac{\Delta_a\cancel{5 \gamma \log (\zeta t + 1)}}{3 \cdot \cancel{5 \gamma \log(\zeta t + 1)}} + \frac{\Delta_a \cancel{R_{t-1}}}{3 \cancel{R_{t-1}} } \right ) \\
      &\quad = 2 \Delta_a\mper
    \end{align}
    In other words, $\{ N_a(t-1) > m \} \not \subseteq B_t^\brackDeltaa$.
    Therefore, by \cref{lemma:suboptimal arm selection implications} (to be stated and proven later), it is the
    case that $\{ A_t = a \text{ and } N_a(t-1) > m \} \subseteq
    B_t^\brackUCBastar \cup B_t^\brackLCBa$, and hence we have that
    \begin{align}
        \eqref{eq:mabproof mceil} &\leq 1 + m + \E_\Q\Brac{\sum_{t=m+1}^n \1 \left \{
        B_{t}^{(\text{UCB-$\optarm$})} \cup B_{t}^{(\text{LCB-$a$})} \right \} } \\
        &\leq 1 + m + \sum_{t=m+1}^n \PP_\Q \left (B_{t}^{(\text{UCB-$\optarm$})} \right ) + \sum_{t=m+1}^n \PP_\Q \left (B_{t}^{(\text{LCB-$a$})} \right ).
    \end{align}
    We now consider each of the sums separately. We define the following
    shorthands: $\ell_{i, \Q}(a) \coloneqq \logp{\optbbet(a)^\trans \bE_i(a)}$
    and $\ell_{i}(a) \coloneqq \logp{\bbet_i(a)^\trans \bE_i(a)}$ for
    every $a \in \cA$. We proceed to inspect the first term.
    Using \cref{proposition:root-n logwealth}, we have that
    \begin{align}
      \sum_{t = m+1 }^n \PP_\Q \left ( B_t^{\UCB\text{-}\optarm} \right ) 
      &\leq
      \sum_{t= m+1 }^n \PP_\Q \left ( \EE_\Q [\ell_{1, \Q}(\optarm)] -
      \sum_{i=1}^{N_{\optarm}(t-1)} \frac{\ell_{i}(\optarm)}{N_{\optarm}(t-1)} \right.\\
      &\qquad\left. > \sqrt{\frac{8b\gamma \log(\zeta t + 1)}{N_{\optarm}(t-1)}} +
      \frac{5\gamma \log(\zeta t + 1)}{N_{\optarm}(t-1)} + \frac{R_{N_{\optarm}
      (t-1)}}{N_{\optarm}(t-1)}  \right )\\
      &\leq \sum_{t=m+1}^n \sum_{s=1}^t \PP_\Q\left (\EE_\Q[\ell_{1, \Q}(\optarm)]
          - \frac{1}{s} \sum_{i=1}^{s} \ell_{i}(\optarm)  > \sqrt{\frac{8b\gamma 
      \log(\zeta t + 1)}{s}} + \frac{5\gamma \log( \zeta t + 1)}{s} + \frac{R_{s}}{s} \right )  \\
      &\leq \sum_{t= m+1 }^n \sum_{s=1}^t  \frac{2}{(\zeta t + 1)^\gamma}\\
      &\leq \zeta^{-\gamma}\sum_{t= m+1 }^n \frac{2}{t^{\gamma - 1}}\\
      &\leq \frac{2 \zeta^{-\gamma}}{\gamma - 2}.
    \end{align}
    Considering now the second term, we similarly have
    \begin{align}
        \sum_{t=m+1}^n \PP_\Q \left (B_{t}^{(\text{LCB-$a$})} \right ) &\leq \frac{2\zeta^{-\gamma}}{\gamma - 2}.
    \end{align}
    Putting all of the previous inequalities together, we finally conclude that
    \begin{equation}
      \EE_\Q \left [ N_a(n) \right ] \leq 1 + \max \left \{ \frac{72 b \gamma \log(\zeta n + 1)}{\Delta_a^2}, \frac{15\gamma \log(\zeta n + 1)}{\Delta_a}, \frac{3R_{n-1}}{\Delta_a} \right \} + \frac{4 \zeta^{-\gamma}}{\gamma - 2},
    \end{equation}
    and hence $\EE_\Q \Brac{N_a(n)} = O \Paren{ \Delta_a^{-2} \log (n)}$
    which completes the proof of the first claim. The second follows from \citet[Lemma 4.5]{lattimore2020bandit}.
\end{proof}

\section{Auxiliary lemmas}

\subsection{Arm- and allocation-wise numeraire portfolios}

\begin{lemma}[Arm-wise numeraire portfolios]\label{lemma:arm-wise-numeraire-portfolio}
    Fix an arm $a \in \cA$. 
    Let $\Hcal$ be the filtration described in \cref{section:prelim-multi-armed},
    $\optbbet(a)$ be the numeraire portfolio under arm $a$,
    and suppose that $\infseqn{A_n}$ and $\infseqn{\bbet_n}$ are predictable
    sequences.
    Then, for any $\Qin$, the process $\infseqn{S_n^\Q(a)}$ given by
    \begin{equation}
        S_n^\Q(a) = \prod_{i=1}^n \left ( \frac{\bbet_i(a)^\trans
        \bE_i(a)}{\optbbet(a)^\trans \bE_i(a)} \right )^{\1 \{ A_i = a \}}
    \end{equation}
    is a nonnegative $\Q$-supermartingale with mean one. Consequently, for any $\Hcal$-stopping time $\tau$,
    \begin{equation}
        \forall \Qin,\quad \EE_\Q[S_\tau^\Q(a)] \leq 1.
    \end{equation}
\end{lemma}

\begin{proof}[Proof of \cref{lemma:arm-wise-numeraire-portfolio}]
  Nonnegativity follows by construction. To demonstrate that $S_n^\Q(a)$ forms
  a nonnegative $\Q$-supermartingale, observe that for any $n \in \NN$,
  \begin{align}
    \EE_\Q \left [ S_n^\Q(a) \mid \Hcal_{n-1} \right ] &= S_{n-1}^\Q(a)  \EE_\Q
    \left [ \left (\frac{\bbet_n(a)^\trans \bE_n(a)}{\optbbet(a)^\trans
    \bE_n(a)} \right )^{\1 \{ A_n = a \}} \Big \vert \Hcal_{n-1} \right ].
  \end{align}
  Analyzing the conditional expectation and keeping in mind that both
  $\bbet_n(a)$ and $A_n$ are $\Hcal_{n-1}$-measurable, we have that
  if $A_n \neq a$, then the right-hand side is one $\Q$-almost surely. On the other hand, 
  if $A_n = a$, then by $\Hcal_{n-1}$-measurability of $\bbet_n(a)$ and independence between
  $\bE_n(a)$ and $\Hcal_{n-1}$, we have, by \citet[Theorem 15.2.2]{cover1999elements} 
  (see also \citep{waudby2025universal}),
  \begin{equation}
    \EE_\Q \left [ \frac{\bbet_n(a)^\trans \bE_n(a)}{\optbbet(a)^\trans
    \bE_n(a)}  \Big \vert \Hcal_{n-1} \right ] \leq 1\quad\text{$\Q$-almost surely.}
  \end{equation}
  Putting these results together, we have that $\infseqn{S_n^\Q(a)}$ is a
  $\Q$-supermartingale with mean one. The final result follows from Doob's
  optional stopping theorem.
\end{proof}

\cref{lemma:arm-wise-numeraire-portfolio}
is used in the derivation of the
confidence intervals for the optimal log-increments under each arm. The arm-wise
numeraire property allows us to obtain the following bound for any arm $a \in
\cA$, $\alpha \in (0, 1)$, and predictable portfolios $(\bbet_i)_{i=1}^n$:
\begin{equation}
    \PP_{\Q}\Paren{\frac{1}{n} \sum_{i=1}^n \Paren{\logp{\bbet_i^\trans \bE_i} -
    \logp{\optbbet^\trans \bE_i}} \geq \frac{\logp{1/\alpha}}{n}} \leq \alpha \mper
\end{equation}
This inequality is then used to relate the average of the log-increments
under the predictable portfolios to the expected log-increment under the numeraire
portfolio by incurring an additional cost of $\logp{1/\alpha} / n$. Note that such a cost is substantially smaller than what one would typically observe in concentration inequalities for sample averages around their \emph{means}---those rates would typically scale as $\sqrt{\log(1/\alpha) / n}$.

\begin{lemma}[Allocation-wise numeraire portfolios]
\label{lemma:allocation-wise numeraire portfolio}
    Let $\cH$ be the filtration described in \cref{section:prelim-multi-armed}.
    Suppose that $\infseqn{A_n}$ and $\infseqn{\bbet_n}$ are predictable
    sequences.
    Then for any $\Q \in \cQ$, the process $\infseqn{S_n^{\Q}}$ given by
    \begin{equation}
        S_n^{\Q} \coloneqq \prod_{i=1}^n \Paren{\frac{\bbet_i(A_i)^\trans
        \bE_i(A_i)}{\optbbet(A_i)^\trans \bE_i(A_i)}}
    \end{equation}
    is a nonnegative $\Q$-supermartingale with mean one. Consequently, for any
    $\cH$-stopping time $\tau$,
    \begin{equation}
        \forall \Q \in \cQ \mcom \quad \E_{\Q}\Brac{S_{\tau}^{\Q}} \leq 1 \mper
    \end{equation}
\end{lemma}
\begin{proof}[Proof of \cref{lemma:allocation-wise numeraire portfolio}]
    Our goal is to demonstrate that $S_n^{\Q}$ forms a nonnegative
    $\Q$-supermartingale. Indeed, nonnegativity follows by construction.
    To show that $S_n^{\Q}$ is a $\Q$-supermartingale, observe that for any $n
    \in \N$,
    \begin{equation}
        \E_{\Q}\Brac{S_n^{\Q} \smvert \cH_{n-1}} = S_{n-1}^{\Q} \E_{\Q}\Brac{
        \frac{\bbet_n(A_n)^\trans \bE_n(A_n)}{\optbbet(A_n)^\trans \bE_n(A_n)} 
        \smvert \cH_{n-1}} \mper
    \end{equation}
    We now focus on the conditional expectation and recall that both $A_n$ and
    $\bbet_n(A_n)$ are $\cH_{n-1}$ measurable, meaning that the only source of
    randomness left in the conditional expectation is due to the independent and
    identically distributed $e$-variables. We have from 
    \citet[Theorem 15.2.2]{cover1999elements} (see also \citep{waudby2025universal})
    that
    \begin{equation}
        \E_{\Q}\Brac{\frac{\bbet_n(A_n)^\trans \bE_n(A_n)}{\optbbet(A_n)^\trans \bE_n(A_n)} 
        \smvert \cH_{n-1}} \leq 1 \quad \Q\text{-almost surely} \mper 
    \end{equation}
    Putting these results together, we have that $\infseqn{S_n^{\Q}}$ 
    is a $\Q$-supermartingale with mean one. The final result follows from
    Doob's optional stopping theorem.
\end{proof}

\begin{corollary}
\label{corollary:arm-wise-portfolio-concentration}
  Let $W_n$ be a $\Pe$-process satisfying \cref{assumption:arm-wise products} and which for each $a \in \Acal$ has an arm-wise portfolio regret bounded by $R_n^\Port(a)$. Then for any $\Qin$ and any $\eps > 0$,
  \begin{equation}
    \PP_\Q \left ( \left \lvert  \sum_{i=1}^n \ell_{i,\Q}(A_i) - \sum_{a=1}^K \log(W_n(a)) \right \rvert \geq \eps \right ) \leq \1  \left \{ \sum_{a=1}^K R^\Port_n(a) \geq \eps n \right  \} + \exp \left \{ -n\eps \right \}.
  \end{equation}
\end{corollary}
\begin{proof}[Proof of \cref{corollary:arm-wise-portfolio-concentration}]
  Fix $\eps > 0$ and first consider the deviation between $\sum_{i=1}^n \ell_{i,\Q}(A_i)$ and $\sum_{a=1}^K \log(W_n(a))$ from below:
  \begin{align}
    &\PP_{\Q} \left (   \frac{1}{n}\sum_{a=1}^K \left ( \sum_{i=1}^n \1 \{ A_i = a \} \ell_{i, \Q}(a) - \logp{W_n(a)}\right ) 
    \geq \epsilon \right ) \\
    &\leq \PP_{\Q} \left ( \frac{1}{n}\sum_{a=1}^K \left ( \max_{\bbet \in \simplex}\sum_{i=1}^n \1 \{ A_i = a \} \log(\bbet^\trans \bE_i(a)) - \logp{W_n(a)}\right ) 
      \geq \epsilon \right ) \\
    &\leq \PP_\Q \left ( \frac{1}{n} \sum_{a=1}^K R_n^\Port(a) \geq \eps \right )\\
    &= \1  \left \{ \sum_{a=1}^K R^\Port_n(a) \geq \eps n \right  \},
  \end{align}
  where the first inequality follows from the fact that $\sum_{i=1}^n \1 \{ A_i = a \}\ell_{i,\Q}(a) \leq \max_{\bbet \in \simplex} \sum_{i=1}^n \1 \{A_i = a\}\log(\bbet^\trans \bE_i(a))$ by definition, the second follows from the portfolio regret bound, and the final equality follows from rearranging terms in the previous probability and noting that all of those terms are deterministic.

  Considering now the same deviation but from above, we have
  \begin{align}
    &\PP_{\Q} \left ( \frac{1}{n}\sum_{a=1}^K\left (\log ( W_n(a) ) - \sum_{i=1}^n \1 \{ A_i = a \} \ell_{i, \Q}(a) \right ) 
    \geq \epsilon \right ) \\
    &\ \leq \PP_{\Q} \left ( \sum_{i=1}^n \left (\log ( \bbet_i(A_i)^\trans \bE_i(A_i) ) - \ell_{i, \Q}(A_i) \right ) 
      \geq \eps n \right ) \\
    &\ = \PP_\Q \left ( \prod_{i=1}^n \frac{\bbet_i(A_i)^\trans \bE_i(A_i)}{\optbbet(A_i)^\trans \bE_i(A_i)} \geq \exp \{ \eps n \} \right ) \\
    &\ \leq \exp \{ -n\eps \} \EE_\Q \left [ \prod_{i=1}^n \frac{\bbet_i(A_i)^\trans \bE_i(A_i)}{\optbbet(A_i)^\trans \bE_i(A_i)} \right ]\\
      &\ \leq \exp \left \{ -n\eps \right \},
  \end{align}
  where the first inequality follows from the fact that $W \in \Wcal$, the second inequality uses Markov's inequality, and the final inequality
  is a consequence of \cref{lemma:allocation-wise numeraire portfolio}.
  It follows that
  \begin{equation}
    \PP_\Q \left ( \left \lvert \sum_{a=1}^K \left ( \sum_{i=1}^n \1\{ A_i = a \} \ell_{i,\Q}(a) - \log(W_n(a))\right ) \right \rvert \geq \eps \right ) \leq \1  \left \{ \sum_{a=1}^K R^\Port_n(a) \geq \eps n \right  \} + \exp \left \{ -n\eps \right \}.
  \end{equation}
  completing the proof.
\end{proof}

\begin{remark}[On the connection to the numeraire
    portfolios~\citep{long1990numeraire,karatzas2007numeraire}]
    \cref{lemma:allocation-wise numeraire portfolio} can be interpreted as an extension of the classical ``numeraire portfolio'' setup of
    \citep{long1990numeraire,karatzas2007numeraire} to a
    setting in which an investor must select which stock market to invest in on each
    day (i.e., arm to pull). In this case, the numeraire
    portfolio is defined with respect to the \emph{path} of stock markets that have
    been selected (i.e., arms pulled), meaning that the comparison at each time
    step is with respect to the stock market that was selected. In this fictional extension, the investor does not see the results of those stock markets in which they did not invest.
\end{remark}

\subsection{Properties of sub-exponential random variables}

\begin{lemma}[Bounds on the $p^\tth$ moments of sub-exponential random variables]
\label{lemma:bound-on-pth-moment-of-log-wealth-increments}
Fix $p \in \NN$. Let $X$ be a mean-zero random variable with distribution $\Q$ and suppose that there exists $b > 0$ so that 
\begin{equation}
  \EE_\Q \left [\exp \left \{ |X| \right \} \right ] \leq b.
\end{equation}
Then the $p^\tth$ moment of $X$ is bounded by $b p!$, i.e.
\begin{equation}
  \EE_\Q \left [ |X|^p \right ] < bp!.
\end{equation}
\end{lemma}

\begin{proof}[Proof of \cref{lemma:bound-on-pth-moment-of-log-wealth-increments}]
  We begin by writing out the $p^\tth$ moment as
    \begin{align}
      \EE_\Q \left [ |X|^p \right ] &= \int_0^\infty \PP_\Q \left ( |X|^p > x \right )\dd x\\
                                   &= \int_0^\infty \PP_\Q \left ( |X| > x^{1/p} \right )\dd x \\
                                   &= \int_0^\infty \PP_\Q \left ( \exp \{ |X| \} > \exp \{ x^{1/p} \} \right )\dd x \\
      &\leq b \int_0^\infty  \exp \{ -x^{1/p} \} \dd x,\label{eq:proof-bdd-moments-integral-x}
    \end{align}
    where in the last two lines we employed a Chernoff bound.
    Letting $z = x^{1/p}$ and through a change of variables we have that
    \begin{align}
        \eqref{eq:proof-bdd-moments-integral-x} &= b \int_{0}^\infty \exps{- z} p z^{p-1} \dd z.
        \label{eq:proof-bounded-moments-intermediate-integral-2}
    \end{align}
    Recognizing that the gamma function $\Gamma(q); q \geq 1$ is given by $\int_{0}^\infty \exp\{-z\} z^{p-1} \dd z$, we have
    \begin{equation}
      \eqref{eq:proof-bdd-moments-integral-x} = bp \Gamma(p) = bp!.
    \end{equation}
    Therefore, we have the desired result: $\EE_\Q[|X|^p] \leq bp!$ which completes the proof.
\end{proof}

\begin{lemma}[Factorial power moments imply sub-exponential tails]\label{lemma:range of subgaussianity}
  For a mean-zero random variable $X$ with distribution $\Q$, suppose that there
  exists $b > 0$ so that for every $p \in \NN$, we have $\EE_\Q[|X|^p] \leq
  bp!$. Then $X$ is $(2 \sqrt{b}, 2)$-sub-exponential, i.e.
  \begin{equation}
    \forall \theta \in [-1/2, 1/2], \quad \EE_\Q \left [ \exp \left \{ \theta X \right \} \right ] \leq \exp \left \{ 2 b \theta^2 \right \}.
  \end{equation}
\end{lemma}
\begin{proof}
  First let us consider the case where $\theta \in [0, 1/2)$. We begin by Taylor expanding $x \mapsto \exp \left \{ \theta x \right \}$ around $x= 0$ and evaluating the expectation $\EE_\Q[\exp \{\theta X\}]$:
  \begin{align}
    \EE_\Q \left [ \exp \left \{ \theta X \right \} \right ] &= \EE_\Q \left [ 1 + \theta X + \sum_{p=2}^\infty \frac{(\theta X)^p}{p!} \right ]\\
                                                            &= \EE_\Q \left [ 1 + \sum_{p=2}^\infty \frac{(\theta X)^p}{p!} \right ] \\
                                                            &\leq 1 + \sum_{p=2}^\infty \frac{\theta^p \EE_\Q \left [ |X|^p \right ] }{p!} \\
    &\leq 1 + b\sum_{p=2}^\infty \theta^p.\label{eq:proof end of taylor}
  \end{align}
  Using the fact that $\theta \in [0,1/2)$ and analyzing the geometric series, we have
  \begin{equation}
    \eqref{eq:proof end of taylor} = 1 + b \theta^2 \sum_{j=0}^\infty \theta^j = 1 + \frac{b \theta^2}{1-\theta}.
  \end{equation}
  Again, using the fact that $\theta \in [0, 1/2)$, we have the upper bound
  \begin{equation}
    1 + \frac{b\theta^2}{1-\theta} \leq 1 + 2b\theta^2 \leq \exp \left \{ 2b\theta^2 \right \}.
  \end{equation}
  Taking the previous arguments and combining them, we have
  \begin{equation}
    \EE_\Q  \left [ \exp \left \{ \theta X \right \} \right ] \leq \exp \left \{ 2b\theta^2 \right \}.
  \end{equation}
  The proof for the case where $\theta \in [-1/2,0)$ proceeds analogously.
\end{proof}

\begin{lemma}[Concentration inequalities for random variables with bounded MGFs]
\label{lem:concentration-sub-exponential-rvs}
    Let $X_1, \dots, X_n$ be i.i.d random variables with distribution $\Q$ and
    mean $\mu = \EE_\Q[X_1]$. Suppose that for some $b > 0$, it holds that
    \begin{equation}
        \forall \theta \in [-1, 1],\quad \E_{\Q}\Brac{\exps{\theta (X_1 - \mu)}}
        \leq b \mper
    \end{equation}
    Then for any $\eps > 0$,
    \begin{equation}\label{eq:concentration probability to bound}
        \PP_{\Q}\Paren{\frac{1}{n} \sum_{i=1}^n \Paren{X_i - \mu} \geq \epsilon}
        \leq \exps{-\frac{\epsilon^2 n}{8b}} + \exps{-\frac{\epsilon n}{4}} \mper
    \end{equation}
\end{lemma}
\begin{proof}[Proof of \cref{lem:concentration-sub-exponential-rvs}]
    Let $\theta = \min\{ \eps / (4b), 1/2 \}$.
    Writing out the probability in \eqref{eq:concentration probability to bound} and applying a Chernoff bound, we have
    \begin{align}
        \PP_{\Q}\Paren{\frac{1}{n} \sum_{i=1}^n \Paren{X_i - \mu} \geq \epsilon}
        &= \PP_{\Q}\Paren{\exps{\theta \sum_{i=1}^n \Paren{X_i - \mu}} 
        \geq \exps{\theta n \epsilon}} \\
        &\leq \exps{- \theta n \epsilon} \E_{\Q}\Brac{\prod_{i=1}^n \exps{\theta
        \Paren{X_i - \mu}}} 
        \label{eq:proof-exp-concentration-bound-chernoff-bound}\\
        &= \exps{-\theta n \epsilon} \prod_{i=1}^n \E_{\Q}\Brac{\exps{\theta
        \Paren{X_i - \mu}}} 
        \label{eq:proof-exp-concentration-bound-rvs-independence}\\
        &\leq \exps{- \theta n \epsilon} \prod_{i=1}^n \exps{2b\theta^2} 
        \label{eq:proof-exp-concentration-bound-mgf-bound}\\
        &= \exps{-\theta n \epsilon + 2nb\theta^2} \mcom
    \end{align}
    where \eqref{eq:proof-exp-concentration-bound-chernoff-bound} follows from
    Markov's inequality, \eqref{eq:proof-exp-concentration-bound-rvs-independence}
    follows from the independence of $X_1, \dots, X_n$, and
    \eqref{eq:proof-exp-concentration-bound-mgf-bound} follows from
    \cref{lemma:bound-on-pth-moment-of-log-wealth-increments,lemma:range of subgaussianity}.
    Recalling that $\theta \coloneqq \min\{ \eps / (4b), 1/2 \}$,
    we have that the bound takes the values
    \begin{align}
        \exps{-\theta n \epsilon + 2nb \theta^2} = \begin{cases}
            \exps{-\epsilon^2 n/(8b)} \quad&\text{if}~ \frac{\epsilon}{4b}
                \leq 1/2 \\
            \exps{-\epsilon n/4} \quad&\text{if}~ \frac{\epsilon}{4b} >
            1/2 \mper
        \end{cases}
    \end{align}
    Consequently, we have that for any $n \in \N$ and any $\epsilon > 0$,
    \begin{equation}
        \PP_{\Q}\Paren{\frac{1}{n}\sum_{i=1}^n \Paren{X_i - \mu} \geq \epsilon}
        \leq \exps{-\frac{\epsilon^2 n}{8b}} + \exps{-\frac{\epsilon n}{4}}
        \mcom
    \end{equation}
    completing the proof.
\end{proof}

\begin{lemma}[Root-$n$ confidence intervals under finite MGFs]\label{lemma:root-n confidence intervals of sub-exp}
  Let $X_1, \dots, X_n$ be \iidtext{} random variables with distribution $Q$ and mean $\mu$. Suppose that
  \begin{equation}
    \forall \theta \in [-1, 1], \quad\EE_\Q[ \exp \{ \theta (X_1 - \mu) \} ] \leq b.
  \end{equation}
  Then for any $\theta \in [0, 1/2]$,
  \begin{equation}\label{eq:proof subgaussianity}
   \quad\EE_\Q \left [ \exp \left \{ \theta (X_1 - \mu) \right \} \right ] \leq \exp \{ 2b\theta^2 \} .
  \end{equation}
  Consequently, we have that for any $\alpha \in (0, 1)$,
  \begin{equation}
    \PP_\Q \left ( \frac{1}{n} \sum_{i=1}^n (X_i - \mu) \geq \sqrt{\frac{8b \log (1/\alpha)}{n}} + \frac{4\log(1/\alpha)}{n} \right ) \leq \alpha.
  \end{equation}
\end{lemma}
\begin{proof}
  \cref{lemma:bound-on-pth-moment-of-log-wealth-increments,lemma:range of subgaussianity}
  imply that for any $\eps > 0$ and any $\theta \in [0, 1/2]$,
  \begin{align}
    \PP_\Q \left ( \frac{1}{n} \sum_{i=1}^n (X_i - \mu)  \geq \eps \right ) &= \PP_\Q \left ( \exp \left \{ \theta \sum_{i=1}^n (X_i - \mu) \right \} \geq \exp \{n\theta \eps \} \right )  \\
                                                                &\leq \exp \left \{ -n\theta \eps \right \}  \EE_\Q \left [\prod_{i=1}^n \exp \left \{ \theta (X_i - \mu) \right \} \right ] \\
                                                                &= \exp \left \{ -n\theta \eps \right \} \prod_{i=1}^n \EE_\Q \left [ \exp \left \{ \theta (X_i - \mu) \right \} \right ] \\
    &\leq \exp \{-n\theta \eps  + 2nb\theta^2\} ,
  \end{align}
  where the first two lines uses a Chernoff bound, the third line uses independence of $X_1, \dots, X_n$, and the last line uses \eqref{eq:proof subgaussianity}. Setting the right-hand side to $\alpha$, we have that for any $\theta \in (0, 1/2]$,
  \begin{equation}
    \PP_\Q \left ( \frac{1}{n} \sum_{i=1}^n (X_i - \mu) \geq \frac{\log(1/\alpha) + 2nb\theta^2}{n \theta} \right ) \leq \alpha.
  \end{equation}
  Setting $\theta := \sqrt{ \log (1/\alpha) / (2nb) } \land 1/2$, we have that the margin takes the values
  \begin{equation}
    \frac{\log(1/\alpha) + 2nb\theta^2}{n\theta} = \begin{cases}
      \sqrt{8b\log(1/\alpha)/n} & \text{if}~\sqrt{\log (1/\alpha) / (2nb)} \leq 1/2\\
      4 \log(1/\alpha) / n & \text{if}~\sqrt{\log (1/\alpha) / (2nb)} > 1/2.
    \end{cases}
  \end{equation}
  Consequently, it holds that for any $n \in \NN$ and any $\alpha \in (0,1)$,
  \begin{equation}
    \PP_\Q \left ( \frac{1}{n} \sum_{i=1}^n (X_i - \mu) \geq \sqrt{\frac{8b\log(1/\alpha)}{n}} + \frac{4 \log (1/\alpha)}{n} \right ) \leq \alpha,
  \end{equation}
  completing the proof.
\end{proof}

\subsection{Single-arm concentration inequalities}

\begin{lemma}[Sub-exponentiality of log-increments under numeraire portfolios]
\label{lemma:numeraire portfolio log wealths are sub exponential}
   Fix $\Qin$. Let $\infseqn{\bE_n} \equiv \infseqn{\bE_n(1)}$ be
   $(d+1)$-vectors of $\Pe$-values satisfying
   \cref{assumption:class of eprocesses} with the constant $b > 1$. Then $\log\Paren{\optbbet^\trans
   \bE_1}$ is $(2\sqrt{b}, 2)$-sub-exponential, meaning that
  \begin{equation}
    \forall \theta \in [-1/2, 1/2], \quad \EE_\Q \left [ \exp \left \{ \theta \Paren{\log\paren{\optbbet^\trans \bE_1} - \EE_\Q[\log\paren{\optbbet^\trans \bE_1}]} \right \} \right ] \leq \exp \left \{ 2 b \theta^2 \right \}.
  \end{equation}
\end{lemma}

\begin{proposition}[Confidence intervals for the optimal log-wealth]\label{proposition:root-n logwealth}
  Consider the same setup as \cref{lemma:sub-exponential-log-wealth-increments}. 
  If $\infseqn{\bbet_n}$ is a predictable sequence with portfolio regret $R_n$, then for any $\alpha \in (0, 1)$,
  \begin{equation}
      \PP_\Q \left ( \frac{1}{n} \sum_{i=1}^n \logp{\bbet_i^\trans \bE_i} 
      - \EE_\Q [\logp{\optbbet^\trans \bE_1}] \geq 
      \sqrt{\frac{8b\log(2/\alpha)}{n}} + \frac{5 \log(2/\alpha)}{n} \right ) \leq \alpha
  \end{equation}
  and
  \begin{equation}
      \PP_\Q \left ( \frac{1}{n} \sum_{i=1}^n \logp{\bbet_i^\trans \bE_i} 
          - \EE_\Q [\logp{\optbbet^\trans \bE_1}] \leq 
          - \left [ \sqrt{\frac{8b\log(1/\alpha)}{n}} + \frac{4
    \log(1/\alpha)}{n} + \frac{R_n}{n} \right ] \right ) \leq \alpha
  \end{equation}
\end{proposition}

\begin{proof}[Proof of \cref{lemma:numeraire portfolio log wealths are sub exponential}]
    Our aim is to show that 
    \begin{equation}
        \forall \theta \in [-1/2, 1/2], \quad \EE_\Q \left [ \exp \left \{ \theta \Paren{\log\paren{\optbbet^\trans \bE_1} - \EE_\Q[\log\paren{\optbbet^\trans \bE_1}]} \right \} \right ] \leq \exp \left \{ 2 b \theta^2 \right \}.
    \end{equation}
   First, using \cref{lemma:sub-exponential-log-wealth-increments} we have that the log-increments under the optimal (i.e., numeraire) portfolios have finite MGFs:
   \begin{equation}
       \forall \theta \in [-1, 1], \quad \E_{\Q}\Brac{\exps{\theta \Paren{\logp{\optbbet^\trans \bE_1} - \E_{\Q}\Brac{\logp{\optbbet^\trans \bE_1}}}}} \leq b \mper
   \end{equation}
   We now appeal to \cref{lemma:bound-on-pth-moment-of-log-wealth-increments} which gives us the following upper bound on the $p^\tth$ moments of the log-increments:
   \begin{equation}
       \E_{\Q}\Brac{\Abs{\logp{\optbbet^\trans \bE_1} - \E_{\Q}\Brac{\logp{\optbbet^\trans \bE_1}}}^p} \leq bp! \mper
   \end{equation}
   We use this result and appeal to \cref{lemma:range of subgaussianity} to conclude that
   \begin{equation}
       \forall \theta \in [-1/2, 1/2], \quad \EE_\Q \left [ \exp \left \{ \theta \Paren{\log\paren{\optbbet^\trans \bE_1} - \EE_\Q[\log\paren{\optbbet^\trans \bE_1}]} \right \} \right ] \leq \exp \left \{ 2 b \theta^2 \right \} \mper
   \end{equation}
   This completes the proof.
\end{proof}

\begin{proof}[Proof of \cref{proposition:root-n logwealth}]
    We start by appealing
    to the arm-wise numeraire property, 
    \cref{lemma:arm-wise-numeraire-portfolio}, which allows us to show that 
    for any $\eps > 0$, $\PP_\Q ( \sum_{i=1}^n (\ell_{i} - \ell_{i, \Q})
    \geq n\eps ) \leq \exp\{-n\eps\}\EE[W_n / W_n(\lambda_\Q)] \leq
    \exp\{-n\eps\}$. Therefore, it holds that
    \begin{equation}
        \PP_\Q \left ( \frac{1}{n} \sum_{i=1}^n (\ell_i - \ell_{i, \Q}) 
        \geq \frac{\log(1/\alpha)}{n} \right ) \leq \alpha.
    \end{equation}
    Hence we have that for any $\eps > 0$,
    \begin{align}
      \PP_\Q \left ( \frac{1}{n} \sum_{i=1}^n \ell_i -
      \EE_\Q[\ell_{1, \Q}] \geq \eps \right ) 
      &\leq \PP_\Q \left (\frac{1}{n} \sum_{i=1}^n \ell_{i, \Q} - \EE_\Q [\ell_{1, \Q}] +
      \frac{\log(1/\alpha)}{n} \geq \eps \right ) + \alpha.
      \label{eq:ci proof upper concentration}
    \end{align}
    Therefore, letting $\eps \coloneqq \sqrt{8 b \log(1/\alpha) / n} + 5 \log(1/\alpha) / n $, we have
    \begin{align}
        &\PP_\Q \left ( \frac{1}{n} \sum_{i=1}^n \ell_i -
        \EE_\Q[\ell_{1, \Q}] \geq \sqrt{\frac{8 b \log(1/\alpha)}{n}} +
        \frac{5 \log(1/\alpha)}{n} \right ) \\
        &\quad\leq \PP_\Q \left ( \frac{1}{n} \sum_{i=1}^n \ell_{i, \Q} -
        \EE_\Q[\ell_{1, \Q}] \geq \sqrt{\frac{8 b \log(1/\alpha)}{n}} +
        \frac{4 \log(1/\alpha)}{ n} \right ) + \alpha\\
        &\quad\leq 2\alpha,
    \end{align}
    where we used \eqref{eq:ci proof upper concentration} in the first
    inequality, and the combination of
    \cref{lemma:sub-exponential-log-wealth-increments,lemma:root-n confidence intervals of sub-exp} in the second.

    Now, analyzing concentration of $\frac{1}{n} \sum_{i=1}^n \ell_i -
    \EE_\Q[\ell_{1, \Q}]$ from below, we have that for any $\eps > 0$,
    \begin{align}
        &\PP_\Q \left ( \EE_\Q[\ell_{1, \Q}] - \frac{1}{n} \sum_{i=1}^n
        \ell_i \geq \eps \right )\\
        &\quad= \PP_\Q \left ( \EE_\Q[\ell_{1, \Q}] - \max_{\bbet \in
            \simplex}
        \frac{1}{n} \sum_{i=1}^n \logp{\bbet^\trans \bE_i} + \max_{\bbet \in \simplex}
        \frac{1}{n} \sum_{i=1}^n \logp{\bbet^\trans \bE_i} - \frac{1}{n} \sum_{i=1}^n
        \ell_i \geq \eps \right ) \\
        &\quad\leq \PP_\Q \left ( \EE_\Q[\ell_{1, \Q}] - \frac{1}{n}
        \sum_{i=1}^n \ell_{i, \Q} + \frac{R_n}{n} \geq \eps \right ),
    \end{align}
    where in the inequality, we used both the regret bound and the trivial bound\\
    $\max_{\bbet \in \simplex}\sum_{i=1}^n \logp{\bbet^\trans \bE_i} \geq \sum_{i=1}^n
    \ell_{i, \Q}$. Setting $\eps \coloneqq \sqrt{8b \log(1/\alpha) / n} + 4
    \log(1/\alpha) / n + R_n / n$, we have
    \begin{align}
        &\PP_\Q \left ( \EE_\Q[\ell_{1, \Q}] - \frac{1}{n} \sum_{i=1}^n
        \ell_i \geq \sqrt{\frac{8b\log(1/\alpha)}{n}} +
        \frac{4\log(1/\alpha)}{n} + \frac{R_n}{n} \right ) \\
        &\quad\leq \PP_\Q \left ( \EE_\Q[\ell_{1,\Q}] - \frac{1}{n}
        \sum_{i=1}^n \ell_{i, \Q} \geq \sqrt{\frac{8b\log(1/\alpha)}{n}} +
        \frac{4\log(1/\alpha)}{n} \right ) \leq \alpha,
    \end{align}
    which completes the proof.
\end{proof}

\begin{lemma}\label{lemma:exponential concentration for bandit sequence - centered with independent copies}
  Let $\bE_1, \dots, \bE_n$ be $(d+1)$-vectors of $\Pe$-values satisfying \cref{assumption:class of eprocesses}. For any $\eps > 0$,
  \begin{equation}
  \PP_\Q \Paren{\frac{1}{n}\sum_{i=1}^n \Paren{\ell_{i,\Q}(A_i) - \EE_\Q [\ell_{i,\Q}(A_i)]} \geq \eps} \leq Kn \left ( \exp \left \{ - \frac{\eps^2n }{8bK^2} \right \} + \exp \left \{ -\frac{\eps n }{4K} \right \} \right ).
  \end{equation}
  Similarly, we have that for any $\epsilon > 0$,
  \begin{equation}
     \PP _{\Q}\Paren{\frac{1}{n}\sum_{i=1}^n \Paren{\E_{\Q}\Brac{\ell_{i, \Q}(A_i)} - \ell_{i, \Q}(A_i)} \geq \epsilon} \leq Kn \Paren{\exps{-\frac{\epsilon^2 n}{8bK^2}} + \exps{- \frac{\epsilon n}{4K}}}.
  \end{equation}
\end{lemma}
\begin{proof}
  Fix $\eps > 0$. For each $a \in \Acal$, define $\Ncal_a(n) := \{ i \in [n] : A_i = a \}$ to be the set of indices for which arm $a$ was pulled.
  First, note that $\sum_{i=1}^n \Paren{\ell_{i,\Q}(A_i) - \EE_\Q \Brac{\ell_{i,\Q}(A_i)}}$ can be re-written as
  \begin{equation}
    \sum_{i=1}^n \Paren{\ell_{i,\Q}(A_i) - \EE_\Q \Brac{\ell_{i,\Q}(A_i)}} = \sum_{a \in \Acal} \sum_{j \in \Ncal_a(n)} \Paren{\ell_{j,\Q}(a) - \EE_\Q[\ell_{1,\Q}(a)]}.
  \end{equation}
  For each $a \in \Acal$, let $\widetilde \ell_{1,\Q}(a), \dots, \widetilde \ell_{n, \Q}(a)$ be independent copies of $\ell_{1,\Q} (a)$.  Since $\sum_{j \in \Ncal_a(n)} \left (\ell_{\jQ}(a) - \EE_\Q[\ell_{\jQ}(a)] \right )$ has the same marginal distribution as $\sum_{j=1}^{N_a(n)} \left (\widetilde \ell_{\jQ}(a) - \EE_\Q[\widetilde \ell_{1,Q}(a)] \right )$,
  it holds that
  \begin{align}
    \PP_\Q \Paren{\frac{1}{n}\sum_{i=1}^n \Paren{\ell_{i,\Q}(A_i) - \EE_\Q [\ell_{i,\Q}(A_i)]} \geq \eps} &\leq \sum_{a \in \Acal} \PP_\Q \Paren{ \sum_{j \in \Ncal_a(n)} \Paren{\ell_{j,\Q}(a) - \EE_\Q [\ell_{j,\Q}(a)]} \geq \frac{\eps}{K}}\\
&= \sum_{a\in \Acal}\PP_\Q \Paren{\frac{1}{n}\sum_{j=1}^{N_a(n)} \Paren{\widetilde \ell_{j, \Q}(a) - \EE_\Q [\widetilde \ell_{1,\Q}(a)]} \geq \frac{\eps}{K}}\\
&\leq \sum_{a\in \Acal} \sum_{s=1}^n \PP_\Q \Paren{\frac{1}{n}\sum_{j=1}^{s} \Paren{\widetilde \ell_{j, \Q}(a) - \EE_\Q [\widetilde \ell_{1,\Q}(a)]} \geq \frac{\eps}{K}}
\end{align}
Now, using sub-exponentiality of the log-wealth increments (\cref{lemma:numeraire portfolio log wealths are sub exponential}), we have that for any $s \in [n]$ and $\theta \in [0,1/2]$,
\begin{align}
  \PP_\Q \left ( \frac{1}{n} \sum_{j=1}^s \Paren{\widetilde \ell_{j,\Q}(a) - \EE_\Q [\widetilde \ell_{1,\Q}(a)]} \geq \frac{\eps }{K} \right ) &= \PP_\Q \left ( \exp \left \{ \theta \sum_{j=1}^s \Paren{\widetilde \ell_{j,\Q}(a) - \EE_\Q [\widetilde \ell_{1,\Q}(a)]} \right \} \geq \exp \left \{ \frac{ \theta n\eps }{K} \right \} \right ) \\
    &\leq \exp \left \{ - \frac{\theta n \eps}{K} \right \} \prod_{j=1}^s \EE_\Q
    \left [ \exp \left \{ \theta \left ( \widetilde \ell_{j,\Q}(a) - \EE_\Q
    [\widetilde \ell_{1,\Q}(a)] \right ) \right \} \right ] \\
    &\leq \exps{2 b s \theta^2 - \frac{\theta n \epsilon}{K}} \\
  &\leq \exp \left \{ 2 b n \theta^2 - \frac{\theta n \eps }{K} \right \}.
\end{align}
Taking $\theta = \frac{\eps}{4bK} \land \frac{1}{2}$, we have that
\begin{equation}
  \exp \left \{ 2bn\theta^2 - \frac{\theta n \eps }{K} \right \} \leq \exp \left \{ - \frac{\eps^2n }{8bK^2} \right \} + \exp \left \{ -\frac{\eps n }{4K} \right \}.
\end{equation}
Putting all of the previous inequalities together, we obtain
\begin{equation}
  \PP_\Q \Paren{\frac{1}{n}\sum_{i=1}^n \Paren{\ell_{i,\Q}(A_i) - \EE_\Q [\ell_{i,\Q}(A_i)]} \geq \eps} \leq Kn \left ( \exp \left \{ - \frac{\eps^2n }{8bK^2} \right \} + \exp \left \{ -\frac{\eps n }{4K} \right \} \right ),
\end{equation}
which completes the proof of the first statement. The second follows similarly but with $\theta = - \left ( \frac{\eps}{4bK} \land \frac{1}{2} \right )$.
\end{proof}

\subsection{A variant of Doob's optional stopping theorem for supermartingales with bounded expected conditional increments}
\label{proof:optional_stopping-supermartingale}
The proof of \cref{proposition:stopping-timelower-bound} relies on an analogue of the optional stopping theorem due to Doob. However, the form of that theorem appearing in many standard textbooks~\citep{williams1991probability,durrett2019probability,grimmett2020probability} relies on certain assumptions that do not hold in the setting we consider, such as almost sure boundedness or uniformly bounded stopping times. As such, we require a variant of Doob's optional stopping theorem for supermartingales with conditionally bounded expected increments and with respect to potentially infinite stopping times. The result (\cref{lem:optional-stopping-supermartingale}) and its proof are routine but we provide them in full because we were not able to find them stated or proven in any standard probability reference.

In \cref{lem:bounded-expected-absolute-log-difference}), we demonstrate that  on the
boundedness of the expected absolute difference between a log-increment made up
of an arbitrary predictable portfolio and a log-increment with oracle access to
the log-optimal portfolio. These results are invoked in the proof of 
\cref{proposition:stopping-timelower-bound}.

\begin{lemma}
\label{lem:optional-stopping-supermartingale}
    Let $(\Omega, \Fcal, \P)$ be a filtered probability space with filtration $\Fcal$. Suppose that $(X_n)_{n \in \NN \cup \{0 \}}$ is a $\P$-supermartingale with respect to $\Fcal$ and that there exists $B > 0$ so that $\E_\P\Brac{\Abs{X_{n+1} - X_n} 
    \smvert \cF_n} \leq B$ $\P$-almost surely for every $n \in \NN \cup \{0 \}$. Furthermore, let $\tau$ be a
    stopping time such that $\E_{\P}\Brac{\tau} < \infty$. We then have that
    $X_\tau$ is integrable and that $\E_\P\Brac{X_\tau} \leq \E_\P\Brac{X_0}$.
\end{lemma}
\begin{proof}[Proof of \cref{lem:optional-stopping-supermartingale}]
    We begin by writing $X_{\tau \land n}$ as a telescoping sum
    \begin{equation}
        X_{\tau \land n} = X_0 + \sum_{m=0}^{\tau \land (n-1)} (X_{m+1} - X_m) = X_0 + \sum_{m=0}^{n-1} (X_{m + 1} - X_m)\1 \{ \tau > m \}.
    \end{equation}
    Applying the triangle inequality, we have that $|X_{\tau \land n}|$ (and hence $X_{\tau \land n}$) is upper bounded by
    \begin{equation}
        X_{\tau \land n} \leq |X_0| + \sum_{m=0}^{n-1} |X_{m+1} - X_m| \1 \{ \tau > m \} \leq |X_0| + \sum_{m=0}^{\infty} |X_{m+1} - X_m| \1 \{ \tau > m \}.
    \end{equation}
    Call the right-most expression $Y := |X_0| + \sum_{m=0}^\infty |X_{m+1} - X_m| \1 \{\tau > m\}$.
    We now show that $Y$ is $\P$-integrable, after which we will apply the dominated convergence theorem with $n$ tending to $\infty$. Indeed, by linearity of expectations and an application of Tonelli's theorem,
    \begin{equation}
        \EE_\P \left [ Y \right ] = \EE_\P[|X_0|] + \sum_{m=0}^\infty \EE_\P \left [|X_{m+1} - X_m| \1 \{ \tau > m \} \right ].\label{eq:proof-optional-stopping-expectation-upper-bound}
    \end{equation}
    Notice that $\Set{\tau > m} \in \cF_m$. By the law of total expectation
    and from the assumption that $\E_\P\Brac{\Abs{X_{m+1} - X_m} \smvert
    \cF_{m}} \leq B$, we have that
    \begin{align}
        \E_\P \Brac{\Abs{X_{m + 1} - X_m} \Ind{\tau > m}} 
        &= \E_\P \Brac{\Ind{\tau > m} \E\Brac{\Abs{X_{m + 1} - X_m} \smvert
        \cF_m}} \\
        &\leq \E_\P\Brac{\Ind{\tau > m} B} \\
        &= B \PP_\P\Paren{\tau > m} \mper
    \end{align}
    Applying this upper bound to 
    \eqref{eq:proof-optional-stopping-expectation-upper-bound} we can see how
    \begin{align}
        \E_\P\Brac{Y} &\leq \E_\P\Brac{\Abs{X_0}} + B
        \sum_{m=0}^\infty \PP_\P(\tau > m) = \E_{\P}\Brac{\Abs{X_0}} + B\E_\P\Brac{\tau} < \infty \mcom
    \end{align}
    where the finiteness of the upper bound follows from the integrability of
    $X_0$ and our assumption that $\E_\P\Brac{\tau} < \infty$. Since $X_{\tau \land n} \leq Y$ with $Y$ integrable, we apply the dominated convergence theorem to obtain
    \begin{equation}
        \lim_{n \to \infty} \E_{\P}\Brac{X_{\tau \wedge n}} =
        \E_{\P}\Brac{X_\tau} \mper
    \end{equation}
    Since $\E_\P\Brac{X_{\tau \wedge n}} \leq
    \E_\P\Brac{X_0}$ for all $n \in \NN$, we conclude that $\E_{\P}\Brac{X_\tau} \leq
    \E_{\P}\Brac{X_0}$, completing the proof.
\end{proof}

\begin{lemma}
\label{lem:bounded-expected-absolute-log-difference}
    Let $\infseqn{\bE_n(a)}$ be $(d+1)$-vectors of $\cP$-e-values that satisfy
    \cref{assumption:class of eprocesses} with the constant $b > 1$ for each arm
    $a \in \cA$. Furthermore, define the log-optimal portfolio under arm $a \in
    \cA$ as
    $\optbbet(a) \coloneqq \argmax_{\bbet \in \simplex}
    \E_\Q\Brac{\logp{\bbet^\trans \bE_1}}$, and assume $A_n$ is $\cF_{n-1}$-measurable. Then, the following inequality
    holds
    \begin{equation}
        \E_\Q\Brac{\Abs{\logp{\optbbet(A_n)^\trans \bE_n(A_n)} -
        \logp{\optbbet(\optarm)^\trans \bE_n(\optarm)}} \smvert \cF_{n-1}} \leq \logp{b} \mper
    \end{equation}
    Moreover, the following inequality also holds
    \begin{equation}
        \E_\Q\Brac{\Abs{\logp{\optbbet(A_n)^\trans \bE_n(A_n)} -
        \logp{\optbbet(\optarm)^\trans \bE_n(\optarm)}}} \leq \logp{b} \mper
    \end{equation}
\end{lemma}
\begin{proof}[Proof of \cref{lem:bounded-expected-absolute-log-difference}]
    We first recall \cref{assumption:class of eprocesses} which states that 
    for any arm $a \in \cA$ the multiplicative increments are almost surely
    upper bounded, $\sup_{\bbet \in \simplex} \Set{\bbet^\trans \bE_n(a)} \leq
    b$, and that there exists a
    $\tbbet \in \simplex$ such that under all $\Q \in \cQ$ and for all $n \in \N$ 
    the following holds: 
    $\tbbet^\trans \bE_n(a) = 1$ almost surely. Notice how this implies that
    $\logp{\tbbet^\trans \bE_n(a)} = 0$, which in turn allows us to conclude
    that $\E_\Q\Brac{\logp{\tbbet^\trans \bE_n(a)} \smvert \cF_{n-1}} = 0$. By definition of 
    $\optbbet(a)$, for any $a \in \cA$, and independence of the vectors of $\Pe$-values  we have the following inequality
    \begin{equation}
        \E_\Q\Brac{\logp{\optbbet(a)^\trans \bE_n(a)} \smvert \cF_{n-1}} \geq
        \E_\Q\Brac{\logp{\tbbet^\trans \bE_n(a)} \smvert \cF_{n-1}} = 0 \mper
    \end{equation}

    With this inequality in hand, we upper bound the expectation
    $\E_\Q[\logp{\optbbet(A_n)^\trans \bE_n(A_n)} -
    \logp{\optbbet(\optarm)^\trans \bE_n(\optarm)} ~|~ \cF_{n-1}]$ as follows:
    \begin{align}
        \E_\Q\Brac{\logp{\optbbet(A_n)^\trans \bE_n(A_n)} -
        \logp{\optbbet(\optarm)^\trans \bE_n(\optarm)} \smvert \cF_{n-1}} 
        &\leq \logp{b} - \E_\Q\Brac{\logp{\tbbet^\trans
        \bE_n(\optarm)} \smvert \cF_{n-1}}\\
        &= \logp{b} \mper
    \end{align}
    We now proceed to lower bound the expectation as:
    \begin{align}
        \E_\Q\Brac{\logp{\optbbet(A_n)^\trans \bE_n(A_n)} -
        \logp{\optbbet(\optarm)^\trans \bE_n(\optarm)} \smvert \cF_{n-1}} 
        &\geq \E_\Q\Brac{\logp{\tbbet^\trans \bE_n(A_n)} \smvert \cF_{n-1}} - \logp{b} \\
        &\geq -\logp{b} \mper
    \end{align}
    Hence, putting the above steps together we can conclude that 
    \begin{equation}
        \E_\Q\Brac{\Abs{\logp{\optbbet(A_n)^\trans \bE_n(A_n)} -
        \logp{\optbbet(\optarm)^\trans \bE_n(\optarm)}} \smvert \cF_{n-1}} \leq \logp{b} \mper
    \label{eq:proof-bounded-conditional-expected-log-increments}
    \end{equation}
    Lastly, the law of iterated expectations in conjunction with \eqref{eq:proof-bounded-conditional-expected-log-increments} give us that 
    \begin{align}
        &\E_{\Q}\Brac{\Abs{\logp{\optbbet(A_n)^\trans \bE_n(A_n)} - \logp{\optbbet(\optarm)^\trans \bE_n(\optarm)}}}\\
        &= \E_{\Q}\Brac{\E\Brac{\Abs{\logp{\optbbet(A_n)^\trans \bE_n(A_n)} - \logp{\optbbet(\optarm)^\trans \bE_n(\optarm)}} \smvert \cF_{n-1}}} \\
        &\leq \E_{\Q}\Brac{\logp{b}}\\ 
        &= \logp{b} \mcom
    \end{align}
    completing the proof.
\end{proof}

\section{Familiar Sequential Testing Problems Satisfying \texorpdfstring{\cref{assumption:class of eprocesses}}{Assumption~\ref{assumption:class of eprocesses}}}

In \cref{section:prelim-test supermartingales and e-processes} we introduced test $\cP$-supermartingales under the multi-armed data collection protocol, and we showed how the two-sided bounded mean testing problem can be mapped into this multi-armed setting. In this section we show how two other testing problems---that have been of interest in the literature---can be instantiated under the multi-armed data collection protocol, and we show how these testing problems satisfy \cref{assumption:class of eprocesses}.

\label{section:multi-armed-testing-examples}
\begin{example}[One-sided bounded mean testing]\label{example:one-sided mean testing}
    Suppose that $\infseqn{(Y_n(1), \dots, Y_n(K))}$ are sampled $\iid$ and are
    supported on $[0, 1]^K$. The statistician is 
    interested in testing the following one-sided global null $\cP^{\leq}$ versus alternative
    $\cQ^{>}$: 
    \begin{equation}
        \cP^{\leq} = \Set{\P \svert \forall a \in \cA,~ \E_{\P}\Brac{Y(a)} \leq \mu_0}
        \quad\mathrm{versus}\quad \cQ^{>} = \Set{\P \svert \exists a \in
        \cA,~ \E_{\P}\Brac{Y(a)} > \mu_0} \mcom
    \end{equation}
    for some $\mu_0 \in [0,1]$. In each round $n \in \N$ the statistician selects which arm
    $A_n$ to pull (i.e., which random variable to observe) in such a way that $\infseqn{A_n}$
    is an $\cH$-predictable sequence. The statistician constructs the 
    following test $\cP^{\leq}$-supermartingale:
    \begin{equation}
        W_n^{\leq} \coloneqq \prod_{i=1}^n \Brac{1 + \bet_i \Paren{\frac{
        Y_i(A_i)}{\mu_0} - 1}} \mcom
    \end{equation}
    where $\infseqn{\bet_n}$ is any $[0,1]$-valued predictable sequence. Under
    this setting the bound on the multiplicative increments is $b = 1/\mu_0$ and
    $\tilde{\bet} = 0$.
\end{example}

\begin{example}[Testing equality of bounded tuples]\label{example:bounded tuple equality testing}
    Suppose that $\infseqn{((X_n(1), Y_n(1)), \dots, (X_n(K), Y_n(K)))}$ is a
    sequence of $\iid$ tuples and that all the random variables are supported in
    $[0,1]$. Define $D_n(a) \coloneqq X_n(a) - Y_n(a)$ for each $a \in \cA$ and
    $n \in \N$. The equality global null $\cP^{(D=)}$ versus alternative
    $\cQ^{(D\neq)}$ are given by:
    \begin{equation}
        \cP^{(D=)} \coloneqq \Set{\P \svert \forall a \in \cA,~
        \E_{\P}\Brac{D(a)} = 0}
        \quad\mathrm{versus}\quad \cQ^{(D \neq)} \coloneqq \Set{\P \svert \exists
            a \in \cA,~ \E_{\P}\Brac{D(a)} \neq 0} \mper
    \end{equation}
    After applying the transformation
    $Z_n(A_n) \coloneqq (D_n(A_n) + 1) / 2 \in [0,1]$ on the observed random
    tuples, the statistician constructs the following test
    $\cP^{(D=)}$-supermartingale:
    \begin{equation}
        \prod_{i=1}^n \Brac{(1 - \bet_i) \frac{(1 - Z_i(A_i))}{1/2} + \bet_i
        \frac{Z_i(A_i)}{1/2}} \mcom
    \end{equation}
    where $\infseqn{A_n}$ is a predictable sequence corresponding to the arm
    pulls and $\infseqn{\bet_n}$ is a $[0,1]$-valued predictable sequence. In
    this scenario we can see how $\tilde{\lambda} = 1/2$ and that the bound on
    the multiplicative increments is $b = 2$. Even though this problem can be reduced to a special case of the bounded mean testing problem, we still highlight it as the literature contains results on asymptotic growth rates and bounds on expected rejection times~\citep{chugg2023auditing,chen2025optimistic,shekhar2023nonparametric,podkopaev2023sequentialKernelized,podkopaev2023sequentialTwoSample}.
\end{example}

\end{document}